\documentclass[sigconf]{acmart}
\AtBeginDocument{%
  }

\copyrightyear{2026}
\acmYear{2026}
\setcopyright{cc}
\setcctype{by}
\acmConference[CCS '26]{Proceedings of the 2026 ACM SIGSAC Conference on Computer and Communications Security}{November 15--19, 2026}{The Hague, Netherlands}
\acmBooktitle{Proceedings of the 2026 ACM SIGSAC Conference on Computer and Communications Security (CCS '26), November 15--19, 2026, The Hague, Netherlands}
\acmDOI{10.1145/3830454.3832667}
\acmISBN{979-8-4007-2871-6/2026/11}

\usepackage{multirow}
\usepackage{graphicx}
\usepackage{crypto/crypto}

\usepackage{xcolor}

\usepackage{makecell}
\usepackage[table]{xcolor}

\usepackage{tikz}
\usepackage{amssymb}
\usetikzlibrary{decorations.pathreplacing}
\usetikzlibrary{patterns}
\usepackage{amsfonts}
\usetikzlibrary{arrows.meta, positioning, calc, backgrounds}
\usetikzlibrary{shapes.geometric}
\usepackage{subcaption}

\usepackage{marvosym}

\newcommand{\meng}[1]{{\color{red} [meng: #1]}}

\newcommand{\blue}[1]{{\color{blue}  #1}}
\newcommand{\red}[1]{{\color{red}  #1}}

\begin{document}


% \title{Towards Scalable Fuzzy Private Set Intersection}
\title{Towards Scalable Fuzzy PSI via Efficient Fuzzy Matching}


\author{Meng Hao}
\affiliation{%
  \institution{Singapore Management University}
  \country{Singapore}
}
\email{menghao303@gmail.com}

\author{Xinpeng Yang}
\affiliation{%
  \institution{Nanyang Technological University}
  \country{Singapore}
}
\email{XINPENG004@e.ntu.edu.sg}

\author{Hanxiao Chen}
\authornote{Corresponding authors.}
\affiliation{%
  \institution{University of Electronic Science and Technology of China}
  \city{Chengdu}
  \country{China}
}
\email{hanxiao.chen@uestc.edu.cn}

\author{Tianwei Zhang}
\affiliation{%
  \institution{Nanyang Technological University}
  \country{Singapore}
}
\email{tianwei.zhang@ntu.edu.sg}

\author{Haiyang Xue}
\affiliation{%
  \institution{Singapore Management University}
  \country{Singapore}
}
\email{haiyangxue@smu.edu.sg}

\author{Guomin Yang}
\affiliation{%
  \institution{Singapore Management University}
  \country{Singapore}
}
\email{gmyang@smu.edu.sg}

\author{Hongwei Li}
\authornotemark[1]
\affiliation{%
  \institution{University of Electronic Science and Technology of China}
  \city{Chengdu}
  \country{China}
}
\email{hongweili@uestc.edu.cn}

\author{Robert H. Deng}
\affiliation{%
  \institution{Singapore Management University}
  \country{Singapore}
}
\email{robertdeng@smu.edu.sg}

%% By default, the full list of authors will be used in the page headers.
%% Often, this list is too long, and will overlap other information printed in the page headers.
%% This command allows you to define a more concise list of authors' names for this purpose.
\renewcommand{\shortauthors}{Meng Hao et al.}


\begin{abstract}

In fuzzy private set intersection (fuzzy PSI), there are two parties, a sender holding a set of $d$-dimensional points $Q = \{\vecq_1, \ldots, \vecq_m\}$ and a receiver holding a set $W = \{\vecw_1, \ldots, \vecw_n\}$ of the same structure. 
It enables the receiver to learn the point $\vecq \in Q$ for which there exists some $\vecw \in W$ satisfying $\mathsf{dist}(\vecq, \vecw) \le \delta$ under a given distance metric.
Although several fuzzy PSI protocols for $L_{p\in[1, \infty]}$ distance are proposed, there are significant efficiency issues, mainly because they (1) heavily rely on expensive cryptographic primitives, e.g., homomorphic encryption or garble circuits, and/or (2) incur undesirable asymptotic communication and computation complexity.

In this paper, we present scalable fuzzy PSI protocols for general $L_{p \in [1, \infty]}$ distance, supporting both low- and high-dimensional sets. 
The core technique is two efficient fuzzy matching protocols that securely evaluate $\mathsf{dist}(\vecq, \vecw) \le \delta$.
The first is built from a role-reversed oblivious PRF (OPRF) and realizes $O(d\log \delta)$ overhead, compared to $O((\log \delta)^d)$ in previous works. The second leverages customized oblivious transfer (OT) with $O(d\ell)$ overhead, where $\ell$ is the bit length of inputs, which is particularly suitable for short inputs.
With these new techniques, we further propose a new dual-layer hashing framework for fuzzy PSI over low-dimensional sets, instantiated with our OT-based fuzzy matching and enhanced with a domain reduction optimization.
The protocols achieve an overhead linear with $n, m, \log \delta, 2^d$, without the $O((\log \delta)^d)$ or $O(\delta)$ factors present in prior works.
{For high-dimensional sets, we construct fuzzy PSI protocols based on our OPRF- and OT-based fuzzy matching, which achieve an asymptotic overhead linear with $n, m, d$, and $\log \delta$ but rely on the strong globally disjoint assumption.}
% an exponential
% improvement on $\delta$ over the state-of-the-art.

% our protocols outperform previous works. For distance thresholds $\delta$ ranging from 16 to 1024,
Extensive evaluations demonstrate that our protocols achieve up to a $145\times$ speedup in running time and a $20\times$ reduction in communication cost compared to van Baarsen and Pu~(ASIACRYPT'25), and achieve up to a $25\times$ speedup in running time and up to a $17\times$ reduction in communication cost compared to Piske et al.~(CCS'25).

% In this paper, we propose a new framework for \textit{scalable} fuzzy PSI for general $L_{p \in [1, \infty]}$ distance.
% The core technique is two efficient fuzzy matching protocols to securely evaluate $\mathsf{dist}(\vecq, \vecw) \le \delta$ with single-point input $\vecq$ and $\vecw$ from both parties, respectively.
% The first is built from oblivious PRF and realizes $O(d\log \delta)$ overhead, compared to $O(d(\log \delta)^d)$ in previous works. The second utilizes customized oblivious transfer with $O(d\ell)$ overhead, where $\ell$ is the bit length of inputs, which is particularly suitable for short inputs.
% With these new techniques, we first propose scalable fuzzy PSI protocols for low-dimensional sets utilizing spatial hashing and OT-based fuzzy matching, along with domain reduction optimization.
% The protocols achieve an overhead linear with $n, m, \log \delta, 2^d$ for the first time, without factors $O((\log \delta)^d)$ or $O(\delta)$ in existing works.
% Moreover, relying on a stronger set assumption used in recent work, we propose fuzzy PSI protocols for high-dimensional sets based on OPRF-based fuzzy matching. It achieves an overhead linear with $n, m, \log \delta, d$, instead of linear dependence on $\delta$.

\end{abstract}




%
% Keywords. The author(s) should pick words that accurately describe
% the work being presented. Separate the keywords with commas.
% A "teaser" image appears between the author and affiliation
% information and the body of the document, and typically spans the
% page.


\maketitle

\section{Introduction}

% Ga22 FSS 2delta 
% Ga24 iFSS mini-universe prefix 
% Ga25 dFSS mini-universe prefix 

% BP24 AHE Lp 2delta 
% Dang25 AHE Lp 2delta prefix

% Mike24 GC Lp disj hash
% Ni25 OPRF Lp disj hash

% BP25 OPRF 2delta/4delta, prefix

% BP24 AHE Lp local disj
% Gao24 AHE Lp local disj
% Zhang25 non-black box AM-PRF Lp s-seperate 
% Dang25 AHE Lp local disj prefix

% Ga22 FSS global disj
% BP25 OPRF global disj 

Private set intersection (PSI), including \cite{freedman2004efficient, dong2013private, pinkas2015phasing, kolesnikov2016efficient, kolesnikov2017practical, pinkas2018scalable, pinkas2019efficient, pinkas2020psi, garimella2021oblivious, rindal2021vole, raghuraman2022blazing} and references therein, has been extensively studied over the past decades.
It enables a sender \SSS holding a set $Q$ and a receiver \RRR holding a set $W$ to privately compute the intersection $Q \cap W$, while revealing no additional information beyond the output.
Traditional PSI protocols focus on \emph{exact} item matching and have been successfully deployed in applications such as private contact discovery \cite{kales2019mobile, wu2023efficient} and password breach monitoring~\cite{google-psi-password}.
However, these protocols are inherently limited in scenarios that require \emph{approximate} matching, such as location-based services and biometric identification.

To address this gap, a lot of fuzzy PSI protocols \cite{garimella2022structure, chakraborti2023distance, richardson2024fuzzy, garimella2024computation, van2024fuzzy, gao2024efficient, dang2025ccs, zhang2025fast, piske2025distance, van2025, bui2025new,blass25} have been proposed.
In fuzzy  PSI, the sender holds a set of $d$-dimensional points $Q = \{\vecq_1, \ldots, \vecq_m\}$ and the receiver holds $W = \{\vecw_1, \ldots, \vecw_n\}$ of the same structure. 
It enables the receiver to learn the point $\vecq \in Q$ for which there exists some $\vecw \in W$ satisfying $\mathsf{dist}(\vecq, \vecw) \le \delta$ under a given distance metric $\mathsf{dist}$, e.g., Chebyshev distance $L_\infty$ or Minkowski distance $L_P$.
Fuzzy PSI is particularly useful in scenarios where inputs are inherently noisy.
For example, in biometric identification, users are matched against biometric
samples such as iris scans, fingerprints, or facial patterns. Similarly, in ride-sharing applications, passengers determine whether there are available vehicles within a certain proximity.

% A naive approach for fuzzy PSI is to compare pair-wise elements from two sets.
A naive approach to fuzzy PSI is to perform pairwise comparisons between elements of the two sets.
However, it incurs a quadratic overhead of $O(mn)$ in the set sizes, making it impractical for large-scale sets.
To avoid the efficiency issue, existing fuzzy PSI protocols, as illustrated by Gao et al. \cite{gao2024efficient}, follow a common paradigm with two phases: coarse mapping and refined filtering. (1) In the coarse mapping phase, initial pairings are established between a sender point $\vecq$ and a receiver point $\vecw$ if they are potentially $\delta$-close, through some \textit{fuzzy mapping} scheme.
% This is achieved with \textit{fuzzy mapping} schemes, which associate a unique identifier (or a set of identifiers) for each point.
However, this phase may introduce false positives, that is, two points with a distance greater than $\delta$ can still be paired.
(2) In the refined filtering phase, a \textit{fuzzy matching} approach is applied to each pair formed during the coarse mapping phase, which determines whether the receiver point $\vecw$ and the sender point $\vecq$ indeed satisfy $\mathsf{dist}(\vecq,\vecw) \le \delta$. 
This eliminates the false positives and yields the final results to the receiver.

% \xp{For a long time, research on fuzzy PSI primarily focused on the Hamming distance. The first protocol was formalized by Freedman et al.~\cite{freedman2004efficient}, who proposed a construction based on additively homomorphic encryption (AHE). Building on this foundational work, a series of subsequent studies~\cite{freedman2004efficient, chmielewski2008fuzzy,ye2009efficient,indyk2006polylogarithmic,uzun2021fuzzy,chakraborti2023distance,blass25} progressively improved computational and communication efficiency under various techniques and cryptographic assumptions. Since the Hamming distance is not the main point of this work, we do not discuss these protocols in detail. Instead, we hereafter concentrate on fuzzy PSI protocols for the general Minkowski distance $L_{p \in [1,\infty]}$.
% }  

\begin{table*}[!t]
\centering
\caption{Comparisons of existing fuzzy PSI protocols for $L_{p \in [1, \infty]}$ distance, where the sender holds $m$ elements and the receiver holds $n$ elements in a $d$-dimensional space, and $\delta$ is the distance threshold.
The protocols of \cite{van2025, piske2025distance} and ours for low-dimensional sets essentially rely on the same assumption, as shown in Appendix \ref{appendix: Additional Details of Spatial Hashing}. For clarity, we unify it as the unique ball/center. 
Note that both \cite{piske2025distance} and our protocols for low-dimensional sets can be extended to arbitrary receiver sets, but introduce an additional multiplicative factor in complexity.
}
% We ignore multiplicative factors of the computational security parameter $\kappa$ and statistical security parameter $\lambda$.
\label{tab:complexities}
\renewcommand{\arraystretch}{1.06}
\resizebox{\textwidth}{!}{%
\begin{tabular}{|c|c|c|c|c|c|cc|}
\hline
\multirow{2}{*}{Setting} & \multirow{2}{*}{Metric}     & \multirow{2}{*}{Protocol}            & \multirow{2}{*}{Technique} & \multirow{2}{*}{Assumption}                                             & \multirow{2}{*}{Communication}                                       & \multicolumn{2}{c|}{Computation}                                                                                         \\ \cline{7-8} 
                         &                             &                                      &                                                                                     &                  &                                                    & \multicolumn{1}{c|}{Sender}                                              & Receiver                                      \\ \hline
\multirow{15}{*}{\begin{tabular}[c]{@{}c@{}}Low\\ Dimension\end{tabular}}    & \multirow{8}{*}{$L_\infty$} & \cite{bui2025new}                    & FSS & $\mathcal{R},\text{mini-univ.}$                                                & $O(dn \log \delta + m2^d \textcolor{red}{(\log \delta)^d})$                            & \multicolumn{1}{c|}{$O(m2^d\textcolor{red}{(\log \delta)^d})$}                              & $O(\textcolor{red}{(\log \delta)^d} n + dm\log \delta)$        \\ \cline{3-8} 
                         &                             & \multirow{2}{*}{\cite{van2024fuzzy}} & AHE &$\mathcal{R}, \min > 2\delta$                                                  & $O(\textcolor{red}{\delta} d n + 2^d m)$                                              & \multicolumn{1}{c|}{$O(2^dd m)$}                                         & $O(\textcolor{red}{\delta} d n + 2^d m)$                       \\
                         \cline{4-8}
                         &                             &                                      & AHE &$\mathcal{R}, \min > 4\delta$                                                  & $O(\textcolor{red}{\delta}2^ dd n + m)$                                               & \multicolumn{1}{c|}{$O(dm)$}                                             & $O(\textcolor{red}{\delta}2^dd n + m)$                         \\ \cline{3-8} 
                         &                             & \cite{dang2025ccs}                   & AHE & $\mathcal{R},\min>4\delta$                                                     & $O(d \log \delta (n 2^d + m))$                                       & \multicolumn{1}{c|}{$O(2^d d n \log \delta + m \textcolor{red}{(\log \delta)^d})$}        & $O(d m \log\delta)$                           \\ \cline{3-8} 
                         &                             & \cite{richardson2024fuzzy}           & GC &$\mathcal{S}, \text{disj. hash}$                                                & $O(d \log \delta (n 2^s + m 2^{d-s}))$                               & \multicolumn{1}{c|}{$O(2^{d-s} dm \log \delta)$}                         & $O(2^{s} dn \log \delta)$                     \\ \cline{3-8} 
                         &                             & \cite{piske2025distance}             & OPRF \& OT & $\mathcal{R} \wedge \mathcal{S}$, unique ball/center                                   & $O(d (\textcolor{red}{\delta} m+2^dn))$                                               & \multicolumn{1}{c|}{$O(\textcolor{red}{\delta} d m)$}                                     & $O(d 2^dn + m)$                               \\ \cline{3-8}  
                         &                             & \cite{van2025}                       & OPRF & $\mathcal{R} \wedge \mathcal{S}$, unique ball/center                     & $O(d\log\delta (m+2^dn))$                                            & \multicolumn{1}{c|}{$O(d (m+2^dn)\log\delta)$}                           & $O(2^dn(d\log\delta+\textcolor{red}{(\log\delta)^{d/2}}))$     \\ \cline{3-8} 
                         &                             & Ours                                 & OPRF \& OT & $\mathcal{R} \wedge \mathcal{S}$, unique ball/center                & $O{(d(m\log \delta+2^dn))}$                                        & \multicolumn{1}{c|}{$O{(md\log\delta)}$}                              & $O{(d(m\log \delta+2^dn))}$                   \\ \cline{2-8} 
                         & \multirow{7}{*}{$L_{p}$}    & \cite{van2024fuzzy}                  & AHE & $\mathcal{R}, \min > 2\delta(d^{1/p}+1)$                                       & $O(\textcolor{red}{\delta}{2^d} d n + \textcolor{red}{\delta^p} m)$                                    & \multicolumn{1}{c|}{$O((d + \textcolor{red}{\delta^p}) m)$}                               & $O(\textcolor{red}{\delta}{2^d} d n + m)$                      \\ \cline{3-8}  
                         &                             & \cite{dang2025ccs}                   & AHE & $\mathcal{R},\min>2\delta(d^{1/p}+1)$                                          & $O(2^{d} p (d n \log \delta + m\textcolor{red}{(\log \delta)^{d}}) + d m\log \delta)$ & \multicolumn{1}{c|}{$O(2^{d} (d p n \log \delta + m\textcolor{red}{(\log \delta)^{d}}))$} & $O(m (d p \log \delta + 2^d\textcolor{red}{(\log \delta)^{d}}))$ \\ \cline{3-8} 
                         &                             & \multirow{2}{*}{\cite{richardson2024fuzzy}}           & GC & $\mathcal{S}, \text{disj. hash},p=1$                                                & $O(d \log(d\delta)(n 2^{d} + m 2^{d-s}))$                            & \multicolumn{1}{c|}{$O(d m 2^{d-s} \log(d\delta) )$}                     & $O(d n 2^{s} \log(d\delta))$                  \\ \cline{4-8} 
                         &                             &           & GC & $\mathcal{S}, \text{disj. hash},p=2$                                                & $O(d \log(d\delta)(n 2^{d} + m 2^{d-s})+\log^3(d\delta)m2^{d-s})$                            & \multicolumn{1}{c|}{$O(m 2^{d-s} ((d\log(d\delta)+\log^3(d\delta)))$}                     & $O(n 2^{s} (d\log(d\delta)+\log^3(d\delta)))$                  \\ \cline{3-8} 
                         &                             & \cite{piske2025distance}             & OPRF \& OT & $\mathcal{R} \wedge \mathcal{S}$, unique ball/center                                    & $O(d (\textcolor{red}{\delta} m+2^dn))$                                               & \multicolumn{1}{c|}{$O(\textcolor{red}{\delta} d m)$}                                     & $O(d 2^dn + m)$                               \\ \cline{3-8}  
                         &                             & \cite{van2025}                       &  OPRF & $\mathcal{R} \wedge \mathcal{S}$, unique ball/center    & $O(d (\textcolor{red}{\delta} m+2^dn)+p(m+2^dn)\log\delta)$                           & \multicolumn{1}{c|}{$O(d (\textcolor{red}{\delta} +2^dn)+p(m+2^dn)\log\delta)$}           & $O(2^dn(d+p\log\delta))$                      \\ \cline{3-8} 
                         &                             & Ours                                 & OPRF \& OT & $\mathcal{R} \wedge \mathcal{S}$, unique ball/center & $O{((mp\log\delta+2^dn)d)}$                                        & \multicolumn{1}{c|}{$O{(mpd\log\delta})$}                              & $O{((mp\log\delta+2^dn)d)}$                 \\ \hline
\multirow{14}{*}{\begin{tabular}[c]{@{}c@{}}High\\ Dimension\end{tabular}}   & \multirow{6}{*}{$L_\infty$} & \cite{van2024fuzzy}                  & AHE & $\mathcal{R}, \text{disj. proj.}$                                              & $O(\textcolor{red}{(\delta d)^2} n + m)$                                              & \multicolumn{1}{c|}{$O(d^2 m)$}                                          & $O(\textcolor{red}{(\delta d)^2} n + m)$                       \\ \cline{3-8}  
                         &                             & \cite{gao2024efficient}              & AHE  & $\mathcal{R} \wedge \mathcal{S}, \text{disj. proj.}$                           & $O(\textcolor{red}{\delta} d (m+n))$                                                  & \multicolumn{1}{c|}{$O(\textcolor{red}{\delta} d m + n)$}                                 & $O(\textcolor{red}{\delta} d n + m)$                           \\ \cline{3-8}  
                         &                             & \cite{dang2025ccs}                   & AHE & $\mathcal{R} \wedge \mathcal{S}, \text{disj. proj.}$                           & $O{ (d(m+n)\log\delta)}$                                           & \multicolumn{1}{c|}{$O{ (d(m+n)\log\delta)}$}                          & $O{ (d(m+n)\log\delta)}$                    \\ \cline{3-8} 
                         &                             & \cite{zhang2025fast}                 & so-OPRF & $\mathcal{R} \wedge \mathcal{S},s\text{-separate}$                            & $O(\textcolor{red}{\delta^s}\frac{d}{s}(m+n))$                                        & \multicolumn{1}{c|}{$O(\textcolor{red}{\delta^s}\frac{d}{s}m+n)$}                         & $O(\textcolor{red}{\delta^s}\frac{d}{s}n+m)$                   \\ \cline{3-8}  
                         &                             & \cite{van2025}                       & OPRF  & $\mathcal{R}\wedge \mathcal{S}, \text{global disj.}$                          & $O((m\textcolor{red}{\delta}+n)d)$                                                    & \multicolumn{1}{c|}{$O((m\textcolor{red}{\delta}+n)d)$}                                   & $O(nd)$                                       \\ \cline{3-8} 
                         &                             & Ours                                 & OPRF  & $\mathcal{R}\wedge \mathcal{S}, \text{global disj.}$                     & $O((m+n)d\log\delta)$                                                & \multicolumn{1}{c|}{$O((m+n)d\log\delta)$}                               & $O(nd\log\delta+md)$                          \\ \cline{2-8} 
                         & \multirow{5}{*}{$L_{p}$}    & \cite{gao2024efficient}              & AHE & $\mathcal{R} \wedge \mathcal{S}, \text{disj. proj.}$                           & $O(\textcolor{red}{\delta} d (m+n)+pm\log\delta)$                                     & \multicolumn{1}{c|}{$O(\textcolor{red}{\delta} d m+pm\log\delta + n)$}                    & $O(\textcolor{red}{\delta} d n+pm\log\delta)$                  \\ \cline{3-8} 
                         &                             & \cite{dang2025ccs}                   & AHE & $\mathcal{R} \wedge \mathcal{S}, \text{disj. proj.}$                           & $O{(dpn\log\delta+dm\log\delta)}$                                  & \multicolumn{1}{c|}{$O{(dpm\log\delta+dn\log\delta)}$}                 & $O{(dpn\log\delta+dm\log\delta)}$           \\ \cline{3-8} 
                         &                             & \cite{zhang2025fast}                 & so-OPRF & $\mathcal{R} \wedge \mathcal{S},s\text{-separate}$                            & $O(\textcolor{red}{\delta^s}\frac{d}{s}(m+n)+pm\log\delta)$                           & \multicolumn{1}{c|}{$O((\textcolor{red}{\delta^s}\frac{d}{s}+p\log\delta)m+n)$}           & $O(\textcolor{red}{\delta^s}\frac{d}{s}n+pm\log\delta)$        \\ \cline{3-8}  
                         &                             & \cite{van2025}                       & OPRF & $\mathcal{R}\wedge \mathcal{S}, \text{global disj.}$                          & $O((m\textcolor{red}{\delta}+n)d+(m+n)p\log\delta)$                                   & \multicolumn{1}{c|}{$O((m\textcolor{red}{\delta}+n)d+(m+n)p\log\delta)$}                  & $O(n(d+p\log\delta))$                         \\ \cline{3-8} 
                         &                             & Ours                                 & OPRF \& OT & $\mathcal{R}\wedge \mathcal{S}, \text{global disj.}$                     & $O((m+np)d\log\delta)$                                               & \multicolumn{1}{c|}{$O((m+np)d\log\delta)$}                              & $O(npd\log\delta+md)$                        \\ \hline
\end{tabular}%
}
% \vspace{3pt}
\captionsetup{justification=raggedright, singlelinecheck=false}
\caption*{\footnotesize 
% -- For protocol~\cite{dang2025ccs}, we revise the computational complexity of the sender.\footref{fn: revise}
-- $\mathcal{R}/\mathcal{S}$ denotes that the set of receiver/sender satisfies the assumption, and $\mathcal{R} \wedge \mathcal{S}$ indicates that both sets satisfy the assumption.

-- unique ball/center means in spatial hashing, each grid cell intersects with at most one ball and each grid cell contains at most one point, respectively.

% is short for the assumption that every point of the set can be mapped to a unique ball/center.

-- mini-univ is short for mini-universe assumption, where every ball can be mapped to a distinct vertex of the space.

-- $\min > l$ means that the minimum distance between any two elements of the set is greater than $l$. 

-- disj. hash is short for disjoint hash assumption, which means the spatial hashing scheme maps at most one point to every possible target grid cell.

-- disj. proj. is short for disjoint projection assumption, where there exists at least one dimension on which the projection of balls does not overlap.

% where each element has disjoint projections in some dimension with all other elements.

-- global disj. is short for globally disjoint projection assumption, where for every dimension, the projection of balls does not overlap.

% where the distance in any dimension of two points in the set is at least $2\delta$

--  $s$-separate means for every $s$ dimensions, there exists at least one dimension on which the projection of balls does not overlap.

% elements have disjoint projections in one of these $s$ dimensions with all other elements where $1\le s \le d$.

% -- $0 < \rho < 1/c$ is a parameter in locality-sensitive hashing if the receiver’s points are distance $>c\delta$ apart

% -- $0 < \rho < 1/c$ is a parameter in locality-sensitive hashing where the distance between any receiver's two elements is greater than $c\delta$.

% -- \xue{what is $p$?} \meng{Lp distance}

}
\end{table*}

{
Recently, a growing body of literature~\cite{garimella2022structure, garimella2024computation, bui2025new, van2024fuzzy, gao2024efficient, dang2025ccs, zhang2025fast, richardson2024fuzzy, piske2025distance, van2025} has explored fuzzy PSI under $L_p$ distance for $p \in [1, \infty]$. As detailed in Section~\ref{sec: Related Work}, we categorize these prior works into low-dimensional and high-dimensional regimes, according to different techniques and designs. In general, these works typically begin by designing 
\emph{single-point} fuzzy matching protocols, and then introduce fuzzy mapping
techniques to compile fuzzy matching into \emph{multi-point} fuzzy PSI. 
To achieve desirable asymptotic complexities, many recent constructions \cite{van2024fuzzy, gao2024efficient, dang2025ccs} rely on heavy cryptographic primitives, such as computation-intensive additively homomorphic encryption (AHE) or communication-intensive garbled circuits (GCs). Consequently, these protocols suffer from prohibitive concrete overhead. More recently, van Baarsen and Pu~\cite{van2025} propose the state-of-the-art fuzzy PSI scheme leveraging lightweight oblivious PRFs (OPRFs) to bypass these heavy primitives. Unfortunately, their protocols incur undesirable asymptotic complexities, scaling as $\mathcal{O}((\log \delta)^d)$ or $\mathcal{O}(\delta)$ for the distance threshold $\delta$ and the dimension $d$. This may introduce severe bottlenecks in real-world applications requiring large distance thresholds. We refer to Table~\ref{tab:complexities} for a comprehensive comparison of existing works.}

Motivated by these limitations, we pose the following question:

% \subsection{Limitations of Previous Works} 

\begin{center}
        % \begin{tcolorbox}[
        %     colback=lightgray!15, 
        %     size=title, 
        %     colframe=black, 
        %     boxrule=0.5pt,
        %     width=\linewidth,
        %     valign=center,
        %     ]
            \textit{Can we construct scalable fuzzy PSI protocols with both concrete and asymptotic efficiency guarantees?}
        % \end{tcolorbox}
\end{center}
% \vspace{0.5pt}

\subsection{Our Contribution}
We answer the above question affirmatively by proposing \textit{scalable} fuzzy PSI for general $L_{p \in [1, \infty]}$ distance over both low- and high-dimensional sets from lightweight oblivious transfer (OT) and oblivious PRF (OPRF).
% \footnote{We note that recently proposed concretely efficient OPRF protocols are based on OT or its generalized variant, i.e., vector oblivious linear evaluation (VOLE)}. 
{
For low-dimensional sets, our protocols achieve communication and computation complexities that scale linearly with the set sizes $m, n$ and logarithmically with the distance threshold $\delta$, along with an additional exponential factor on the dimension $d$.
% Our protocols for low-dimensional sets realize an asymptotic communication and computation complexity that is linear in the set sizes $m, n$ and logarithmic in the distance threshold $\delta$, 
This factor also appears in previous works for low-dimensional (small $d$) settings.
% ; however, it applies only to lightweight operations in our protocols.
For high-dimensional sets, by leveraging stronger assumptions, we achieve an overhead that scales linearly with the dimension $d$.  
We summarize our core contributions as follows.}

% \subsection{Our Contribution}
% We answer the above question affirmatively by proposing \textit{scalable} Fuzzy PSI for the $L_p$ distance metric (for $p \in [1, \infty]$), applicable to both low- and high-dimensional sets, built upon lightweight Oblivious Transfer (OT) and Oblivious PRF (OPRF) primitives. 
% \blue{For low-dimensional sets, our protocols achieve communication and computation complexities that scale linearly with the set sizes ($m, n$) and logarithmically with the distance threshold ($\delta$). While the complexity exhibits an exponential dependence on the feature dimension $d$---an inherent bound in existing low-dimensional schemes---our construction is distinct in that it isolates this exponential factor exclusively to lightweight cryptographic operations. Furthermore, for high-dimensional sets, by leveraging the \textit{strong global disjointness} assumption, we successfully reduce this dimensional dependence from exponential to strictly linear. We summarize our core contributions as follows:}

% \textit{linear} communication and computation overhead on the set sizes $m, n$, the dimension $d$, and $\log \delta$ without relying on public-key primitives and generic two-party computation techniques such as garbled circuits.
%  
% An exception is an additional $2^d$ factor in our fuzzy PSI protocols (also in previous works) for low-dimensional (small $d$) sets, but unlike previous works, which just applies to lightweight operations.

\begin{enumerate}

% With these new techniques, we first propose scalable fuzzy PSI protocols for low-dimensional sets utilizing spatial hashing and OT-based fuzzy matching, along with domain reduction optimization.
% The protocols achieve an overhead linear with $n, m, \log \delta, 2^d$ for the first time, without factors $O((\log \delta)^d)$ or $O(delta)$ in existing works.
% Moreover, relying on a stronger set assumption used in recent work, we propose fuzzy PSI protocols for high-dimensional sets based on OPRF-based fuzzy matching. It achieves an overhead linear with $n, m, \log \delta, d$, instead of linear dependence on $\delta$.

    \item \textbf{Asymptotically efficient fuzzy matching.} We design two asymptotically efficient fuzzy matching protocols. The first is built from our customized role-reversed OPRF and realizes $O(d\log \delta)$ overhead, compared to $O((\log \delta)^d)$ in previous works. The second utilizes customized OT with $O(d\ell)$ overhead, which is particularly suitable for inputs with short bit length $\ell$.

    \item \textbf{New framework of fuzzy PSI for low-dimensional sets.} We propose a new fuzzy PSI framework on $L_{p \in [1, \infty]}$ distance from a dual-layer hashing design, incorporating the spatial hashing and Cuckoo hashing. Our framework reduces fuzzy PSI to about $m$ fuzzy matching invocations, which outperforms previous constructions with $2^d n$ invocations. Moreover, we introduce an input space reduction optimization while addressing subtle issues of false positives, which enables the use of our efficient OT-based fuzzy matching.

    \item \textbf{Improved fuzzy PSI for high-dimensional sets.} We propose efficient fuzzy PSI protocols for high-dimensional sets on $L_{p \in [1, \infty]}$ distance based on our OPRF- and OT-based fuzzy matching. {Under the {strong} globally disjoint assumption on both parties' sets, our protocols achieve an asymptotic complexity that is linear in the set sizes $m, n$ and the dimension $d$, and logarithmic in the distance threshold $\delta$.}
    
    % Our protocols achieve linear complexity in the set size $m, n$, the dimension $d$, and logarithmic dependence on the distance threshold $\log \delta$, an exponential improvement on $\delta$ over the state-of-the-art work with overhead $O(d(n+\delta m))$. 
    % Besides, we identify that the fuzzy PSI protocols of \cite{van2025} for $L_\infty$ distance leak the sender's private information through. We give a concrete attack and fix this issue.

    % \item \textbf{Prefix optimizations.} We incorporate prefix techniques to optimize the overhead of fuzzy PSI protocols for large distance thresholds, which reduces the computational and communication complexity to logarithmic in the distance threshold.

    \item \textbf{Extensive experimental evaluations.} We conduct extensive experiments under a wide range of parameter settings. The experimental results demonstrate that our protocols significantly outperform state-of-the-art constructions. In low-dimensional settings, for dimensions $d$ ranging from 2 to 8, our protocols achieve up to $145\times$ faster computation and up to $20\times$ lower communication cost. In high-dimensional settings, for dimensions $d$ ranging from 16 to 64, our protocols require up to $36\times$ less running time and up to $54\times$ less communication.

\end{enumerate}

% We give a high-level overview of existing works in the following sections. Technically, existing works first design \textit{single-point} fuzzy matching protocols, and then propose fuzzy mapping to compile fuzzy matching to \textbf{multi-point} fuzzy PSI, while avoiding quadratic invocations of fuzzy matching.
% Due to we classify previous protocols for different scenarios: low-dimensional sets and high-dimensional sets. 
% Moreover, we summarize these works according to different technical lines.

% At a high level, these fuzzy PSI protocols follow a common paradigm: (1) A single-point \emph{fuzzy matching} protocol is first proposed to determine whether a Sender's point $\vecq$ and a Receiver's point $\vecq$ satisfy $\mathsf{dist}(\vecq, \vecw) \le \delta$. (2) The multi-point fuzzy PSI task is then reduced to single-point fuzzy matching via abstract point-pairing techniques, which avoids quadratic invocations of fuzzy matching.

\section{{Related Work}}
\label{sec: Related Work}

% \blue{
% We discuss existing fuzzy PSI protocols for $L_{p \in [1, \infty]}$ distance\footnote{There is also a substantial body of work on fuzzy PSI for Hamming distance~\cite{freedman2004efficient, chmielewski2008fuzzy,ye2009efficient,indyk2006polylogarithmic,uzun2021fuzzy,chakraborti2023distance,blass25}, which we do not discuss in detail here.}.
% We classify these protocols according to the dimensionality of the input sets,
% distinguishing between low-dimensional and high-dimensional settings, which
% exhibit different technical designs.
% The comparison is provided in Table~\ref{tab:complexities}.
% }

{In this section, we discuss existing fuzzy PSI protocols for $L_{p \in [1, \infty]}$ distance\footnote{There is also a substantial body of work on fuzzy PSI for Hamming distance~\cite{freedman2004efficient, chmielewski2008fuzzy,ye2009efficient,indyk2006polylogarithmic,uzun2021fuzzy,chakraborti2023distance,blass25}, which we do not discuss in detail here.}, as they align most closely with our setting.
Because the protocol design of fuzzy PSI is largely dictated by the dimensionality of the input sets, we categorize state-of-the-art constructions into low-dimensional and high-dimensional regimes. A comprehensive comparison of these approaches alongside our work is summarized in Table~\ref{tab:complexities}.}

\subsection{Fuzzy PSI on Low-dimensional Sets} 

Previous fuzzy PSI works for low-dimensional sets instantiate fuzzy mapping using spatial hashing or its variants \cite{garimella2022structure}.
The high-level idea is to partition the space into hypercubes of small side length (also referred to as grid cells), so that the receiver only needs to determine whether any sender’s points lie in cells that are $\delta$-close to receiver's points, without enumerating the entire space.
Despite this efficiency improvement, spatial hashing–based solutions incur an overhead factor of $2^d$, which makes them practical primarily for low-dimensional settings.
We summarize these works below according to their underlying technical approaches.

% Previous fuzzy PSI works for low-dimensional sets instantiate fuzzy mapping by utilizing spatial hashing or its variants pairing. 
% The high-level idea is to tile the entire space by hypercubes of small side-length (also referred as to  grid cells) such that the receiver only need determine whether there are the sender's points lying in its $\delta$-close cells, without enumerating the entire space
% Despite the efficiency improvement, spatial hashing-based solutions incur a factor $2^d$ in the overhead, making them practical primarily for low-dimensional settings. We summarize these works according to different technical paths.

% , i.e., function secret sharing, additive homomorphic encryption, garbled circuits, and oblivious pseudo-random function, respectively.

% Garimella et al. \cite{garimella2022structure} introduce a general-purpose structure-aware PSI framework that can be adapted to achieve fuzzy PSI.
% The key technique is a new fuzzy matching protocol constructed from customized FSS for membership test of $L_\infty$ balls.
% To extend FSS-based fuzzy matching to fuzzy PSI, the work assumes the points in the receiver's set are at least $2\delta$ apart in the $L_\infty$ distance.
% Their communication overhead scales with $O((2 \log \delta)^d)$, but the computation overhead has a factor of $O(\delta^d)$ since the receiver needs to traverse all points that are $\delta$-close with the receiver's points.

\textbf{Function secret sharing (FSS).}
Garimella et al.~\cite{garimella2022structure} introduce a general structure-aware PSI framework that can be adapted to fuzzy PSI.
The key technique is a customized fuzzy matching protocol built from FSS \cite{boyle2015function, boyle2016function, guo2023half} for testing membership in $L_\infty$ balls.
To extend this FSS-based fuzzy matching to fuzzy PSI, they assume that points in the receiver’s set are at least $2\delta$ apart in $L_\infty$ distance.
Their communication overhead scales as $O((2 \log \delta)^d)$.
However, the computation overhead incurs an additional factor of $O(\delta^d)$, since the receiver must traverse all points that are $\delta$-close to each receiver point.

% Built on FSS-based fuzzy matching, their fuzzy PSI protocols reduce the communication complexity from naive $O(\delta^d)$ to $O((2 \log \delta)^d)$.

% for membership tests within a $\delta$-radius $L_\infty$ ball.
% It reduces the communication complexity of fuzzy PSI from $O(\delta^d)$ to $O((2 \log \delta)^d)$.
% However, the work still requires the receiver’s computation to scale with $O(\delta^d)$.
% To address this computation overhead, Garimella et al.~\cite{garimella2024computation} further present an improved fuzzy matching based incremental FSS.
% The core idea is to traverse the prefix representation of the receiver's $\delta$-radius ball, which can be represented by $(\log \delta)^d$ points.
% This fuzzy PSI achieves both computational and communication overhead that scales with $O((2 \log \delta)^d)$.
% Bui et al. \cite{bui2025new} further improve its concrete overhead by eliminating a hidden multiplicative overhead proportional to the
% computational security parameter, while maintaining the same asymptotic complexity.

To reduce this computational overhead, Garimella et al.~\cite{garimella2024computation} propose an improved fuzzy matching protocol based on incremental FSS.
The key idea is to traverse the prefix representation of the receiver’s $\delta$-radius ball, which can be represented using only $(\log \delta)^d$ points.
As a result, their fuzzy PSI achieves both computational and communication overheads scaling as $O\bigl((2 \log \delta)^d\bigr)$.
Bui et al.~\cite{bui2025new} further optimize the concrete efficiency by removing a hidden multiplicative factor proportional to the computational security parameter, while preserving the same asymptotic complexity.

% of all the points
% and reduce the overhead of fuzzy matching to .

However, the main limitation of FSS-based fuzzy PSI is the overhead factor of $O((2 \log \delta)^d)$, which becomes prohibitive for large distance thresholds.
Moreover, these methods only support $L_\infty$ distance, due to the lack of efficient FSS-based constructions for fuzzy matching under general $L_p$ distance.

% However, the main drawback of FSS-based fuzzy PSI is the overhead with a factor of $O((2 \log \delta)^d)$, which becomes inefficient for large distance thresholds.
% Besides, these fuzzy PSI methods only support $L_\infty$ distance, owing to the absence of efficient FSS-based constructions for fuzzy matching under general $L_p$ distance.
 % As a result, the state-of-the-art FSS-based fuzzy PSI achieves .

\textbf{Additively homomorphic encryption (AHE).}
To support general $L_{p \in [1, \infty]}$ distance, van Baarsen and Pu \cite{van2024fuzzy} propose fuzzy matching protocols leveraging the linear homomorphism of AHE.
Their construction incurs $O(d\delta)$ computational and communication overhead.
To realize fuzzy PSI, they employ improved spatial hashing, with the assumption that the receiver’s points are at least $2\delta(d^{1/p}+1)$ apart for $L_{p \in [1, \infty]}$ distance.
Subsequently, Dang et al. \cite{dang2025ccs} improve the fuzzy matching protocol by integrating prefix trie techniques, reducing communication to $O(\log \delta)$ and computation to $O((\log \delta)^d)$.
They adopt the same framework to implement fuzzy PSI for general $L_p$ distances.

However, both of these fuzzy PSI methods suffer from high computational overhead, scaling as $O(d\delta)$ or $O((\log \delta)^d)$, and rely on expensive public-key primitives.

% To support generalized $L_{p \in [1, \infty]}$ distance, Baarsen and Pu \cite{van2024fuzzy} propose new fuzzy matching protocols utilizing the linear homomorphism of AHE. 
% Their fuzzy matching incur $O(d\delta)$ computation and communication overhead.
% Subsequently, Dang et al. \cite{dang2025ccs} improve their fuzzy matching by integrating prefix trie techniques.
% The improved protocols achieve the communication and computation overheads that scale with $O(\log \delta)$ and $O((\log \delta)^d)$, respectively.
% The work \cite{van2024fuzzy}, followed by \cite{dang2025ccs}, utilizes spatial hashing to realize fuzzy PSI with the assumption that the receiver’s points are at least $2\delta(d^{1/p}+1)$ apart from each other for general $L_{p \in [1, \infty]}$ distance.

% Furthermore, the authors extend fuzzy matching to fuzzy PSI using spatial hashing, utilizing the assumption that the receiver’s points are at least $2\delta(d^{1/p}+1)$ apart from each other for general $L_{p \in [1, \infty]}$ distance. 
% Subsequently, Dang et al. \cite{dang2025ccs} improved the overhead of fuzzy matching by integrating prefix trie techniques, which achieve $O(d \log \delta)$ communication overhead and $O((\log \delta)^d)$ computation overhead. 
% It follows the same framework to achieve fuzzy PSI.

% Dang25 further improves the overhead by integrating prefix trie techniques and realizes the communication and computation overhead linear to $O(\log \delta)$.
% To support fuzzy PSI, both works also use spatial hashing to map potentially close
% points into the same cell.

\textbf{Garbled circuits (GC).} To avoid expensive public-key operations, Richardson et al.~\cite{richardson2024fuzzy} leverage GC~\cite{yao1982protocols}, which is built on symmetric primitives, to implement fuzzy matching, and further propose a fuzzy PSI framework based on disjoint hashing, a generalization of spatial hashing. Their framework encodes messages from the GC-based fuzzy matching protocols in an oblivious key-value store~\cite{garimella2021oblivious,bienstock2023near},  which maps a hash identifier for each party’s point to the corresponding next GC message. The resulting fuzzy PSI protocols for $L_\infty$ distance achieve $O(m 2^{d-s} + n 2^{s}) d \log \delta$ communication, where $s \in [0, d]$ is an integer parameter that can be freely chosen.

However, {relying on GCs} results in poor concrete communication and computational efficiency for fuzzy PSI, as observed in~\cite{piske2025distance}.

% Although their protocols are built on symmetric primitives and have logarithmic dependency
% on the distance threshold $\delta$, 

\textbf{Oblivious pseudo-random function (OPRF).}
Recently, some of the most efficient fuzzy PSI protocols~\cite{van2025, piske2025distance} have been constructed using lightweight OPRF techniques. Specifically, van Baarsen and Pu~\cite{van2025} propose OPRF-based fuzzy matching protocols that support general $L_{p \in [1, \infty]}$ distance and further optimize them using prefix trie techniques. To realize fuzzy PSI with these OPRF-based fuzzy matching protocols, they employ spatial hashing and assume that both the sender’s and receiver’s points are slightly separated.
That is the points of both the sender's and receiver’s sets are $2\delta d^{1/p}$ or $2\delta(d^{1/p}+1)$ apart from each other.
Despite these optimizations, the resulting fuzzy PSI still incurs overhead scaling as $O((\log \delta)^{d/2})$ for $L_\infty$ distance and $O(\delta)$ for $L_p$ distance.

Concurrently, Piske et al.~\cite{piske2025distance} propose fuzzy matching protocols based on OPRF and oblivious transfer (OT). Their core technique, distance-aware OT, is a variant of OT in which the choice bit depends on whether the sender’s and receiver’s inputs are close. To achieve fuzzy PSI, they employ spatial hashing, assuming that each grid cell contains at most one sender input. However, unlike other protocols with prefix trie techniques, their protocols incur linear complexity in the distance threshold $\delta$.

\subsection{Fuzzy PSI on High-dimensional Sets}
To make fuzzy PSI practical for high-dimensional sets, it is necessary to mitigate the exponential overhead regarding $d$.
Several works have addressed this challenge, but they rely on some strong assumptions as discussed below. 
Similar to low-dimensional sets, we organize these works according to different technical approaches.
% We also organize these works into different technical lines, as in low-dimensional sets.

\textbf{Function secret sharing (FSS).}
To mitigate the $2^d$ factor in the complexity, Garimella et al.~\cite{garimella2022structure} assume that the receiver’s $\delta$-radius balls are \emph{globally disjoint}, meaning that the projections of all receiver balls do not overlap in any dimension. Under this assumption, they propose an FSS-based fuzzy PSI protocol that achieves communication complexity scaling to $O(d \log \delta)$, but incurs a computational overhead of $O(\delta^d)$.

\textbf{Additively homomorphic encryption (AHE).}
Subsequently, van Baarsen and Pu \cite{van2024fuzzy} relax the globally disjoint assumption by introducing the \textit{existing disjoint} assumption.
That is, for each receiver's ball, there exists at least one dimension on which its projection does not overlap with that of others.
They propose a fuzzy matching protocol based on DDH and extend it to realize fuzzy PSI.
However, their protocols' complexity remains quadratic in $d$.
Recently, Gao et al. \cite{gao2024efficient} introduce fuzzy mapping, generalizing spatial hashing-based solutions, and construct fuzzy PSI for high-dimensional sets based on AHE and DDH.
Their fuzzy PSI protocol requires that both parties satisfy the \textit{existing disjoint} assumption and realizes communication and computation costs linear in $\delta, d, m, n$. 
Dang et al. \cite{dang2025ccs} further optimize these protocols using prefix representations, reducing the dependence on $\delta$ to logarithmic.
Despite these improvements, a key limitation of these works is poor concrete efficiency due to expensive public-key operations.

% , which leads to  in practice.

\textbf{Oblivious pseudo-random function (OPRF).}
To address the above efficiency issue, Zhang et al. \cite{zhang2025fast} recently propose an intermediate assumption, which they call the $s$-separated set assumption.
This assumption states that for every $s$ dimensions of each point, there exists at least one dimension on which its projection does not overlap with that of others.
When $s = 1$, this degenerates the globally disjoint assumption \cite{garimella2022structure}, and when $s = d$, it corresponds to the existing disjoint assumption \cite{van2024fuzzy}. 
They construct fuzzy PSI using shared-output OPRF \cite{dinur2021mpc, alamati2024improved}, which rely on two-party computation to evaluate MPC-friendly low-degree PRF, introducing large communication and computation overhead.
More recently, relying on the globally disjoint assumption on both parties, van Baarsen and Pu \cite{van2025} propose the first fuzzy PSI protocol based solely on lightweight OPRF, achieving communication and computation costs that are linear in $\delta, d, m, n$.

% Utilizing shared-output OPRF, they present a new distributed identifier generation protocol for producing unique identifiers for the parties' points and further construct fuzzy PSI based on these identifiers.
% However, currently efficient instantiations of shared-output OPRF rely on generic two-party computation to evaluate MPC-friendly low-degree PRF. They still introduce larger communication and computation overhead than standard OPRF protocols.

\section{Preliminaries}

% We use a tilde to denote the prefix representation , e.g., $\tilde{q}_k$

\subsection{Notations}
\label{subsec: Notations}

We use $\kappa$ and $\lambda$ to denote the computational and statistical security parameters, respectively. 
The symbols $m$ and $n$ denote the set sizes of the sender and the receiver,
respectively.
Bold letters (e.g., $\vecq$) represent input points with dimension $d$, and $q_k$ denotes the $k$-th coordinate of $\vecq$.
% and $d$ denotes the dimension of an input point.
For a positive integer $a$, $[a]$ denotes the set $\{1, \ldots, a\}$.
The notation $s \leftarrow S$ indicates that $s$ is sampled uniformly at random
from $S$.
The indicator function $1\{\mathsf{event}\}$ equals $1$ if the event occurs and
$0$ otherwise.

\subsection{Fuzzy Matching and Fuzzy PSI}
\label{subsec: Fuzzy Matching and Fuzzy PSI}

We introduce the functionality of fuzzy matching and its secret-shared variant in Figure \ref{Func:FMatch} and that of fuzzy PSI in Figure \ref{Func:FPSI} for some distance metric $\mathsf{dist}$ on $d$-dimensional space. Fuzzy matching can be viewed as a simplified single-point version of fuzzy PSI.

\begin{figure}[!t]
\begin{nffunc}{\Func[(ss)FMatch]}

\noindent \textbf{Parameters:} Sender \SSS and receiver \RRR. Distance metric $\mathsf{dist}(\cdot,\cdot)$ and threshold $\delta$.

\noindent \textbf{Functionality:}

\begin{enumerate}

    \item Wait for input $\vecq \in \ZZ_{2^\ell}^{d}$ from \SSS and $\vecw \in \ZZ_{2^\ell}^{d}$ from \RRR.

    \item Compute $z := 1\{\mathsf{dist}(\vecq, \vecw) \le \delta\}$
    
    % Output 1 to \RRR if , and 0 otherwise.

    \item \textbf{$\mathsf{FMatch}$}: Output $z$ to \RRR.
    
    \item \textbf{$\mathsf{ssFMatch}$}: Sample $z^\SSS, z^\RRR \from \FF_2$ such that $z^\SSS \xor z^\RRR = z$. Output $z^\SSS$ to \SSS and $z^\RRR$ to \RRR.
    
\end{enumerate}

\end{nffunc}
% \vspace{0.5em}
\caption{Functionality of (secret-shared) fuzzy matching.}
\label{Func:FMatch}
\end{figure}

\begin{figure}[!t]
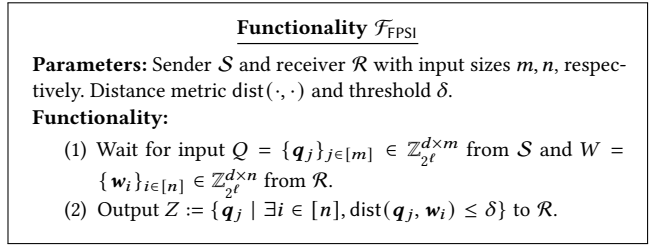

\begin{nffunc}{\Func[FPSI]}

\noindent \textbf{Parameters:} Sender \SSS and receiver \RRR with input sizes $m, n$, respectively. Distance metric $\mathsf{dist}(\cdot,\cdot)$ and threshold $\delta$.

\noindent \textbf{Functionality:}

\begin{enumerate}

    \item Wait for input $Q = \{\vecq_j\}_{j \in [m]} \in \ZZ_{2^\ell}^{d\times m}$ from \SSS and $W = \{\vecw_i\}_{i \in [n]} \in \ZZ_{2^\ell}^{d\times n}$ from \RRR.

    \item Output $Z := \{\vecq_j \mid \exists i\in[n], \mathsf{dist}(\vecq_j, \vecw_i) \le \delta \}$ to \RRR.
    
\end{enumerate}

\end{nffunc}
% \vspace{0.5em}
\caption{Functionality of fuzzy PSI.}
\label{Func:FPSI}
\end{figure}

\subsection{Prefix Representations of Intervals}
\label{subsec: Prefix Representations of Intervals}

We present the prefix trie technique \cite{chakraborti2023distance, van2025, dang2025ccs} that succinctly and efficiently represents an integer interval.

\begin{theorem}[\cite{van2025}]
\label{theorem: prefix}
Given an integer interval $[q-\delta, q+\delta] \in \ZZ_{2^\ell}^{2\delta+1}$ and a point $w \in \ZZ_{2^\ell}$,
there are two algorithms with $u, u^\prime = O(\log \delta)$:

\begin{enumerate}
    \item $\mathsf{PxTrie}(q-\delta, q+\delta) \to \{\tilde{q}_1,\ldots,\tilde{q}_u\}$: The $\mathsf{PxTrie}$ algorithm succinctly encodes the interval $[q-\delta, q+\delta]$ into $u$ prefix nodes $\{\tilde{q}_1,\ldots,\tilde{q}_u\} \in \ZZ_{2^\ell}^u$ that cover the entire interval without overlap. Specifically, $u \le 2 + \log\delta$ when $\delta$ is a power of $2$. 
    % \meng{Add: prefix covers the whole interval}

    \item $ \mathsf{PxPath}(w, \delta) \to \{\tilde{w}_1,\ldots,\tilde{w}_{u^\prime}\}$: The $\mathsf{PxPath}$ algorithm expands each query point $w \in \ZZ_{2^\ell}$ into a path of prefix nodes $\{\tilde{w}_1,\ldots,\tilde{w}_{u^\prime}\} \in \ZZ_{2^\ell}^{u^\prime}$, where {$u^\prime := \lfloor \log (2\delta+1) \rfloor + 1$.} 
    % $u^\prime := 1 + \log\delta$. 
\end{enumerate}
Then, it holds that $\tilde{w}_t \in \{\tilde{q}_1,\ldots,\tilde{q}_u\}$ for some unique $t \in [u^\prime]$ if and only if $w \in [q-\delta, q+\delta]$. 
% \meng{Add unique property for multiple intervals.}
\end{theorem}

{That means if $w \in [q-\delta, q+\delta]$, there exists a unique $t \in [u^\prime]$ such that $\tilde{w}_{t} \in \{\tilde{q}_1,\ldots,\tilde{q}_u\}$; otherwise, $\{\tilde{q}_{t}\}_{t\in[u]} \intersection \{\tilde{w}_{t}\}_{t\in[u^\prime]}   = \emptyset$.
To simplify the exposition of security proofs, we establish the following three direct corollaries.}

\begin{corollary}
\label{corollary: Distinctness of PxTrie}
{Given any two intervals $[q-\delta, q+\delta], [q^\prime-\delta, q^\prime+\delta] \in \ZZ_{2^\ell}^{2\delta+1}$ such that $|q - q^\prime| > 2\delta$, it holds that $\mathsf{PxTrie}(q-\delta, q+\delta)\cap \mathsf{PxTrie}(q^\prime-\delta, q^\prime+\delta) = \emptyset$.}
\end{corollary}

\begin{proof}
{Suppose for contradiction there exists $z^* \in \mathsf{PxTrie}(q-\delta, q+\delta)\cap \mathsf{PxTrie}(q^\prime-\delta, q^\prime+\delta)$.
Let $z$ be any specific value covered by the prefix $z^*$ (e.g., the value obtained by setting all trailing wildcard bits to $0$).
% Let $z$ be a value with prefix $z^*$ with all 0s in trailing wildcards. 
This means $z \in [q-\delta, q+\delta]\cap [q^\prime-\delta, q^\prime+\delta]$, which implies $|z - q| \le \delta$ and $|z - q^\prime| \le \delta$.
This yields  $|q - q^\prime| = |(q - z) - (q^\prime - z)| \le |q - z| + |q^\prime - z| \le 2\delta$.
This contradicts the assumption $|q - q^\prime| > 2\delta$.
Therefore, the intersection is empty.}
\end{proof}

\begin{corollary}
\label{corollary: Distinctness of PxPath}
{Given any two points $w, w^\prime \in \ZZ_{2^\ell}$ such that $|w - w^\prime| > 2\delta$, it holds that $\mathsf{PxPath}(w,\delta)\cap \mathsf{PxPath}(w^\prime,\delta) = \emptyset$.}
\end{corollary}

% \begin{corollary}
% \label{corollary: Distinctness of PxPath}
% \blue{Given a point $q \in \ZZ_{2^\ell}$, for any two points $w, w^\prime \in \ZZ_{2^\ell}$ such that $|w - w^\prime| > 2\delta$, it holds that there do not exist two distinct prefixes $z, z^\prime$ such that $z = \mathsf{PxTrie}(q-\delta, q+\delta) \cap \mathsf{PxPath}(w, \delta)$ and $z^\prime = \mathsf{PxTrie}(q-\delta, q+\delta) \cap \mathsf{PxPath}(w^\prime, \delta)$.
% }
% \end{corollary}

\begin{proof}
{Suppose for contradiction there exists $z \in \mathsf{PxPath}(w,\delta)\cap \mathsf{PxPath}(w^\prime,\delta)$.
By Theorem~\ref{theorem: prefix}, $\mathsf{PxPath}(\cdot)$ outputs $u^\prime$ prefixes with at most $u^\prime - 1$ trailing wildcards, where $u^\prime = \lfloor \log (2\delta+1) \rfloor + 1$.
Since the shared prefix $z$ has at most $u^\prime - 1$ wildcards, this yields $|w - w^\prime| \leq 2^{u^\prime - 1} - 1 = 2^{\lfloor \log (2\delta+1) \rfloor} - 1 \leq 2^{\log (2\delta+1)} - 1 = 2\delta$. 
This contradicts the assumption $|w - w^\prime| > 2\delta$.
Therefore, the intersection is empty.}
\end{proof}

\begin{corollary}
\label{corollary: uniqueness}
{Given any interval $[q-\delta, q+\delta] \in \ZZ_{2^\ell}^{2\delta+1}$, for any two points $w, w^\prime \in \ZZ_{2^\ell}$ such that $|w - w^\prime| > 2\delta$, there exists at most one point $v \in \{w, w^\prime\}$ such that $\mathsf{PxTrie}(q-\delta, q+\delta) \cap \mathsf{PxPath}(v, \delta) \neq~\emptyset$.}
\end{corollary}

\begin{proof}
{Suppose for contradiction that both intersections are non-empty. By Theorem~\ref{theorem: prefix}, this implies $w \in [q-\delta, q+\delta]$ and $w^\prime \in [q-\delta, q+\delta]$.
This contradicts the assumption that $|w - w^\prime| > 2\delta$. Therefore, at most one point $v \in \{w, w^\prime\}$ can yield a non-empty intersection.}
\end{proof}

\subsection{Spatial Hashing}
\label{subsec: Spatial Hashing}

To extend single-point fuzzy matching to multi-point fuzzy PSI,
prior works~\cite{garimella2022structure, gao2024efficient, bui2025new,
van2024fuzzy, van2025, dang2025ccs, piske2025distance, richardson2024fuzzy}
partition the entire input space into smaller hypercubes with side length
$2\delta$, also referred to as grid cells. This tiling ensures that each point
needs to be matched only against points residing in a bounded number of cells, instead of the whole input space.
We formally define this approach, known as \emph{spatial hashing}, below.

% Since the number of relevant cells still scales  exponentially in the dimension of the space, this technique is especially useful  in low dimensions.
We first introduce the definition of a cell and its identifier. A cell is defined as a $L_\infty$ hypercube of side length $2\delta$. 
The function $\mathsf{cell}_{2\delta}(\cdot)$ maps a point $\vecq$ to the identifier of the cell.
The identifier of the 
corresponding cell $\CCC_\vecq := \mathsf{cell}_{2\delta}(\vecq)$ is given by $(\lfloor \frac{q_1}{2\delta}\rfloor \cdot 2\delta, \ldots, \lfloor \frac{q_d}{2\delta}\rfloor \cdot~2\delta)$.
Spatial hashing consists of two algorithms:
\begin{enumerate}
    \item $\mathsf{cellhash}(\cdot)$: The function maps a point $\vecq \in \ZZ_{2^\ell}^d$ into its unique identifier $\CCC_\vecq$ of the cell containing $\vecq$.

    \item $\mathsf{ballhash}(\cdot)$: The function maps a point $\vecq \in \ZZ_{2^\ell}^d$ into a set of identifiers $\{\CCC_{\vecq,h}\}_{h \in [2^d]}$ of the $2^d$ cells that intersect the $L_\infty$ ball of radius $\delta$ centered at $\vecq$. 
    % \meng{Add refs to explain why $2^d$ cells.}
    
\end{enumerate}

\subsection{(Programmable) Oblivious PRF}
\label{subsec: Oblivious Programmable PRF}

An oblivious pseudorandom function (OPRF) is a two-party protocol in which the
receiver inputs a batch\footnote{We note that our protocols invoke only
multi-query OPRFs. For simplicity, we omit the qualifier ``multi-query'' in the
remainder of the paper.} of queries $\{q_i\}_{i \in [n]}$. The sender obtains a random function $F$, while the receiver learns only the
corresponding evaluations $\{F(q_i)\}_{i \in [n]}$.

Programmable OPRF (OPPRF), introduced by Kolesnikov et al.~\cite{kolesnikov2017practical}, extends this functionality by allowing the sender to program the function $F$ at designated inputs: for each sender-chosen pair $(k_j, v_j)$, the functionality enforces $F(k_j) = v_j$, whereas all non-programmed inputs are mapped to random outputs.
We present the functionality of OPPRF in Figure \ref{Func: opprf}.
The state-of-the-art OPPRF protocols~\cite{rindal2021vole,
raghuraman2022blazing} are constructed by combining OPRF with oblivious key-value stores (OKVS)~\cite{garimella2021oblivious}. Our protocols directly use OPPRF as a fundamental building block, instead of OPRF.

\begin{figure}[t]
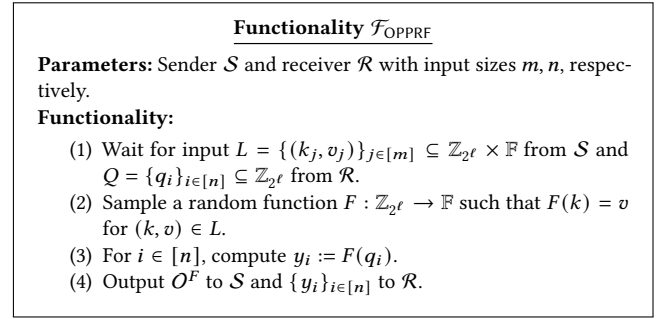

\begin{nffunc}{\Func[OPPRF]}

\noindent \textbf{Parameters:} Sender \SSS and receiver \RRR with input sizes $m, n$, respectively. 

\noindent \textbf{Functionality:}

\begin{enumerate}

    \item Wait for input $L = \{(k_{j}, v_j)\}_{j \in [m]} \subseteq \ZZ_{2^\ell} \times \FF$ from \SSS and $Q = \{q_i\}_{i \in [n]} \subseteq \ZZ_{2^\ell}$ from \RRR.

    \item Sample a random function $F : \ZZ_{2^\ell} \to \FF$ such that $F(k) = v$ for $(k, v) \in L$. 
    
    \item For $i \in [n]$, compute $y_i := F(q_i)$.

    \item Output $\OOO^{F}$ to \SSS and $\{y_i\}_{i \in [n]}$ to \RRR. 

\end{enumerate}

\end{nffunc}
% \vspace{0.5em}
\caption{Functionality of oblivious programmable PRF.}
\label{Func: opprf}
\end{figure}

\subsection{Weak Labeled PSI}
\label{subsec: Weak Labeled PSI}

As illustrated in Figure \ref{Func: wLPSI}, weak labeled PSI \cite{van2025} is a two-party protocol, where the sender inputs $\{(q_{j}, \mathsf{label}_j)\}_{j \in [m]}$ and the receiver inputs $W = \{w_i\}_{i \in [n]}$.  The receiver learns $\{(q_j, \mathsf{label}_j)\}$ for $q_j \in W$.
It differs from the standard notion of labeled PSI in that the functionality outputs the tuple of the intersection item and its corresponding label, instead of just the label. \cite{van2025} presents an efficient protocol realizing weak labeled PSI from OPRF.

\begin{figure}[t]
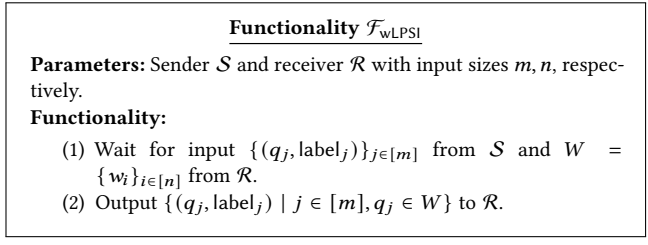

\begin{nffunc}{\Func[wLPSI]}

\noindent \textbf{Parameters:} Sender \SSS and receiver \RRR with input sizes $m, n$, respectively. 

\noindent \textbf{Functionality:}

\begin{enumerate}

    \item Wait for input $\{(q_{j}, \mathsf{label}_j)\}_{j \in [m]}$ from \SSS and $W = \{w_i\}_{i \in [n]}$ from \RRR.

    \item Output $\{(q_j, \mathsf{label}_j) \mid j\in[m], q_j \in W \}$ to \RRR.

\end{enumerate}

\end{nffunc}
% \vspace{0.5em}
\caption{Functionality of weak labeled PSI.}
\label{Func: wLPSI}
\end{figure}

\begin{figure}[!t]
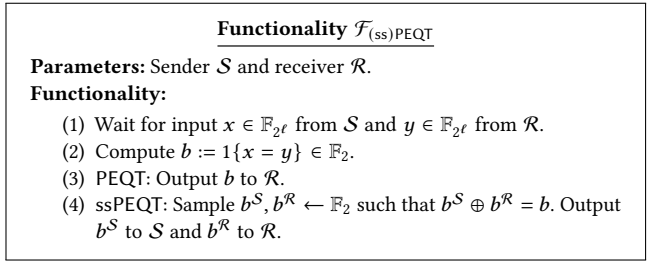

\begin{nffunc}{\Func[(ss)PEQT]}
% \vspace{0.5em}

\noindent \textbf{Parameters:} Sender \SSS and receiver \RRR.

\noindent \textbf{Functionality:}

\begin{enumerate}

    \item Wait for input $x \in \FF_{2^\ell}$ from \SSS and $y \in \FF_{2^\ell}$ from \RRR.

    \item Compute $b := 1\{x = y\} \in \FF_2$.

    \item \textbf{$\mathsf{PEQT}$}: Output $b$ to \RRR.
    
    \item \textbf{$\mathsf{ssPEQT}$}: Sample $b^\SSS, b^\RRR \from \FF_2$ such that $b^\SSS \xor b^\RRR = b$. Output $b^\SSS$ to \SSS and $b^\RRR$ to \RRR.
    
\end{enumerate}

\end{nffunc}
% \vspace{0.5em}
\caption{Functionality of (secret-shared) private equality test.}
\label{Func: ssPEQT}
\end{figure}

\subsection{(Secret-shared) Private Equality Test}

The functionality of the private equality test (PEQT) and its secret-shared variant are given in Figure \ref{Func: ssPEQT}.

\subsection{Cuckoo Hashing}
\label{Cuckoo Hashing}

Cuckoo hashing was introduced by Pagh and Rodler~\cite{pagh2004cuckoo}. It maps $m$ items into a table $T$ with $m' = \epsilon m$ bins using $k$ hash functions $H_1, \ldots, H_k$.
Cuckoo hashing ensures that each bin contains at most one item through the following procedure.
For an item $x$, insert it into any empty bin among $T[H_1(x)], \ldots, T[H_k(x)]$.
If all such bins are occupied, randomly select one of the $k$ bins, evict the existing item $y$, and attempt to reinsert $y$ using the same set of hash functions. This process is repeated until all items are successfully placed.
Recent works~\cite{pinkas2018scalable} empirically show that with appropriate parameters of $\epsilon, k$, Cuckoo hashing can achieve negligible failure probability.

\subsection{Oblivious Permutation}

The oblivious permutation functionality receives a secret-shared vector from the two parties. 
Then, the sender learns a random permutation, and the receiver obtains the permuted vector.
The functionality is given in Figure \ref{Func:Perm}. Recent works \cite{peceny2025efficient, han2025concretely, chase2020secret} provide efficient protocols for oblivious permutation and its variants.

\begin{figure}[!t]
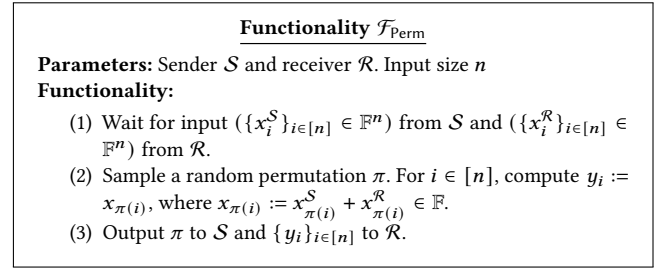

\begin{nffunc}{\Func[Perm]}

\noindent \textbf{Parameters:} Sender \SSS and receiver \RRR. Input size $n$

\noindent \textbf{Functionality:}

\begin{enumerate}

    \item Wait for input $(\{x_i^\SSS\}_{i \in [n]} \in \FF^n)$ from \SSS and $(\{x_i^\RRR\}_{i \in [n]} \in \FF^n)$ from \RRR.

    \item Sample a random permutation $\pi$. For $i \in [n]$, compute $y_i := x_{\pi(i)}$, where $x_{\pi(i)} := x_{\pi(i)}^\SSS + x_{\pi(i)}^\RRR \in \FF$.

    \item Output $\pi$ to \SSS and $\{y_i\}_{i \in [n]}$ to \RRR.
    
\end{enumerate}

\end{nffunc}
% \vspace{0.5em}
\caption{Functionality of oblivious permutation.}
\label{Func:Perm}
\end{figure}

\section{Technical Overview}
\label{Sec: Technical Overview}

\iffalse
Gao et al. highlight that most fuzzy PSI protocols typically involve two distinct phases: coarse mapping and refined filtering. (1) In the coarse mapping phase, initial pairings are established between a receiver point $\vecw$ and a sender point $\vecq$ through some abstract fuzzy mapping schemes, if they are potentially $\delta$-close.
However, this phase may introduce false positives.
% identifiers are assigned to points from both the sender and receiver through some abstract hashing techniques, establishing  between a receiver point $\vecw$ and a sender point $\vecq$ if they share the same identifier and.
(2) In the refined filtering phase, a \textit{fuzzy matching} approach is applied to each pair formed during the coarse mapping phase, which determines whether a receiver’s point $\vecw$ and a sender’s point $\vecq$ satisfy $\mathsf{dist}(\vecq,\vecw) \le \delta$. 
This eliminates the false positives and yields the final results to the receiver.
\fi

% Our protocols follow the same two-phase paradigm, but introduces new techniques for fuzzy matching and  \meng{framework from fuzzy mapping and fuzzy matching}

In this section, we first introduce efficient fuzzy matching protocols and then present a new framework to compile our single-point fuzzy matching to multi-point fuzzy PSI.

\subsection{Asymptotically Efficient Fuzzy Matching}
\label{subsec: Asymptotically Efficient Fuzzy Matching}

We propose two fuzzy matching protocols for $L_\infty$ distance based on OPRF and OT, respectively. 
The two protocols target different application scenarios, depending on the bit length of the input points.
The extension to $L_p$ distance of OT-based fuzzy matching is provided in Appendix \ref{appendix: OT-based Fuzzy Matching for $L_p$ Distance}.

\textbf{Construction from role-reversed oblivious PRF.} 
Our first construction is based on programmable OPRF (OPPRF, see Section~\ref{subsec: Oblivious Programmable PRF}),
% See Section \ref{subsec: Oblivious Programmable PRF} for details of OPPRF.  
and achieves $O(d \log \delta)$ computational and communication overhead.
This improves upon the state-of-the-art work~\cite{van2025}, which also relies on OPPRF but incurs
computational overhead $O(d \log \delta + (\log \delta)^{d/2})$.

Our starting point is the OPPRF-based fuzzy matching protocol from \cite{van2025}.
Specifically, for each dimension $k \in [d]$, the sender programs an OPPRF, which maps all points in the interval $[q_k-\delta, q_k+\delta]$ to the same random value $r_k$. The receiver then obtains the OPPRF evaluation $\hat{r}_k$ on input $w_k$. Finally, the two parties privately check whether $\sum_{k \in [d]} r_k = \sum_{k \in [d]} \hat{r}_k$. With appropriate parameters, this holds if and only if $w_k \in [q_k-\delta, q_k+\delta]$ for all $k \in [d]$. 
This elegant protocol incurs an overhead of $O(d \delta)$.
The work \cite{van2025} further optimizes this protocol to achieve logarithmic dependence on $\delta$ using the prefix trie technique (see Section \ref{subsec: Prefix Representations of Intervals}).
This enables representing the interval $[q_k-\delta, q_k+\delta]$ with only $u = O(\log \delta)$ elements, hence reducing the number of points programmed by the sender in the OPPRF from $2\delta + 1$ to $u$.
Nevertheless, as a trade-off, the receiver obtains $u^\prime = O(\log \delta)$ OPPRF evaluations $\{\hat{r}_{k, t}\}_{t \in [u^\prime]}$ for each dimension $k$.
As a result, the receiver has to compare all possible combinations of $\sum_{k \in [d]} \hat{r}_{k, t_k}$ for $t_k \in [u^\prime]$ with the sender's $\sum_{k \in [d]} r_k$.
The optimization achieves $O(d \log \delta)$ communication overhead; however, it unfortunately causes a computational overhead of $O((\log \delta)^{d/2})$ even with additional optimizations.

To address this problem, we propose a new fuzzy matching protocol based on two
\emph{role-reversed} OPPRF invocations.
% , achieving only $O(d \log \delta)$ overhead. 
Specifically, our protocol incorporates prefix trie techniques and executes the first OPPRF like that of~\cite{van2025}. 
After the receiver obtains OPPRF evaluations
$\{\hat{r}_{k,1}, \ldots, \hat{r}_{k,u'}\}$ for each dimension $k$, the parties invoke a second OPPRF with reversed roles to compress these $u'$ evaluations into a
single value.
More concretely, for each dimension $k \in [d]$, the receiver (as a sender $\SSS$ in this OPPRF) programs an
OPPRF that maps these $u'$ points
$\{\hat{r}_{k,1}, \ldots, \hat{r}_{k,u'}\}$ to the same random value
$\hat{s}_k$, while the sender (as a receiver $\RRR$ in this OPPRF) obtains the OPPRF evaluation $s_k$ on input
$r_k$. 
{In Figure \ref{Fig: technical overview oprf-based fmatch}, we illustrate the two role-reversed invocations of OPPRF for clarity.}
As shown by our careful analysis in
Section~\ref{sec: Fuzzy Matching for L-inf Distance from Oblivious PRF}, this
construction guarantees that $s_k = \hat{s}_k$ if and only if $w_k \in [q_k-\delta, q_k+\delta]$. 
% if $r_k \in \{\hat{r}_{k,1}, \ldots, \hat{r}_{k,u'}\}$, then
Finally, the parties privately check whether
$\sum_{k \in [d]} s_k = \sum_{k \in [d]} \hat{s}_k$.

It is worth noting that, although our protocol requires two OPPRF invocations, it maintains $O(d \log \delta)$ communication complexity (with a slightly larger constant factor) and significantly reduces the computational overhead from $O((\log \delta)^{d/2})$ in \cite{van2025} to $O(d \log \delta)$.

\begin{figure}[!t]
	\centering
	\includegraphics[width=0.42\textwidth]{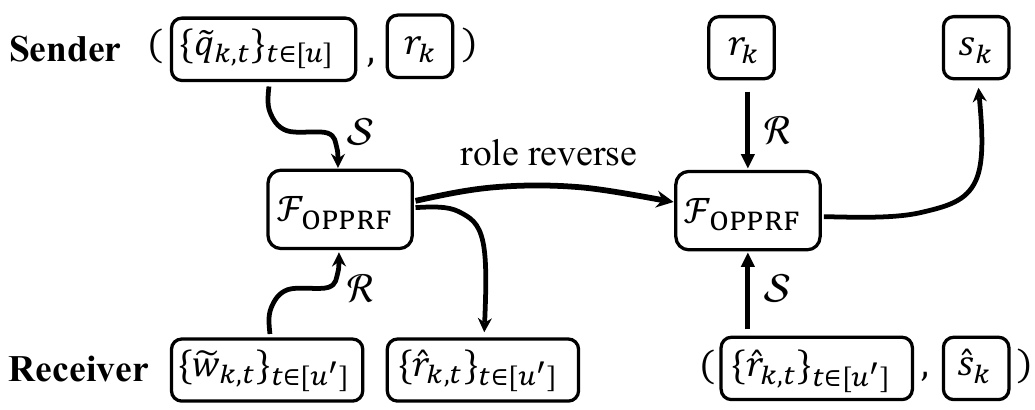}
	\caption{Fuzzy matching from role-reversed OPPRF.}
	\label{Fig: technical overview oprf-based fmatch}
\end{figure}

\textbf{Construction from oblivious transfer.}
Our second construction is based on OT and incurs $O(d\ell)$ communication and computational overhead, where $\ell$ denotes the bit length of the input points. When $\ell = O(\log \delta)$, this construction
achieves $O(d \log \delta)$ overhead as our OPRF-based approach. 
Notably, the OT-based construction achieves smaller concrete communication costs, as it eliminates the hidden dependence on the statistical security parameter~$\lambda$ inherently required in OPRF-based constructions \cite{van2025}.
In particular, our OT-based construction leverages recent OT protocols for computing the most significant bit (MSB) and AND~\cite{rathee2020cryptflow2, huang2022cheetah}. Although technically straightforward and largely overlooked in recent works, this approach achieves highly competitive performance, particularly for small input domains.
As demonstrated in our fuzzy PSI protocols for low-dimensional sets, we reduce fuzzy PSI to fuzzy mapping with a small input bit length, namely $\ell = O(\log \delta)$.
This significantly reduces communication overhead compared to OPRF-based approaches due to eliminating the factor $\lambda$, while incurring only $O(\log \log \delta)$ communication rounds.

% As demonstrated in our fuzzy PSI protocols, we reduce fuzzy PSI to fuzzy mapping with small input bit length, namely $\ell = O(\log \delta)$, which substantially reduces communication overhead compared to previous approaches and our OPRF-based protocol, while incurring modest communication rounds $O(\log \log \delta)$.

% \blue{We note that, unlike the constant-round OPRF-based fuzzy matching, our OT-based
% protocols require $O(\log \ell)$ communication rounds. Nevertheless, our
% experimental results show that even in bandwidth- or latency-constrained
% networks, the overall performance remains superior. This is because our
% protocols substantially reduce communication overhead compared to
% previous approaches, while the number of rounds remains modest (e.g.,
% $\log \log \delta$ in our fuzzy PSI constructions).}

% the round is still small, e.g., $\log \log \delta$ in our fuzzy PSI in low-dimensional sets, and our fuzzy matching protocols still have advantages over previous protocols due to significantly reduced communication complexity.
% our protocols greatly balance the round and communication complexity}

\subsection{New Framework for Fuzzy PSI}
\label{subsec: New Framework for Fuzzy PSI}
% \blue{To avoid quadratic overhead in set size throughout pairwise fuzzy matching, similar to previous works, we introduce additional mild assumptions on sets.
% More precisely, we follow the same assumptions for low- and high-dimensional sets as the state-of-the-art work BP25\footnote{We note that although xxx presents fuzzy PSI protocols with weaker assumptions on the receiver side, their evaluation utilizes the same assumption as our protocols}.}

% \begin{figure}[htbp]
%     \centering
%     % 第一张子图
%     \begin{subfigure}[b]{0.1\textwidth}
%         \centering
%         \includegraphics[width=\textwidth]{figure/Figure 3a.pdf}
%         \caption{Prior work (8 cells)}
%         \label{fig:sub1}
%     \end{subfigure}%
%     \begin{subfigure}[b]{0.1\textwidth}
%         \centering
%         \includegraphics[width=\textwidth]{figure/Figure 3b.pdf}
%         \caption{Our work (2 cells)}
%         \label{fig:sub2}
%     \end{subfigure}
    
%     \caption{XXX}
%     \label{fig:total_figure}
% \end{figure}

{Building upon our single-point fuzzy matching, we now construct multi-point fuzzy PSI protocols. In this section, we focus exclusively on low-dimensional sets, deferring our high-dimensional constructions to Section~\ref{Sec: Fuzzy PSI over High-dimensional Sets} and Appendix~\ref{Appendix: Details of Fuzzy PSI for High-dimensional Sets}. To circumvent the prohibitive $\mathcal{O}(mn)$ quadratic overhead inherent to naive pairwise comparisons, we adopt spatial hashing (see Section~\ref{subsec: Spatial Hashing}) to partition the global $d$-dimensional input space into grid cells of side length $2\delta$. For the sender's set, we assume each grid cell contains at most one point. This structural assumption exactly mirrors the setting in~\cite{piske2025distance} and is strictly weaker than the constraints imposed by the majority of prior fuzzy PSI schemes~\cite{van2025, gao2024efficient, dang2025ccs, zhang2025fast}.
Furthermore, for clarity of exposition, the following protocols further assume that each grid cell intersects at most one $\delta$-radius ball centered at a receiver point. This constraint aligns with the protocols in~\cite{van2025} and the empirical evaluation framework utilized in~\cite{piske2025distance}. We discuss the extension to arbitrary receiver sets, achieved by leveraging stateful spatial hashing~\cite{richardson2024fuzzy, piske2025distance}, in Section~\ref{subsec: Extension to Arbitrary Receiver's Set}.
}

% We further extend our single-point fuzzy matching to multi-point fuzzy PSI. 
% Here, we focus on low-dimensional sets and defer the protocols for high-dimensional sets to Section~\ref{Sec: Fuzzy PSI over High-dimensional Sets}.
% To avoid quadratic overhead through pairwise fuzzy matching, similar to most previous works, we employ spatial hashing (see Section \ref{subsec: Spatial Hashing}) to partition the entire $d$-dimensional input space into grid cells with side length $2\delta$, and assume that for the sender’s set, each grid cell contains at most one sender's point.
% This assumption is identical to that made in
% \cite{piske2025distance} and is strictly weaker than those adopted in most prior fuzzy PSI constructions \cite{van2025, gao2024efficient, dang2025ccs, zhang2025fast}.
% For ease of exposition, we additionally assume in the following protocols that, for the receiver’s set, each grid cell intersects with at most one
% $\delta$-radius ball centered at a receiver point. This assumption is consistent
% with that of~\cite{van2025}, and the experimental evaluation in
% \cite{piske2025distance} also adopts the same setting for efficiency. 
% {We defer the extension to arbitrary receiver sets using stateful spatial
% hashing~\cite{richardson2024fuzzy, piske2025distance} to
% Section~\ref{subsec: Extension to Arbitrary Receiver's Set}.}

{\textbf{New framework via dual-layer hashing.} We present a new fuzzy PSI framework for generalized $L_{p \in [1, \infty]}$ distance, driven by a \textit{dual-layer hashing} strategy that combines spatial hashing with Cuckoo hashing. The core innovation of our approach is an interactive fuzzy mapping phase, which filters the candidate space down to approximately $m$ pairs of points for the final fuzzy matching. This design substantially reduces the overhead of prior OPRF-based constructions~\cite{piske2025distance, van2025}, which implicitly require $2^d n$ fuzzy matching invocations.} 

{To understand this reduction, we first recall existing fuzzy mapping approaches based on spatial hashing. The sender applies a $\mathsf{cellhash}$ function to assign each of its points $\vecq$ to a \emph{unique} cell $\mathcal{C}_{\vecq}$. Conversely, the receiver applies a $\mathsf{ballhash}$ function to associate each of its points $\vecw$ with $2^d$ distinct cells, $\{\mathcal{C}_{\vecw, h}\}_{h \in [2^d]}$, representing all cells intersected by the $\delta$-radius ball centered at $\vecw$. To find matches, the receiver only needs to compare $\vecw$ against sender points located in these $2^d$ cells, since any sender point $\vecq$ satisfying $\mathsf{dist}(\vecq, \vecw) \le \delta$ must reside in one of them. However, when scaled across all $n$ receiver points, this paradigm forces prior protocols to execute approximately $\mathcal{O}(2^d n)$ heavy fuzzy matching operations.}

\begin{figure}[!t]
    \centering
    % --- 子图 (a) Prior work ---
    \begin{subfigure}[b]{0.23\textwidth}
        \centering
        \resizebox{0.8\textwidth}{!}{
        \begin{tikzpicture}[scale = 1]
            % 绘制内部浅灰色网格线
            \draw[step=1cm, black!50] (0,0) grid (4,4);

            % 2. 绘制最外层黑色边框
            \draw[thick, black] (0,0) rectangle (4,4);
            
            % 绘制斜线背景区域
            \fill[pattern=north east lines, pattern color=gray!50] (0,0) rectangle (2,2);
            \fill[pattern=north east lines, pattern color=gray!50] (2,2) rectangle (4,4);

            % 绘制半透明蓝色方块 w1
            \filldraw[fill=blue!30, fill opacity=0.5, draw=blue, thick] (0.2, 0.2) rectangle (1.2, 1.2);
            \fill (0.7, 0.7) circle (1.5pt) node[right, opacity=1] {$\vecw_1$};
            
            % 【新增】在 w1 左上角增加 q2 点
            \fill (0.3, 1.05) circle (1.2pt) node[right, opacity=1] {$\vecq_2$};
            
            % 绘制半透明蓝色方块 w2
            \filldraw[fill=blue!30, fill opacity=0.5, draw=blue, thick] (2.3, 2.3) rectangle (3.3, 3.3);
            \fill (2.8, 2.8) circle (1.5pt) node[right, opacity=1] {$\vecw_2$};
            
            % 【修改】q1 处增加圆点
            \fill (2.7, 1.3) circle (1.5pt) node[right] {$\vecq_1$};
        \end{tikzpicture}
        }
        % \caption{Prior work}
    \end{subfigure}%
    \begin{subfigure}[b]{0.23\textwidth}
        \centering
        \resizebox{0.8\textwidth}{!}{
        \begin{tikzpicture}[scale = 1]
            % 绘制内部浅灰色网格线
            \draw[step=1cm, black!50] (0,0) grid (4,4);

            % 2. 绘制最外层黑色边框
            \draw[thick, black] (0,0) rectangle (4,4);
            
            % 绘制斜线背景区域
            \fill[pattern=north east lines, pattern color=gray!50] (2,1) rectangle (3,2);
            \fill[pattern=north east lines, pattern color=gray!50] (0,1) rectangle (1,2);
            
            % 绘制半透明蓝色方块 w1
            \filldraw[fill=blue!30, fill opacity=0.5, draw=blue, thick] (0.2, 0.2) rectangle (1.2, 1.2);
            \fill (0.7, 0.7) circle (1.5pt) node[right, opacity=1] {$\vecw_1$};
            
            % 【新增】在 w1 左上角增加 q2 点
            \fill (0.3, 1.05) circle (1.2pt) node[right, opacity=1] {$\vecq_2$};
            
            % 绘制半透明蓝色方块 w2
            \filldraw[fill=blue!30, fill opacity=0.5, draw=blue, thick] (2.3, 2.3) rectangle (3.3, 3.3);
            \fill (2.8, 2.8) circle (1.5pt) node[right, opacity=1] {$\vecw_2$};
            
            % 【修改】q1 处增加圆点
            \fill (2.7, 1.3) circle (1.5pt) node[right] {$\vecq_1$};
        \end{tikzpicture}
        }
        % \caption{Our work}
    \end{subfigure}
    
    \caption{Comparison of fuzzy matching invocations of prior works (left) and ours (right) represented by gray shaded cells.}
    \label{Fig: Comparison of fuzzy matching invocations}
\end{figure}
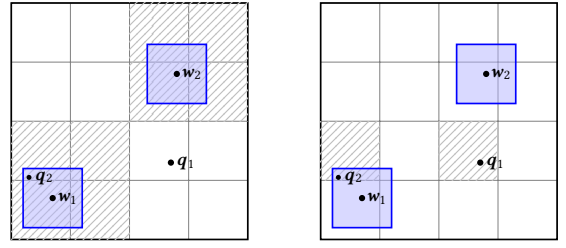

{Our key observation is that the vast majority of these $2^d n$ receiver search cells are empty on the sender's side, since there are at most $m$ cells in which the sender's points lie.
Consequently, the parties only need to perform fuzzy matching on the $m$ specific cells that actually contain the sender's input points. This insight reduces the required number of fuzzy matching invocations from $\mathcal{O}(2^d n)$ to just $\mathcal{O}(m)$ (as illustrated in Figure~\ref{Fig: Comparison of fuzzy matching invocations}). 
To realize this efficiently, we draw inspiration from circuit-based PSI~\cite{pinkas2019efficient, rindal2021vole, hao2024unbalanced} and employ Cuckoo hashing (detailed in Section~\ref{Cuckoo Hashing}) over spatial hashing. Concretely, the sender uses Cuckoo hashing to map its $m$ cell-point tuples $(\mathcal{C}_{\vecq}, \vecq)$ into a table of size $\mathcal{O}(m)$, while the receiver applies simple hashing to its $2^d n$ tuples $(\mathcal{C}_{\vecw,h}, \vecw)$.
Because both parties use identical hash functions, candidate pairs sharing the same cell identifier will collide in the same bin, isolating the valid candidates for fuzzy matching.}

{To further minimize the number of fuzzy matching invocations per bin, our protocol integrates an OPPRF.
Specifically, for each bin $i$, the receiver programs the OPPRF\footnote{
{
Our protocol actually performs a single OPPRF invocation over all bins jointly, instead of an OPPRF instance per bin. This avoids the receiver from padding each bin with dummy items to hide the number of items per bin. 
Similar techniques have been adopted in circuit-based PSI~\cite{pinkas2019efficient, rindal2021vole, raghuraman2022blazing}.}
}
such that, for every tuple $(\mathcal{C}_{\vecw,h}, \vecw)$ residing in that bin, the function maps the cell identifier $\mathcal{C}_{\vecw,h}$ to $\vecw - \vecr_i$, utilizing a single random mask $\vecr_i$ for the entire bin. The sender then queries the OPPRF using its input cell $\mathcal{C}_{\vecq}$. If a matching cell exists in bin $i$, the sender successfully learns $\vecw - \vecr_i$, while the receiver retains the mask $\vecr_i$. Together, these two values constitute secret shares of the receiver's point $\vecw$. As a result, each bin $i$ yields at most one candidate pair: the sender holds its plaintext point $\vecq$, and the receiver's point $\vecw$ is secret-shared between them. Otherwise (no matched cells), according to the functionality of OPPRF, secret shares of a random value will be learned. The parties can then process these pairs independently across all bins using a secret-shared variant of our fuzzy matching protocol.}

% , resulting in a heavier communication overhead of $\mathcal{O}(n 2^d \cdot d\log \delta \cdot \lambda)$.

% We emphasize that while our protocol still technically processes $\mathcal{O}(2^d n)$ items, we effectively decouple this exponential dimensional blowup from the expensive fuzzy matching phase. By shifting the $\mathcal{O}(2^d n)$ workload entirely into the lightweight, batched OPPRF, we achieve a dramatic performance gain. This is because the cost of a fuzzy matching invocation scales multiplicatively by at least $\mathcal{O}(\log \delta)$, whereas the overhead of the OPPRF is strictly independent of the distance threshold $\delta$.

% It is worth noting that while our protocols still have the $O(2^d n)$ overhead, we actually move this overhead from fuzzy matching to significantly more efficient OPPRF.
% The reason is that fuzzy matching has a multiplicative factor at least $O(\log \delta)$ in the overhead, but this overhead of OPPRF is independent of the threshold $\delta$.

{\textbf{OT-based fuzzy matching with reduced input domain.} 
We instantiate the remaining fuzzy matching using our OT-based construction over a reduced input domain. Building on observations from prior work~\cite{garimella2022structure, garimella2024computation, bui2025new}, spatial hashing effectively partitions the large input space into small cells. Specifically, for any sender point $\vecq$, any potentially $\delta$-close receiver point must reside within the neighboring cells of $\vecq$. Geometrically, this search space is bounded by a $d$-dimensional hypercube of side length $6\delta$ (as illustrated in Figure~\ref{Fig: Illustration of input domain reduction optimization}). By restricting operations to this reduced domain, as illustrated in Section~\ref{subsec: Asymptotically Efficient Fuzzy Matching}, our OT-based fuzzy matching achieves substantially better concrete performance than the OPRF-based alternative.}

{However, this domain reduction introduces a subtle risk of false positives. Consider a scenario where a cell contains a sender point $\vecq$, but none of the receiver's $\delta$-radius balls intersect this cell. As illustrated in the above dual-layer fuzzy mapping, the protocol might yield secret shares of a uniformly random point rather than a valid receiver point $\vecw$. Because the input domain is now artificially small, the probability that this random point accidentally falls within a distance of $\delta$ from $\vecq$ becomes non-negligible. This not only triggers false positives but also risks leaking information about the sender's private input.
To resolve this vulnerability, we integrate an additional consistency check. This mechanism explicitly verifies that the cell containing the sender's point $\vecq$ is identical to the cell intersected by the receiver's point $\vecw$. This structural guarantee ensures the receiver only learns $\vecq$ if the fuzzy matching succeeds \emph{and} both points legitimately map to the same cell. We defer the full technical details to Section~\ref{sec: Fuzzy PSI over Low-dimension Sets}.
}

{In summary, the concrete efficiency gains of our protocol over prior works stem from a synergistic combination of our dual-layer fuzzy PSI framework and our reduced-domain OT-based fuzzy matching.}
{
We emphasize that while our fuzzy PSI protocol still relies on OPPRF to process $\mathcal{O}(2^d n)$ items, it achieves substantially lower communication and computation overheads compared to prior OPPRF-based constructions~\cite{van2025}. 
% Generally, in OPPRF-based constructions, computational overhead is dominated by the number of queries, whereas communication bandwidth is dictated by the total bit-length of the programmed payloads. 
Our construction requires only $\mathcal{O}(m)$ OPPRF queries and transmits programmed payloads bounded by $\mathcal{O}(2^d n (d\log \delta + \lambda))$ bits due to reduced input domain. In contrast, prior works such as~\cite{van2025} suffer from a multiplicative factor $O(d \log \delta \cdot \lambda)$ in the $O(2^d n)$ invocations of fuzzy matching.
}

\section{Fuzzy Matching}

\label{sec: Fuzzy Matching}

\iffalse
\begin{figure}[!t]
\begin{nffunc}{\Func[ssFMatch]}

\noindent \textbf{Parameters:} Sender \SSS and receiver \RRR. Distance metric $\mathsf{dist}(\cdot,\cdot)$ and threshold $\delta$.

\noindent \textbf{Protocol:}

\begin{enumerate}

    \item Wait for input $\vecq \in \ZZ_{2^\ell}^{d}$ from \SSS and $\vecw \in \ZZ_{2^\ell}^{d}$ from \RRR.

    \item Output $[z]^1$ to \SSS and \RRR, where $z := 1\{\mathsf{dist}(\vecq, \vecw) \le \delta\}$.
    
\end{enumerate}

\end{nffunc}
% \vspace{0.5em}
\caption{Functionality of secret-shared fuzzy match.}
\label{Func:ssFMatch}
\end{figure}
\fi

% We propose two protocols for fuzzy matching. 

\begin{figure}[!t]
\begin{nfprot}{$\Pi_\mathsf{FMatch\text{-}oprf}^{L_\infty}$}

\noindent \textbf{Parameter:} Sender \SSS and receiver \RRR. Threshold $\delta$. $\ell_1 := \lambda + 2 \log(u^\prime+1) + \log d$ and $\ell_2 := \lambda$.
% $u, u^\prime= O(\log \delta)$.

\noindent \textbf{Input:} \SSS inputs $\vecq = (q_1, \ldots, q_d) \in \ZZ_{2^\ell}^{d}$ and \RRR inputs $\vecw = (w_1, \ldots, w_d) \in \ZZ_{2^\ell}^{d}$.

\noindent \textbf{Protocol:}

\begin{enumerate}
    
    \item For $k \in [d]$, \SSS computes $\{\tilde{q}_{k,t}\}_{t\in[u]} := \mathsf{PxTrie}(q_{k}-\delta, q_{k}+\delta)$, and \RRR computes $\{\tilde{w}_{k,t}\}_{t\in[u^\prime]} := \mathsf{PxPath}(w_{k}, \delta)$.
    
    \item For $k \in [d]$, \SSS samples $r_{k} \from \FF_{2^{\ell_1}}$. \SSS and \RRR invoke functionality \Func[OPPRF], where \SSS inputs $\{(k\|\tilde{q}_{k, t}, r_{k})\}_{k \in [d], t \in [u]}$ and learns $\OOO^F$, while \RRR inputs $\{k\|\tilde{w}_{k, t}\}_{k \in [d], t \in [u^\prime]}$ and learns $\{\hat{r}_{k, t} = F(k\|\tilde{w}_{k, t})\}_{k \in [d], t \in [u^\prime]}$.

    \item For $k \in [d]$, \RRR samples $\hat{s}_{k} \from \FF_{2^{\ell_2}}$.
    \SSS and \RRR invoke functionality \Func[OPPRF], where \RRR as the sender inputs $\{(k\|\hat{r}_{k, t}, \hat{s}_{k})\}_{k \in [d], t \in [u^\prime]}$  and learns $\OOO^{F^\prime}$, while \SSS as the receiver inputs $\{k\|r_{k}\}_{k \in [d]}$ and learns $\{s_{k} = F^\prime(k\|r_{k})\}_{k \in [d]}$.

    \item \SSS computes $s := \sum_{k \in [d]} s_{k} \in \FF_{2^{\ell_2}}$, while \RRR computes $\hat{s} := \sum_{k \in [d]} \hat{s}_{k} \in \FF_{2^{\ell_2}}$.

    \item \SSS and \RRR invoke functionality \Func[PEQT], where \SSS inputs ${s}$, while \RRR inputs $\hat{s}$ and learns $z \in \FF_2$. 

\end{enumerate}

\end{nfprot}
\caption{Fuzzy matching for $L_\infty$ distance from OPPRF}
\label{Prot: fmatch-infty-oprf}
\end{figure}

\subsection{Fuzzy Matching for $L_\infty$ Distance from Oblivious PRF}
\label{sec: Fuzzy Matching for L-inf Distance from Oblivious PRF}

Our first construction is based on OPRF. As the design intuition has already
been presented in Section~\ref{subsec: Asymptotically Efficient Fuzzy Matching}, we describe the detailed protocol in
Figure~\ref{Prot: fmatch-infty-oprf}. The key technical contribution is two role-reversed invocations of OPPRF, which enable communication overhead $O(\lambda d \log \delta)$ and computational overhead $O(d \log \delta)$.
We note that the use of OPPRF in our protocol can be further optimized by invoking
$d$ independent OPPRFs across different dimensions, each involving fewer
programmed points and queries. This optimization is formalized as
\emph{multi-batch OPPRF} in~\cite{van2025}, where batches of points are evaluated
under multiple independent random functions.

Extending the above protocol to support $L_p$ distance with the same asymptotic overhead remains challenging and is left as future work, since the use of prefix representations obscures precise distance information. 
We next state the security of the protocol and defer the detailed correctness
and security analysis to Appendix~\ref{proof: fuzzy-match-oprf}. The key point is to ensure that $\{\hat{r}_{k, t}\}_{t \in [u^\prime]}$ are different since they are the programmed points in the second OPPRF invocation, and that $r_k \notin \{\hat{r}_{k, t}\}_{t \in [u^\prime]}$ if $w_k \notin [q_k-\delta, q_k+\delta]$ to avoid false positives.

\begin{theorem}
\label{thm: fuzzy-match-oprf}
    The protocol $\Pi_\mathsf{FMatch\text{-}oprf}^{L_\infty}$ in Figure \ref{Prot: fmatch-infty-oprf} realizes the functionality \Func[FMatch] for $L_\infty$ distance in Figure \ref{Func:FMatch} against semi-honest adversaries in the $(\Func[OPPRF], \Func[PEQT])$-hybrid model.
\end{theorem}

\subsection{Fuzzy Matching for $L_\infty$ Distance from Oblivious Transfer}
\label{subsec: Fuzzy Matching for L-inf Distance from Oblivious Transfer}

\begin{figure}[!t]
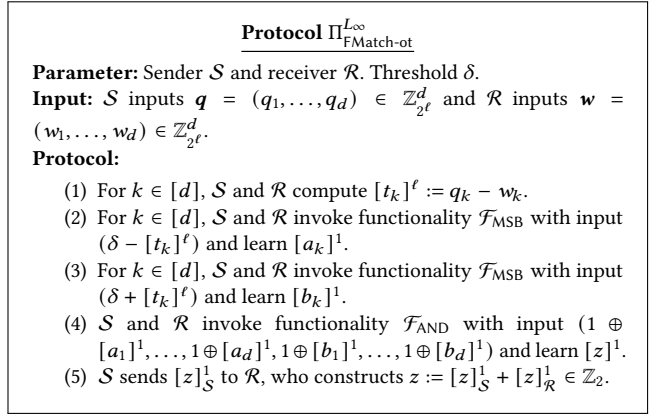

\begin{nfprot}{$\Pi_\mathsf{FMatch\text{-}ot}^{L_\infty}$}

\noindent \textbf{Parameter:} Sender \SSS and receiver \RRR. Threshold $\delta$.
% , $\ell_2:= \lceil\log (d \cdot (2^{\ell_1}-1)+1) \rceil$

\noindent \textbf{Input:} \SSS inputs $\vecq = (q_1, \ldots, q_d) \in \ZZ_{2^\ell}^{d}$ and \RRR inputs $\vecw = (w_1, \ldots, w_d) \in \ZZ_{2^\ell}^{d}$.

\noindent \textbf{Protocol:}

\begin{enumerate}
    \item For $k \in [d]$, \SSS and \RRR compute $[t_k]^{\ell} := q_k - w_k$.
    
    \item For $k \in [d]$, \SSS and \RRR invoke functionality $\Func[MSB]$ with input $(\delta - [t_k]^{\ell})$ and learn $[a_k]^{1}$.

    \item For $k \in [d]$, \SSS and \RRR invoke functionality $\Func[MSB]$ with input $(\delta + [t_k]^{\ell})$ and learn $[b_{k}]^{1}$.

    % \item For $k \in [2d-1]$, \SSS and \RRR invoke functionality \Func[AND] with input $([t_k]^{1}, [t_{k+1}]^{1})$ and learn $[t_{k+1}]^{1}$, and finally return $t_{2d}$ to \RRR.

    \item \SSS and \RRR invoke functionality \Func[AND] with input $(1\xor[a_1]^{1}, \ldots, 1\xor[a_d]^{1}, 1\xor[b_1]^{1}, \ldots, 1\xor[b_d]^{1})$ and learn $[z]^1$.

    \item \SSS sends $[z]^1_\SSS$ to \RRR, who constructs $z := [z]^1_\SSS + [z]^1_\RRR \in \ZZ_{2}$.
\end{enumerate}

\end{nfprot}
\caption{Fuzzy matching for $L_\infty$ distance from OT}
\label{Prot: fmatch-infty-2pc}
\end{figure}

Our second fuzzy matching construction is built on OT. To realize fuzzy matching, the goal is to securely evaluate the function $z = 1\{\mathsf{dist}_\infty(\vecq, \vecw) \le \delta\}$, where $\mathsf{dist}_\infty(\vecq, \vecw) = \max_{k \in [d]} |q_k - w_k|$.
Rather than naively computing this function via absolute-value (ABS)
and maximum (MAX) operations followed by a comparison with the threshold
$\delta$, we optimize the protocol by directly checking whether
$|q_k - w_k| \le \delta$ for every dimension $k \in [d]$, and then aggregating
the resulting comparison bits using AND operations. This approach reduces the overhead: the naive method requires $2d$ MSB operations and $2d$ Multiplexer operations to realize the ABS and MAX functionalities, whereas our optimized protocol only incurs $2d$ MSB operations together with lightweight AND operations.

% . We present the fuzzy matching protocol for $L_\infty$ distance in Figure \ref{Prot: fmatch-infty-2pc}.
Our protocol for $L_\infty$ distance is presented in Figure \ref{Prot: fmatch-infty-2pc}. For simplicity, we use $[x]^\ell$ to denote the parties' secret shares $[x]_\SSS^\ell$ and $[x]_\RRR^\ell$ of $x \in \ZZ_{2^\ell}$.
Specifically, our protocol relies on efficient and customized OT-based protocols for computing MSB and AND operations \cite{rathee2020cryptflow2, huang2022cheetah, rathee2021sirnn}.
We describe the
functionalities required for OT-based fuzzy matching and summarize the
corresponding communication overhead in
Appendix~\ref{appendix: Additional Functionalities for OT-based Fuzzy Matching}. Overall, our protocol incurs $O(d \ell)$ computational and communication overhead.

We extend OT-based fuzzy matching to arbitrary $L_p$ distance in Appendix \ref{appendix: OT-based Fuzzy Matching for $L_p$ Distance}.
Besides, OT-based fuzzy matching protocols are readily extended to realize the secret-shared fuzzy matching functionality in Figure~\ref{Func:FMatch} by omitting the last reconstruction step.
% and provide a comparison between our OPRF- and OT-based constructions in Appendix \ref{appendix: Comparisons between OPRF-based and OT-based Fuzzy Matching}.
We next state the security of the protocol and defer the detailed correctness
and security analysis to Appendix~\ref{proof: fuzzy-match-ot}. 

\begin{theorem}
\label{thm: fuzzy-match-ot}
    The protocol $\Pi_\mathsf{FMatch\text{-}ot}^{L_\infty}$ in Figure \ref{Prot: fmatch-infty-2pc} realizes the functionality \Func[FMatch] for $L_\infty$ distance in Figure \ref{Func:FMatch} against semi-honest adversaries in the $( \Func[MSB], \Func[AND])$-hybrid model.
\end{theorem}

\section{Fuzzy PSI over Low-dimensional Sets}
\label{sec: Fuzzy PSI over Low-dimension Sets}

\begin{figure}[!t]
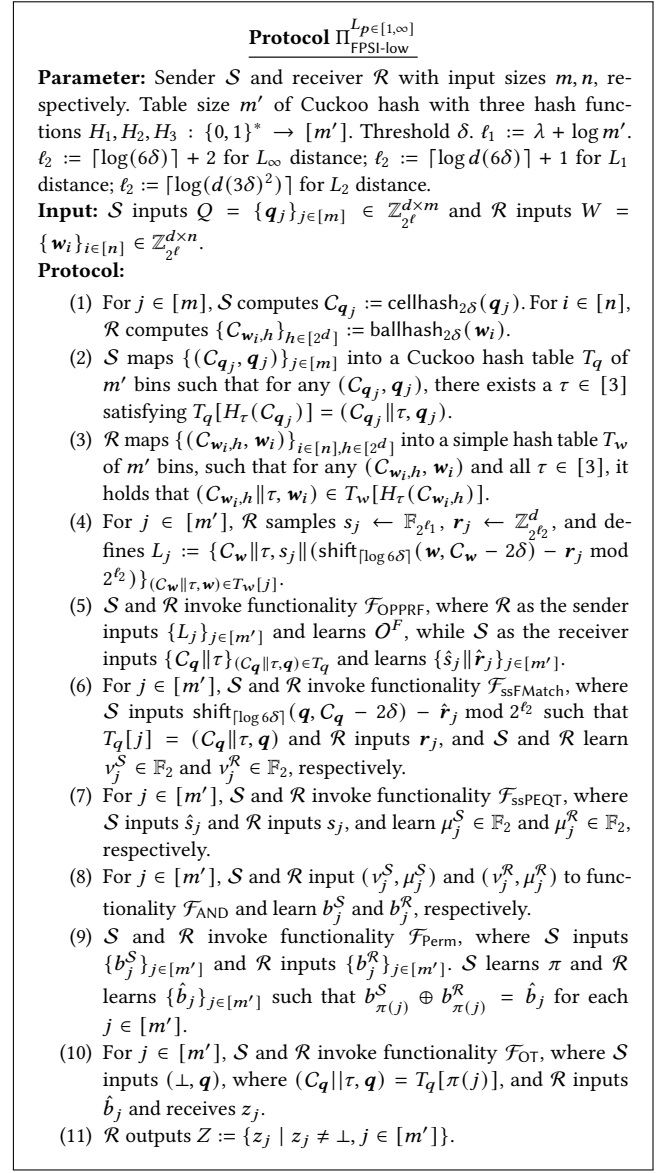

\begin{nfprot}{$\Pi_\mathsf{FPSI\text{-}low}^{L_{p\in[1, \infty]}}$}

\noindent \textbf{Parameter:} Sender \SSS and receiver \RRR with input sizes $m, n$, respectively. Table size $m^\prime$ of Cuckoo hash with three hash functions $H_1, H_2, H_3: \{0, 1\}^* \to [m^\prime]$.
Threshold $\delta$. $\ell_1:= \lambda + \log m^\prime$.
{$\ell_2 := \lceil \log (6\delta)\rceil + 2$ for $L_\infty$ distance; $\ell_2 := \lceil \log d(6\delta) \rceil + 1$ for $L_1$ distance; $\ell_2 := \lceil \log (d(3\delta)^2) \rceil$ for $L_2$ distance}. 
% $\ell_2 := \lceil \log (6\delta)\rceil + 2$ for $L_\infty$ distance and $\ell_2 := \lceil \log d (6\delta)^p \rceil + 1$ for $L_{p}$ distance.

% {$\ell_2 := \lceil \log (6\delta)\rceil + 2$ for $L_\infty$ distance; $\ell_2 := \lceil \log (6\delta d) \rceil + 1$ for $L_1$ distance; $\ell_2 := \lceil \log (d(3\delta)^2) \rceil$ for $L_2$ distance}. 

\noindent \textbf{Input:} \SSS inputs $Q = \{\vecq_j\}_{j \in [m]} \in \ZZ_{2^\ell}^{d\times m}$ and \RRR inputs $W = \{\vecw_i\}_{i \in [n]} \in \ZZ_{2^\ell}^{d\times n}$.

% \noindent \textbf{Communication/Computation:} (1) \SSS $\to$ \RRR: $dm\delta$ (DID) + $m$. (2) \RRR $\to$ \SSS: $dn$ (DID) + $n$.

\noindent \textbf{Protocol:}

\begin{enumerate}

        \item For $j \in [m]$, \SSS computes $\CCC_{\vecq_j} := \mathsf{cellhash}_{2\delta}(\vecq_j)$. For $i \in [n]$, \RRR computes $\{\CCC_{\vecw_i,h}\}_{h \in [2^d]} := \mathsf{ballhash}_{2\delta}(\vecw_i)$.

    \item \SSS maps $\{(\CCC_{\vecq_j}, \vecq_j)\}_{j \in [m]}$ into a Cuckoo hash table $T_q$ of $m^\prime$ bins such that for any $(\CCC_{\vecq_j}, \vecq_j)$, there exists a $\tau \in [3]$ satisfying $T_q[H_\tau(\CCC_{\vecq_j})] = (\CCC_{\vecq_j}\|\tau, \vecq_j)$. 
    % Define a function $\pi: [t] \rightarrow [m]$ that maps an element index of $X$ to the position in $T_X$, i.e., $\pi(i) = h_j(x_i)$ such that $T_X[h_j(x_i)] = x_i\|j$.

    \item \RRR maps $\{(\CCC_{\vecw_i, h}, \vecw_i)\}_{i \in [n], h \in [2^d]}$ into a simple hash table $T_w$ of $m^\prime$ bins, such that for any $(\CCC_{\vecw_i, h}, \vecw_i)$ and all $\tau \in [3]$, it holds that $(\CCC_{\vecw_i, h}\|\tau, \vecw_i) \in T_w[H_\tau(\CCC_{\vecw_i, h})]$.

    \item For $j \in [m^\prime]$, \RRR samples $s_{j} \from \FF_{2^{\ell_1}}$, $\vecr_{j} \from \ZZ_{2^{\ell_2}}^d$, and defines $L_j := \{\CCC_\vecw\|\tau, s_j \|(\mathsf{shift}_{{\lceil \log6\delta \rceil}}(\vecw, \CCC_\vecw - 2\delta)-\vecr_j \bmod 2^{\ell_2})\}_{(\CCC_\vecw\|\tau, \vecw) \in T_w[j]}$.
    % for all $(\CCC_\vecw\|\tau, \vecw) \in T_w[j]$. 
    % \meng{L1: u+2. L2: 2u+2}

    % \item \SSS defines $Q := \{\CCC_q\|\tau\}_{(\CCC_q\|\tau, q) \in T_q}$.

    \item \SSS and \RRR invoke functionality \Func[OPPRF], where \RRR as the sender inputs $\{L_j\}_{j \in [m^\prime]}$ and learns $\OOO^F$, while \SSS as the receiver inputs $\{\CCC_\vecq\|\tau\}_{(\CCC_\vecq\|\tau, \vecq) \in T_q}$ and learns $\{{\hat{s}_j} \|\hat{\vecr}_j\}_{j \in [m^\prime]}$. 
    % \meng{okvs bitlen: $\ell+du$}

    \item For $j \in [m^\prime]$, \SSS and \RRR invoke functionality $\FFF_\mathsf{ssFMatch}$, where \SSS inputs $\mathsf{shift}_{{\lceil \log6\delta \rceil}}(\vecq, \CCC_\vecq - 2\delta) - \hat{\vecr}_j \bmod 2^{\ell_2}$ such that $T_q[j] = (\CCC_{\vecq}\|\tau, \vecq)$ and \RRR inputs $\vecr_j$, and \SSS and \RRR learn $\nu^\SSS_j \in \FF_2$ and $\nu^\RRR_j \in \FF_2$, respectively. 

    \item For $j \in [m^\prime]$, \SSS and \RRR invoke functionality $\FFF_\mathsf{ssPEQT}$, where \SSS inputs $\hat{s}_j$ and \RRR inputs $s_j$, and learn $\mu^\SSS_j \in \FF_2$ and $\mu^\RRR_j \in \FF_2$, respectively.

    \item For $j \in [m^\prime]$, \SSS and \RRR input $(\nu^\SSS_j, \mu^\SSS_j)$ and $(\nu^\RRR_j, \mu^\RRR_j)$ to functionality $\FFF_\mathsf{AND}$ and learn $b^\SSS_j$ and $b^\RRR_j$, respectively. 
    
    % $\FFF_\mathsf{CNorm}^{\delta^p, \blue{2^u}}$, where \SSS inputs $(\blue{{g_j}}, \mathsf{shift}_{2^u}(\CCC_q, q), f_j)$ such that $T_q[j] = (\CCC_{q}\|\tau, q)$ and \RRR inputs $(\blue{{v_j}}, r_j)$. \SSS and \RRR learn $b^\SSS_j$ and $b^\RRR_j$, respectively, \blue{such that $b_j = 1\{ v_j = g_j \And \mathsf{dist}(q, r_j + f_j) \leq \delta\}$}. \meng{Secret-shared interval}

    \item \SSS and \RRR invoke functionality \Func[Perm], where \SSS inputs $\{b^\SSS_j\}_{j \in [m^\prime]}$ and \RRR inputs $\{b^\RRR_j\}_{j \in [m^\prime]}$. \SSS learns $\pi$ and \RRR learns $\{\hat{b}_j\}_{j \in [m^\prime]}$ such that $b^\SSS_{\pi(j)} \oplus b^\RRR_{\pi(j)} = \hat{b}_j$ for each $j \in [m^\prime]$. 
    
    \item For $j \in [m^\prime]$, \SSS and \RRR invoke functionality \Func[OT], where \SSS inputs $(\bot, \vecq)$, where $(\CCC_{\vecq}||\tau, \vecq) = T_q[\pi(j)]$, and \RRR inputs {$\hat{b}_j$} and receives $z_j$.

    \item \RRR outputs $Z := \{z_j \mid z_j \neq \bot, j \in [m^\prime]\}$.

\end{enumerate}

\end{nfprot}
\caption{Protocol of fuzzy PSI for $L_{p\in[1, \infty]}$ distance over low dimensional sets
}
% with \SSS's points $2\delta d^{1/p}$-apart and \RRR's points $2\delta(d^{1/p}+1)$-apart
\label{Prot: fpsi-p-low-px}
\end{figure}

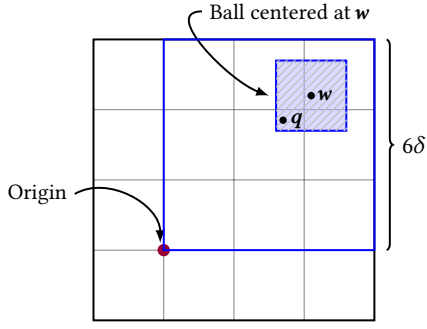
\begin{figure}[!t]
    \centering
    \resizebox{0.33\textwidth}{!}{
    \begin{tikzpicture}
        % 1. 绘制内部浅灰色网格 (0,0) 到 (4,4)
        \draw[step=1cm, black!50] (0,0) grid (4,4);

        % 2. 绘制最外层黑色边框
        \draw[thick, black] (0,0) rectangle (4,4);
        
        % 2. 绘制 Origin 标注
        \fill[purple!80!black] (1,1) circle (2.5pt); 
        \node (origin_text) at (-0.8, 1.8) { Origin};
        % 修改后的箭头：从文字下方(south)出发，平缓指向原点左侧(-0.1, 1)
        % out=-90 表示向下出发，in=180 表示从左侧水平进入
        % 箭头从文字右方 (east) 出发，平缓指向原点 (0,1) 的左侧
        % out=20 让它稍微向上扬起出发，in=160 从左上方平滑切入点边缘
        % 修改后的箭头：从文字右方(east)出发，平缓指向点(0,1)的上侧(0, 1.15)
        % out=30 稍微向上出发，in=90 从正上方垂直进入，形成平滑的抛物线弧度
        \draw[-latex, thick] (origin_text.east) to[out=30, in=90] (0.95, 1.10);
        
        % 3. 绘制右侧的 6\delta 括号
        \draw[thick, decoration={brace, mirror, raise=5pt}, decorate] (4,1) -- (4,4);
        \node[right=8pt] at (4, 2.5) {\large $6\delta$};
        
        % 4. 绘制蓝色带阴影的方块 (Ball centered at w)
        % 定义位置
        \coordinate (w) at (3.1, 3.2);
        \coordinate (q) at (2.7, 2.85);
        
        % 绘制蓝色方块并填充阴影
        \filldraw[fill=blue!30, fill opacity=0.5, draw=blue, thick] (2.6, 2.7) rectangle (3.6, 3.7);
        % 绘制阴影斜线 (需要 patterns 库)
        \fill[pattern=north east lines, pattern color=gray!50] (2.6, 2.7) rectangle (3.6, 3.7);
        
        % 5. 绘制点 w 和 q
        % \fill (w) circle (1.5pt) node[right, yshift=2pt] { $w$};
        \fill (w) circle (1.5pt) node[right, opacity=1] {$\vecw$};
        % \fill (q) circle (1.5pt) node[right, yshift=2pt] { $q$};
        \fill (q) circle (1.5pt) node[right, opacity=1] {$\vecq$};
        
        % 6. 绘制 "Ball centered at w" 的标注
        \node (ball_text) at (2.8, 4.4) {Ball centered at $\vecw$};
        % 使用 out/in 控制的弯曲箭头
        \draw[-latex, thick] (ball_text.west) to[out=180, in=160] (2.5, 3.2);
        
        % 7. 额外加粗蓝色边界（模仿手绘中的蓝色笔迹）
        \draw[blue, thick] (1,1) -- (4,1) -- (4,4) -- (1,4) -- cycle;
    \end{tikzpicture}
    }
\caption{Illustration of reduced input domain after shifting}
\label{Fig: Illustration of input domain reduction optimization}
\end{figure}

\subsection{Constructions}
\label{subsec: low dim Constructions}

% For realizing multi-point fuzzy PSI, similar to recent works \blue{refs}, we leverage spatial hashing techniques.
% Spatial hashing 

We present our fuzzy PSI protocol for $L_{p \in [1,\infty]}$ distance over low-dimensional sets in Figure~\ref{Prot: fpsi-p-low-px}. The protocol relies on assumptions on both the sender’s and receiver’s sets, which are consistent with those in~\cite{van2025} and the evaluation setup of~\cite{dang2025ccs}.
That is, for the sender’s set, each grid cell contains at most one sender's point; for the receiver’s set, each grid cell intersects with at most one $\delta$-radius ball centered at a receiver point.
We discuss how to extend to arbitrary receiver sets in Section \ref{subsec: Extension to Arbitrary Receiver's Set}.

Since the protocol intuition has already been presented in
Section~\ref{Sec: Technical Overview}, we focus here on providing additional details of the domain-reduction optimization used when invoking fuzzy matching.

% for each sender's point $\vecq$, each possible point $\vecw$ must be contained in the neighborhood cells of the $\vecq$
% This neighborhood is precisely a $d$-dimensional cube of side length $6\delta$.

% With the above alignment, we invoke fuzzy matching protocols in a black-box manner for each aligned pair. 
% As shwon in recent works, 
% We utilize this fact to reduce the domain of fuzzy matching.

% each input ball of fixed radius $\delta$ would be entirely contained in the neighborhood of its center.
% Therefore, we reduce the domain of fuzzy matching to size-$6\delta$ cube.

% With these neighbor cells, the input bit length is reduced to $\lceil \log 6\delta \rceil$. 
% However, we can not directly truncate the input to its last $\lceil \log 6\delta \rceil$ bits, since truncation may lead to incorrect distances. 
% For example, given two values 0100 0000 and 0011 1111 that differ by 1, after truncating to the last 4 bits, they become 0000 and 1111, respectively, and have a large distance of 15.

% and is $\CCC_{\vecw, h} - 2\delta$ for each $h \in [2^d]$.

\textbf{Reducing input domain of fuzzy matching.} For a sender point $\vecq$, any receiver point that is potentially $\delta$-close to $\vecq$ must lie in one of the neighboring grid cells of $\vecq$.
Formally, the neighboring cells of $\vecq$ are defined as
$\{(\CCC_{\vecq,1}+t_1, \ldots, \CCC_{\vecq,d}+t_d)\}_{t_i \in \{-1,0,1\}}$, which collectively form a $d$-dimensional hypercube with side length $6\delta$.
This significantly reduces the effective input domain for fuzzy matching.
To execute fuzzy matching over this reduced domain, we define a
\emph{shift-and-truncate} operation as follows:
\[
\mathsf{shift}_{\ell}(\vecq, \veco) := (\vecq - \veco) \bmod 2^\ell,
\]
where $\veco$ denotes the shift origin, and $\ell$ is the required bit length to represent points in the reduced domain.
In our protocol, for each sender point $\vecq$, the shift origin $\veco$ is chosen
as the lower-left corner of the neighboring cells of $\vecq$, which can be
formally expressed as $\CCC_{\vecq} - 2\delta$.
{Figure~\ref{Fig: Illustration of input domain reduction optimization} provides an illustrative example.}
We note that, for a receiver point $\vecw$, the shift origin is determined by the
cell intersected by the $\delta$-radius ball centered at $\vecw$, rather than by
the cell containing $\vecw$ itself.

% \textbf{Addressing false positives caused by samll input domain.} While the overhead of fuzzy matching can be reduced with the above input domain reduction, as analyzed in Section \ref{Sec: Technical Overview}, it introduces false positives since the input points pair to fuzzy matching could be some $\vecq$ from the sender and a random point, instead of a valid $\vecw$ from the receiver.
% Due to the small input space, the probability that these two points are $\delta$-close is non-negligible, resulting in incorrectly revealing $\vecq$ to the receiver.
% To address this issue, we let the receiver additionally append a random $s$ in the programmed point's output of OPPRF.
% As a result, the sender only obtains the same $s$ if both parties' inputs share the same cell identifier.
% Then, an additional equality test is invoked along with fuzzy matching, and the result indicates whether the cell of $\vecq$ is the same as the cell that $\vecw$ intersects with.
% This ensures that the receiver obtains $\vecq$ only when the fuzzy matching of the pair succeeds \textit{and} the two cells are the same.

\textbf{Addressing false positives caused by small input domain.}
Although the above domain-reduction technique significantly reduces the overhead of fuzzy matching, as discussed in
Section~\ref{subsec: New Framework for Fuzzy PSI}, it introduces false positives.
In particular, after executing our OPPRF-based fuzzy mapping, the produced pair of points may consist of a sender point $\vecq$ and a uniformly random point, rather than a valid receiver point $\vecw$.
Because the input domain is small, the probability that these two points are $\delta$-close is non-negligible, which may incorrectly reveal $\vecq$ to the receiver.

To mitigate this issue, the receiver appends an additional random value $s$ to the output of each programmed point in the OPPRF. Consequently, the sender obtains the same value $s$ if and only if both parties’ inputs correspond to the same cell identifier. We then perform an additional equality test in conjunction with fuzzy matching, which verifies whether the cell containing $\vecq$ is identical to the cell intersected by $\vecw$. This ensures that the receiver learns $\vecq$ only when the fuzzy matching succeeds \emph{and} the two points are associated with the same cell.

Besides, the ordering of Cuckoo hashing leaks additional information about the sender’s entire set. 
To prevent the leakage, we adopt an oblivious permutation on the fuzzy matching results.
% \blue{
% We note that while our fuzzy PSI protocol still invokes OPPRF, it achieves low communication and computation overhead compared to prior OPPRF-based fuzzy PSI~\cite{van2025}. The reason is that the computation of OPPRF is dominated by the number of queries, and the communicaiton is dominated by the total bit lengths of programmed payloads.
% Our protocol involves only a small number $O(m)$ of OPPRF queries and $O(n 2^d (d\log \delta + \lambda)$-bit programmed payloads.
% In contrast, prior OPPRF-based works like \cite{van2025} have a communication overhead of $O(n 2^d (d\log \delta \lambda)$.
% }
% \blue{
% We emphasize that while our fuzzy PSI protocol still relies on OPPRF, it achieves substantially lower communication and computation overheads compared to prior OPPRF-based constructions~\cite{van2025}. 
% % Generally, in OPPRF-based constructions, computational overhead is dominated by the number of queries, whereas communication bandwidth is dictated by the total bit-length of the programmed payloads. 
% Our construction requires only $\mathcal{O}(m)$ OPPRF queries and transmits payloads bounded by $\mathcal{O}(n 2^d (d\log \delta + \lambda))$ bits. In contrast, prior works such as~\cite{van2025} suffer from a multiplicative dependence on $\lambda$, resulting in a heavier communication overhead of $\mathcal{O}(n 2^d \cdot d\log \delta \cdot \lambda)$.
% }
% , ensuring that the receiver does not obtain information about the sender’s Cuckoo hashing index.
The security of the protocol is shown as follows, and we defer the detailed correctness and security analysis to Appendix~\ref{Proof of Theorem fuzzy PSI L-p-low-px}.

\begin{theorem}
\label{thm: fuzzy PSI L-p-low-px}
    The protocol $\Pi_\mathsf{FPSI\text{-}low}^{L_{p\in[1, \infty]}}$ in Figure \ref{Prot: fpsi-p-low-px} realizes the functionality \Func[FPSI] for $L_{p\in[1, \infty]}$ distance in Figure \ref{Func:FPSI} against semi-honest adversaries in the $(\Func[OPPRF], \Func[ssFMatch], \Func[ssPEQT], \Func[Perm], \Func[AND], \Func[OT])$-hybrid model, if the sender's (resp. receiver's) set satisfies the unique center (resp. ball) assumption.
\end{theorem}

\begin{figure}[!t]
\begin{nfprot}{$\Pi_\mathsf{FPSI\text{-}high}^{L_\infty}$}

\noindent \textbf{Parameter:} Sender \SSS and receiver \RRR with input sizes $m, n$, respectively. Threshold $\delta$. Bit lengths $\ell_1 := \lambda + 2 \log(m+nu^\prime) + \log d$, $\ell_2:= \lambda + \log m + \log n$.

\noindent \textbf{Input:} \SSS inputs $Q = \{\vecq_j\}_{j \in [m]} \in \ZZ_{2^\ell}^{d\times m}$ and \RRR inputs $W = \{\vecw_i\}_{i \in [n]} \in \ZZ_{2^\ell}^{d\times n}$.

% \noindent \textbf{Communication/Computation:} (1) \SSS $\to$ \RRR: $dm\delta$ (DID) + $m$. (2) \RRR $\to$ \SSS: $dn$ (DID) + $n$.

\noindent \textbf{Protocol:}

\begin{enumerate}

    % such that $\sum_{k \in [d]} r_{j, k} = 0$ for $j \in [m]$.

    \item For $j \in [m], k \in [d]$, \SSS computes $\{\tilde{q}_{j,k,t}\}_{t\in[u]} := \mathsf{PxTrie}(q_{j,k}-\delta, q_{j,k}+\delta)$.

    \item For $i \in [n], k \in [d]$, \RRR computes $\{\tilde{w}_{i,k,t}\}_{t\in[u^\prime]} := \mathsf{PxPath}(w_{i,k}, \delta)$.
    
    \item \SSS samples $\{r_{j, k}\} \from \FF_{2^{\ell_1}}$ for $j \in [m], k \in [d]$. \SSS and \RRR invoke functionality \Func[OPPRF], where \SSS inputs $\{(k\|\tilde{q}_{j, k, t}, r_{j, k})\}_{j \in [m], k \in [d], t \in [u]}$ and learns $\OOO^F$, while \RRR inputs $\{k\|\tilde{w}_{i, k, t}\}_{i \in [n], k \in [d], t \in [u^\prime]}$ and learns $\{\hat{r}_{i, k, t} = F(k\|\tilde{w}_{i, k, t})\}_{i \in [n], k \in [d], t \in [u^\prime]}$.

    \item \RRR samples ${\hat{s}_{i, k}}  \from \FF_{2^{\ell_2}}$ for $i \in [n], k \in [d]$.
    \SSS and \RRR invoke functionality \Func[OPPRF], where \RRR as the sender inputs $\{(k\|\hat{r}_{i, k, t}, \hat{s}_{i, k})\}_{i \in [n], k \in [d], t \in [u^\prime]}$ and learns $\OOO^{F^\prime}$, while \SSS as the receiver inputs $\{k\|r_{j, k}\}_{j \in [m], k \in [d]}$ and learns $\{s_{j, k} := F^\prime(k\|r_{j, k})\}_{j \in [m], k \in [d]}$.

    \item For $j \in [m]$, \SSS computes ${s}_{j} := \sum_{k \in [d]} {s}_{j, k}$. For $i \in [n]$, \RRR computes $\hat{s}_{i} := \sum_{k \in [d]} \hat{s}_{i, k}$.

    \item \SSS and \RRR invoke functionality \Func[wLPSI], where \SSS inputs $\{(s_j, \vecq_j)\}_{j \in [m]}$, while \RRR inputs $\{\hat{s}_i\}_{i \in [n]}$ and learns $I$. 

    \item \RRR outputs $Z := \{\vecz \mid (b, \vecz) \in I\}$.

\end{enumerate}

\end{nfprot}
\caption{Protocol of fuzzy PSI for $L_\infty$ distance over high dimensional sets}
% with \SSS and \RRR globally disjoint sets}
\label{Prot: fpsi-infty-high-px}
\end{figure}

\begin{figure}[!t]
\begin{nfprot}{$\Pi_\mathsf{FPSI\text{-}high}^{L_p}$}

\noindent \textbf{Parameter:} Sender \SSS and receiver \RRR with input sizes $m, n$, respectively. Threshold $\delta$. Table size $n^\prime$ of Cuckoo hash with three hash functions $H_1, H_2, H_3: \{0, 1\}^* \to [n^\prime]$. Bit length $\ell^\prime:= \lambda + p\log \delta + \log n^\prime$.

% $\ell^\prime_1 := \lambda + 2 \log(m+nu^\prime) + \log d$, $\ell^\prime_2:= \lambda + \log m + \log n$.

\noindent \textbf{Input:} \SSS inputs $Q = \{\vecq_j\}_{j \in [m]} \in \ZZ_{2^\ell}^{d\times m}$ and \RRR inputs $W = \{\vecw_i\}_{i \in [n]} \in \ZZ_{2^\ell}^{d\times n}$.

% \noindent \textbf{Communication/Computation:} (1) \SSS $\to$ \RRR: $dm\delta$ (DID) + $m$. (2) \RRR $\to$ \SSS: $dn$ (DID) + $n$.

\noindent \textbf{Protocol:}

\begin{enumerate}

    \item \SSS and \RRR execute the steps 1-5 of protocol $\Pi_\mathsf{FPSI\text{-}high}^{L_\infty}$, and learn $\{\mathsf{id}_{\vecq_j}\}_{j \in [m]}$ and $\{\mathsf{id}_{\vecw_i}\}_{i \in [n]}$, respectively.

    \item \SSS maps $\{(\mathsf{id}_{\vecq_j}, \vecq_j)\}_{j \in [m]}$ into a simple hash table $T_q$ of $n^\prime$ bins, such that for any $(\mathsf{id}_{\vecq_j}, \vecq_j)$ and all $\tau \in [3]$, it holds that $(\mathsf{id}_{\vecq_j}\|\tau, \vecq_j) \in T_q[H_\tau(\mathsf{id}_{\vecq_j})]$.

    \item \RRR maps $\{(\mathsf{id}_{\vecw_i}, \vecw_i)\}_{i \in [n]}$ into a Cuckoo hash table $T_w$ of $n^\prime$ bins such that for any $(\mathsf{id}_{\vecw_i}, \vecw_i)$, there exists $\tau \in [3]$ satisfying $T_w[H_\tau(\mathsf{id}_{\vecw_i})] = (\mathsf{id}_{\vecw_i}\|\tau, \vecw_i)$. 

    \item \SSS samples $\{\vecg_{i}\}_{i \in [n^\prime]} \from \ZZ_{2^{\ell^\prime}}^d$. For $i \in [n^\prime]$, \SSS defines $L_i := \{\mathsf{id}_\vecq\|\tau, \vecq-\vecg_i \bmod 2^{\ell^\prime}\}$ for all $(\mathsf{id}_\vecq\|\tau, \vecq) \in T_q[i]$.

    \item \SSS and \RRR invoke functionality \Func[OPPRF], where \SSS inputs $\{L_i\}_{i \in [n^\prime]}$ and learns $\OOO^G$, while \RRR inputs $\{\mathsf{id}_\vecw\|\tau\}_{(\mathsf{id}_\vecw\|\tau, \vecw) \in T_w}$ and learns $\{\hat{\vecg}_i := G(\mathsf{id}_\vecw\|\tau)\}_{i \in [n^\prime]}$.

    \item For $i \in [n^\prime]$, \SSS and \RRR invoke functionality $\FFF_\mathsf{FMatch}$, where \SSS inputs $\vecg_i$ and \RRR inputs $\vecw - \hat{\vecg}_i \bmod 2^{\ell^\prime}$ such that $T_w[i] = (\mathsf{id}_{\vecw}\|\tau, \vecw)$ and learns $b_i \in \FF_2$.

    \item For $i \in [n^\prime]$, \SSS and \RRR invoke functionality \Func[OT], where \SSS inputs $(\bot, \vecg_i)$, and \RRR inputs $b_i$ and receives $\vecz_i$.

    \item \RRR outputs $Z := \{\vecz_i+\hat{\vecg}_i \mid z_i \neq \bot, i \in [n^\prime]\}$.

\end{enumerate}

\end{nfprot}
\caption{Protocol of fuzzy PSI for $L_p$ distance over high dimensional sets}
% with \SSS and \RRR globally disjoint assumption
\label{Prot: fpsi-p-high-px}
\end{figure}

\subsection{Extension to Arbitrary Receiver's Set}
\label{subsec: Extension to Arbitrary Receiver's Set}

% Similar to the work \cite{piske2025distance}, we can extend our protocols to arbitrary receiver's set distribution without any assumption.
% We define a density parameter $\rho$, which records the maximum number of balls of the receiver's set that intersect with a cell.
% We need to correspondingly adapt the spatial hashing such that the cell identifier of each receiver's point is unique.
% Specifically, we add a counter after the original cell identifier for each receiver's ball. 
% Formally, we redefine $\mathsf{ballhash}(\vecw, \rho) = \{\CCC_{\vecw, h}\|D[\CCC_{\vecw, h}]++\}_{h \in [2^d]}$, where $D$ is a dictionary and we initialize $D[\CCC] := 1$ for all cell $\CCC$.
% Correspondingly, we record all possible cell identifiers for each sender's point as follows: 
% $\mathsf{cellhash}(\vecq, \rho) = \{\CCC_\vecq\|h\}_{h \in [\rho]}$.
% The remaining part of fuzzy PSI protocol is the same as that of Figure \ref{Prot: fpsi-p-low-px}.
% Same as the work \cite{piske2025distance}, our original protocol assumes $\rho = 1$. For $\rho > 1$, our protocol complexity has an additional multiplicative factor $\rho$. If there is no limitation on $\rho$, our protocol realizes fuzzy PSI on arbitrary receiver's set with the only unique cell assumption on the sender's set.

Similar to~\cite{piske2025distance}, we can extend our protocols to support arbitrary receiver set distributions without imposing any structural assumptions.
To this end, we define a \emph{density parameter} $\rho$, which denotes the maximum number of receiver balls that intersect the same cell.
We then adapt the spatial hashing scheme so that each receiver point is assigned a unique cell identifier.
Specifically, we append a counter to the original cell identifier for each receiver ball.
Formally, we redefine $\mathsf{ballhash}(\vecw, \rho) = \{\CCC_{\vecw, h}\|D[\CCC_{\vecw, h}]++\}_{h \in [2^d]}$,
where $D$ is a dictionary initialized as $D[\CCC] := 1$ for every cell identifier $\CCC$.
This ensures that even when multiple receiver balls map to the same cell, their identifiers remain unique.
Correspondingly, for each sender point $\vecq$, we enumerate all possible matching cell identifiers as $\mathsf{cellhash}(\vecq, \rho) = \{\CCC_\vecq\|h\}_{h \in [\rho]}$.
The remainder of the fuzzy PSI protocol remains identical to that shown in Figure~\ref{Prot: fpsi-p-low-px}.
As in~\cite{piske2025distance}, our original protocol assumes $\rho = 1$.
When $\rho > 1$, the complexity of our protocol increases by an additional multiplicative factor of $\rho$.
Therefore, our construction supports fuzzy PSI for arbitrary receiver sets, requiring only the unique cell assumption on the sender’s set.

\section{Fuzzy PSI over High-dimensional Sets}
\label{Sec: Fuzzy PSI over High-dimensional Sets}

% \blue{In this section, we present efficient fuzzy PSI protocols for high-dimensional sets.
% In this setting, the protocols in Section~\ref{sec: Fuzzy PSI over Low-dimension Sets} incur large overhead due to the factor $2^d$, which is a fundamental limitation of most spatial hashing–based approaches~\cite{garimella2022structure, garimella2024computation, bui2025new, richardson2024fuzzy, van2024fuzzy, van2025, piske2025distance}.
% To overcome this, we utilize the globally disjoint assumption on both parties, also used in \cite{garimella2022structure, van2025}. That is, the projections of the sender’s and receiver’s balls are globally disjoint in every dimension.
% We acknowledge that this assumption is rather strong, but this benefits the exploration of asymptotically efficient fuzzy PSI with $O((m+n)d\log \delta)$ computation and communication overhead.}

{
In this section, we present efficient fuzzy PSI protocols tailored for high-dimensional sets. In this setting, the protocols introduced in Section~\ref{sec: Fuzzy PSI over Low-dimension Sets} suffer from an exponential $2^d$ blowup in overhead, which is a fundamental bottleneck inherent to almost all spatial hashing-based approaches~\cite{garimella2022structure, garimella2024computation, bui2025new, richardson2024fuzzy, van2024fuzzy, van2025, piske2025distance}. To bypass this curse of dimensionality, we leverage the \textit{globally disjoint} assumption, previously adopted in~\cite{garimella2022structure, van2025}. That is, the projections of the sender’s and receiver’s balls are globally disjoint in every dimension. While our protocols achieve desirable $\mathcal{O}((m+n)d\log \delta)$ computation and communication complexity, we acknowledge that this assumption enforces a strong structural constraint and is not entirely realistic.
}

% As illustrated in \cite{garimella2022structure, van2025}, 

% we adopt the same assumption as in~\cite{van2025}, namely that the projections of both the sender’s and receiver’s balls are globally disjoint across all dimensions.
% As illustrated in~\cite{van2025}, this assumption is rather strong, \blue{and we discuss how to extend our protocols to the milder assumption in Section~\ref{subsec: Extension to Mild Assumptions}.}
% Leveraging our asymptotically efficient fuzzy matching, the resulting fuzzy PSI protocols achieve a logarithmic dependence on the distance threshold~$\delta$, an exponential improvement over the state-of-the-art protocol in~\cite{van2025}.

% Below, we separately present the constructions for $L_\infty$ and $L_p$ distances.

% For this setting, the protocols from the previous sections have a prohibitively large overhead due to the factor $2^d$. 
% This factor exists in all existing works based on spatial hashing \cite{garimella2022structure, garimella2024computation, bui2025new, van2024fuzzy, van2025, richardson2024fuzzy, piske2025distance, dang2025ccs}. 
% Our protocols have $O((m+n)d\log \delta)$ computation and communication overhead, an exponential improvement on $\delta$ over the state-of-the-art work \cite{van2025} with $(m\delta+n)d$ overhead.

%\subsection{Constructions}

% We utilize the same globally disjoint assumption as used in the recent works \cite{garimella2022structure, van2025}. That is, the projection of the sender’s and receiver’s balls are globally disjoint in every dimension.

We formally present the globally disjoint assumption as follows.

\begin{definition}[Globally Disjoint Assumption \cite{garimella2022structure, van2025}]
A set of $d$-dimensional balls of radius $\delta$ with centers $\vecq_1, \ldots, \vecq_m$ is globally disjoint if, for each dimension $k \in [d]$, the projections $[q_{j, k} - \delta, q_{j, k} + \delta]$ are disjoint
for all $j \in [m]$.
\end{definition}

\textbf{Construction for $L_\infty$ distance.} 
{
We compile our OPRF-based fuzzy matching protocol to construct fuzzy PSI for $L_\infty$ distance.
We do not adopt our OT-based construction, since without spatial hashing it does not admit a sufficiently reduced input domain, leading to large overhead, especially for inputs with large bit length.
The idea is that since both parties' sets are globally disjoint, fuzzy PSI protocols can be constructed by using multiple invocations of disjoint OPRF-based fuzzy matching. 
We present here the protocol in Figure~\ref{Prot: fpsi-infty-high-px}, and defer the detailed description to Appendix~\ref{Appendix: Details of Fuzzy PSI for High-dimensional Sets}.
}

The security is shown in the following, and its proof is deferred to Appendix \ref{proof: Proof of fuzzy PSI L-inf-high-px}.

% With this assumption, we utilize our OPRF-based fuzzy matching in a non-black box to construct fuzzy PSI protocols.
% Recall that in our OPRF-based fuzzy matching protocol, for each dimension $k$, the OPPRF protocol programs (the prefix representation of) the interval $[q_{k} - \delta, q_{k} + \delta]$ to the same random value $r_k$. 
% Once querying two points $w_{k}, w_{k}^\prime$ that belong to $[q_{k} - \delta, q_{k} + \delta]$, the receiver will obtain the same $r_k$.
% This leaks more information about partial matches on individual dimension. Fortunately, the globally disjoint assumption on the receiver side ensures that the above issue will not occur, since for any two points $\vecw, \vecw^\prime$, it holds that $w_{k}, w_{k}^\prime$ are at least $2\delta$ apart for all $k \in[d]$.
% Furthermore, the globally disjoint assumption on the sender side ensures that for any $\vecq, \vecq^\prime$, $[q_{k} - \delta, q_{k} + \delta]$ and $[q_{k}^\prime - \delta, q_{k}^\prime + \delta]$ are disjoint for all $k \in[d]$.
% This enables the sender to program all intervals in one OPPRF.
% Therefore, our fuzzy PSI protocol incurs a factor of $O(m+n)$ overhead increase compared to our OPRF-based fuzzy matching protocol.

\begin{theorem}
\label{thm: fuzzy PSI L-inf-high-px}
    The protocol $\Pi_\mathsf{FPSI\text{-}high}^{L_\infty}$ in Figure \ref{Prot: fpsi-infty-high-px} realizes the functionality \Func[FPSI] for $L_\infty$ distance in Figure \ref{Func:FPSI} against semi-honest adversaries in the $(\Func[OPPRF], \Func[wLPSI])$-hybrid model if both parties' sets satisfy the globally disjoint assumption.
\end{theorem}

% \textbf{Security issue of \cite{van2025}.} \blue{Give a figure}

% \textbf{Construction for $L_p$ distance.} \blue{
% We further present the fuzzy PSI protocol on high-dimensional sets for $L_p$ distance.
% Our main observation is that the above protocol for $L_\infty$ distance produces unique identifiers for each point prior to output, namely $s_j$ for $\vecq_j$ and $\hat{s}_i$ for $\vecw_i$.
% For clarify, we denote these identifiers as $\mathsf{id}_{\vecw_i}$ and $ \mathsf{id}_{\vecq_j}$, which satisfy $\mathsf{id}_{\vecw_i} = \mathsf{id}_{\vecq_j}$ whenever $\mathsf{dist}_\infty(\vecq_j, \vecw_i) \le \delta$. 
% The parties must further verify whether $\mathsf{dist}_p(\vecq_j, \vecw_i) \le \delta$ for those point pairs that share the same identifier.
% We achieve this by reusing the Cuckoo hashing–based design as in our fuzzy PSI for low-dimensional sets, in which, instead of using grid cells as identifiers
% $\mathsf{id}_{\vecq_j}$ and $\mathsf{id}_{\vecw_i}$, we use identifiers obtained from the $L_\infty$-based fuzzy PSI in Figure~\ref{Prot: fpsi-infty-high-px}.
% We present the detailed protocol in Figure~\ref{Prot: fpsi-p-high-px} and defer the detailed description (including customized input domain reduction optimization) to Appendix~\ref{Appendix: Details of Fuzzy PSI for High-dimensional Sets}.
% }

\textbf{Construction for $L_p$ Distance.} 
{
We further present the fuzzy PSI protocol on high-dimensional sets for $L_p$ distance based on both OPRF- and OT-based fuzzy matching. Our core insight is that the $L_\infty$ protocol described above inherently generates unique identifiers prior to its final output, specifically, $s_j$ for each sender point $\vecq_j$ and $\hat{s}_i$ for each receiver point $\vecw_i$. For clarity, we relabel these intermediate values as $\mathsf{id}_{\vecq_j}$ and $\mathsf{id}_{\vecw_i}$. By design, a collision $\mathsf{id}_{\vecq_j} = \mathsf{id}_{\vecw_i}$ occurs, if and only if $\mathsf{dist}_\infty(\vecq_j, \vecw_i) \le \delta$.
Because the $L_\infty$ ball strictly bounds the $L_p$ ball, these identifiers serve as an effective coarse filter. The parties only need to perform the rigorous $\mathsf{dist}_p(\vecq_j, \vecw_i) \le \delta$ check on candidate pairs that share an identical identifier. To execute this efficiently, we repurpose the Cuckoo hashing-based design from our low-dimensional protocol. However, instead of using spatial grid cells, we use the above intermediate values from the $L_\infty$ protocol in Figure~\ref{Prot: fpsi-infty-high-px} as identifiers. The formal protocol is presented in Figure~\ref{Prot: fpsi-p-high-px}, and additional descriptions are deferred to Appendix~\ref{Appendix: Details of Fuzzy PSI for High-dimensional Sets}.
}

% such that 
% Since $\mathsf{dist}_p(\vecq_j, \vecw_i) \ge \mathsf{dist}_\infty(\vecq_j, \vecw_i)$,
% % the parties must further verify whether $\mathsf{dist}_p(\vecq_j, \vecw_i) \le \delta$ for those point pairs that share the same identifier.
% To this end, we follow the same Cuckoo hashing–based design as in our fuzzy PSI for low-dimensional sets.

% Moreover, we further optimize the bit length to reduce the overhead of our OT-based fuzzy matching.

% \textbf{Construction for $L_p$ Distance.} In Figure \ref{Prot: fpsi-p-high-px}, we extend our fuzzy PSI protocol to $L_p$ distance over high-dimensional sets. 
% We can further optimize the bit length when invoking OT-based fuzzy matching. Our main observation is that for points $\vecq_j$ and $\vecw_i$, the same identifier $\mathsf{id}_{\vecq_j} = \mathsf{id}_{\vecw_i}$ means that $\mathsf{dist}_{\infty}(\vecw_i, \vecq_j) \le \delta$ and hence $|w_{i,k} - q_{j,k}| \le \delta$ for all dimension $k \in [d]$.
% As shown in our OT-based fuzzy matching protocol in Figure \ref{Prot: fmatch-p-2pc}, the first step is to compute $t_{k} = w_{i,k} - q_{j,k}$. Therefore, we can use small bit length to represent the secret shares of $t_{k}$ when executing the remaining operations that are dependent on the bit length of secret sharing.

The security analysis is shown in the following, and its proof is deferred
to Appendix \ref{proof: fuzzy PSI L-p-high-px}.

% The protocol combines both our OPRF-based and OT-based fuzzy matching.
% Our main observation is that after executing the steps 1-5 of the $L_\infty$ fuzzy PSI in Figure \ref{Prot: fpsi-infty-high-px}, $\{\hat{s}_j\}_{j \in [m]}$ and $\{{s}_i\}_{i \in [n]}$ can viewed as unique identifiers for points.
% It satisfies that any two elements $\vecq_j \in Q$ and $\vecw_i \in W$ satisfying $\mathsf{dist}_\infty(\vecq_j, \vecw_i) \leq \delta$, it holds that ${\hat{s}}_i = {s}_i$. 

% Since $\mathsf{dist}_p(\vecq_j, \vecw_i) < \mathsf{dist}_\infty(\vecq_j, \vecw_i)$, we need to further determine whether $\mathsf{dist}_p(\vecq_j, \vecw_i) \leq \delta$. To this end, we use our $OT$-based fuzzy matching to compare.
% Similar to our fuzzy PSI protocol for high-dimensional sets, we employ Cuckoo hash to xxxxxxxxxxxxxxxxxxx.

\begin{theorem}
\label{thm: fuzzy PSI L-p-high-px}
    The protocol $\Pi_\mathsf{FPSI\text{-}high}^{L_p}$ in Figure \ref{Prot: fpsi-p-high-px} realizes the functionality \Func[FPSI] for $L_p$ distance in Figure \ref{Func:FPSI} against semi-honest adversaries {in the $(\Func[OPPRF], \Func[FMatch], \Func[OT])$-hybrid model} if both parties' sets satisfy the globally disjoint assumption.
\end{theorem}

\section{Evaluation}

We provide the full-fledged implementation of our protocols, and our code is available at {\color{blue}\url{https://github.com/Th0masAndy/Scalable-FPSI}} and {\color{blue}\url{https://doi.org/10.5281/zenodo.20607518}}.
We conduct our evaluation on a single server, equipped with two AMD EPYC 9554 processors and 512 GB of RAM. 
The sender and receiver are emulated as two separate threads within a single process. Each experiment is repeated ten times, and we report the average values. We set the computational security parameter to $\kappa = 128$ and the statistical security parameter to $\lambda = 40$. Network conditions are simulated using the Linux \texttt{tc} command, configured with a bandwidth of 10 Gbps and a latency of 0.02 ms, following recent works~\cite{van2024fuzzy, piske2025distance, bui2025new}.

\textbf{Implementation details.} We provide the implementation details as follows. For OPPRF, we adopt the implementation from~\cite{raghuraman2022blazing}, which includes the classic construction based on OKVS and VOLE-OPRF. For OT, we use silent OT from~\cite{libOTe}. For MSB, we build on the implementation of~\cite{huang2022cheetah}, replacing their OT component with~\cite{libOTe}. All hash functions are instantiated using BLAKE3. For cuckoo hashing parameters, we use 3 hash functions and no stash.

\textbf{Comparison baselines.} 
We compare our protocols with the state-of-the-art fuzzy PSI schemes~\cite{van2025, piske2025distance}, as their set assumptions align with ours, and all these protocols rely on OPRF (and OT). While~\cite{piske2025distance} theoretically supports arbitrary receiver sets, their experimental evaluation adopts the same assumptions as ours and as in~\cite{van2025}.
Since~\cite{van2025} does not provide an implementation, for a fair comparison, we implement their constructions with prefix-trie optimizations using the same components as in our protocols.
We also incorporate their optimization for solving Knapsack problem~\cite{horowitz1974computing}, which is required in their protocols for $L_\infty$ distances to reduce computation overhead.

% The work \cite{van2025} does not provide implementations. For a fair comparison, we additionally implement their constructions with prefix-trie optimizations, using the same components as in our implementation.
% We also implementation the optimizations of solving a Knapsack problem~\cite{horowitz1974computing}, which is required in their protocols for $L_\infty$ distances to reduce computation overhead.

% \subsubsection{Our Implementation for van Baarsen \& Pu~\cite{van2025}}
% Since their work does not provide a concrete performance evaluation, we additionally implement the protocol in~\cite{van2025} to enable a fair comparison. 

% We implement their construction primarily based on OPPRF, reusing the same components as in our implementation. 

% \blue{We observe that their protocol for $L_\infty$ requires solving a Knapsack problem~\cite{horowitz1974computing} with complexity $(\log \delta)^{\frac{d}{2}}$. When combined with other factors, the overall complexity becomes $n \cdot 2^d \cdot (\log \delta)^{\frac{d}{2}}$, a cost that is not accounted for in their performance estimates. Although this step is executed in plaintext, we find that it constitutes a significant bottleneck as $\delta$ and $d$ increase.}

% \subsection{Concrete Performance and Comparison}

\subsection{Performance for Low-dimensional Sets}
\label{Performance for Low-dimensional Sets}

We evaluate our protocols on low-dimensional sets under $L_\infty$, $L_1$, and $L_2$ distances, and compare them with the state-of-the-art protocols in~\cite{van2025, piske2025distance}. The results are reported in Tables~\ref{table: L0}, \ref{table: L1}, and~\ref{table: L2}. Since~\cite{piske2025distance} supports only 8-bit inputs, we do not report its performance for $\delta = 1024$. 
% Notably, this is the first evaluation of fuzzy PSI with input size $2^{18}$, and prior works incur prohibitive overhead even at an input size of $2^{16}$.

% We set the input size to $m = n =\{ 2^{12}, 2^{16},2^{18}\}$, the dimension to $d = \{2, 4, 8\}$ (low dimension), and the distance threshold to $\delta = \{16, 64, 256, 1024\}$. 

% To demonstrate scalability, we are the first to evaluate fuzzy PSI at an input size of $2^{18}$. In contrast, prior works incur prohibitive overhead even at an input size of $2^{16}$.

\begin{table*}[!t]
\centering
\caption{Comparisons of fuzzy PSI over low-dimensional sets for $L_\infty$ distance. "*" indicates unsupported or no implementation for the specific parameter setting. "---" indicates out of memory or too long running time. The best result is marked in {\color{blue!30}blue}.}
\label{table: L0}
\resizebox{0.91\textwidth}{!}{%
\begin{tabular}{|c|c|c|c|c|c|c|c|c|c|c|}
\hline
\multirow{2}{*}{\makecell{Set size \\ {$m = n$}} }    & \multirow{2}{*}{\makecell{Dim. \\ {$d$}}} & \multirow{2}{*}{Protocol}         &\multicolumn{2}{c|}{Threshold  $\delta = 16$} & \multicolumn{2}{c|}{Threshold  $\delta = 64$} & \multicolumn{2}{c|}{Threshold  $\delta = 256$} & \multicolumn{2}{c|}{Threshold  $\delta = 1024$} \\ \cline{4-11}                                                       
                       &                    &                        & Time (s)      & Comm. (MB)     & Time (s)      & Comm. (MB)      & Time (s)       & Comm. (MB)      & Time (s)        & Comm. (MB)      \\ \hline
\multirow{9}{*}{$2^{12}$}  & \multirow{3}{*}{2} & van Baarsen \& Pu~\cite{van2025}                   & {\cellcolor{blue!13}0.14}      & 5.65       & {\cellcolor{blue!13}0.15}      & 7.34       & {\cellcolor{blue!13}0.21}      & 9.07       & {\cellcolor{blue!13}0.23}        & 10.76      \\
                       &                    & Piske et al.~\cite{piske2025distance}                    & 1.92      & {\cellcolor{blue!13}2.21}       & 1.92      & {\cellcolor{blue!13}2.21}       & 1.89       & {\cellcolor{blue!13}2.21}       & *           & *          \\
                       &                    & Ours                   & {1.00}      & 4.07       & 0.94      & 4.07       & 0.99       & 4.07       & 0.96        & {\cellcolor{blue!13}4.07}       \\ \cline{2-11}
                       & \multirow{3}{*}{4} & van Baarsen \& Pu~\cite{van2025}                   & {\cellcolor{blue!13}0.85}      & 35.35      & {\cellcolor{blue!13}1.08}      & 46.85      & 1.34       & 58.53      & 1.70        & 70.11      \\
                       &                    & Piske et al.~\cite{piske2025distance}                    & 4.57      & 8.98       & 4.61      & 8.98       & 4.48       & 8.98       & *           & *          \\
                       &                    & Ours                   & 1.22      & {\cellcolor{blue!13}8.57}       & 1.16      & {\cellcolor{blue!13}8.57}       & {\cellcolor{blue!13}1.22}       & {\cellcolor{blue!13}8.57}       & {\cellcolor{blue!13}1.14}        & {\cellcolor{blue!13}8.57}       \\ \cline{2-11}
                       & \multirow{3}{*}{8} & van Baarsen \& Pu~\cite{van2025}                   & 64.27     & 1053.25    & 145.74    & 1401.62    & 348.73     & 1755.75    & 698.87      & 2106.82    \\
                       &                    & Piske et al.~\cite{piske2025distance}                    & 33.18     & 182.91     & 32.93     & 182.91     & 33.07      & 182.91     & *           & *          \\
                       &                    & Ours                   & {\cellcolor{blue!13}4.83}      & {\cellcolor{blue!13}105.06}     & {\cellcolor{blue!13}4.81}      & {\cellcolor{blue!13}105.06}     & {\cellcolor{blue!13}4.84}       & {\cellcolor{blue!13}105.06}     & {\cellcolor{blue!13}4.78}        & {\cellcolor{blue!13}105.06}     \\ \hline
\multirow{9}{*}{$2^{16}$} & \multirow{3}{*}{2} & van Baarsen \& Pu~\cite{van2025}                   & {\cellcolor{blue!13}1.68}      & 84.14      & {\cellcolor{blue!13}2.25}      & 111.23     & {\cellcolor{blue!13}2.83}       & 138.72     & {\cellcolor{blue!13}3.43}        & 165.97     \\
                       &                    & Piske et al.~\cite{piske2025distance}                    & 37.46     & 34.25      & 38.73     & 34.25      & 37.37      & 34.25      & *           & *          \\
                       &                    & Ours                   & 4.92      & {\cellcolor{blue!13}32.20}      & 5.06      & {\cellcolor{blue!13}32.20}      & 4.90       & {\cellcolor{blue!13}32.20}      & 4.96        & {\cellcolor{blue!13}32.20}      \\ \cline{2-11}
                       & \multirow{3}{*}{4} & van Baarsen \& Pu~\cite{van2025}                   & 13.87     & 560.04     & 18.89     & 744.36     & 24.68      & 931.54     & 29.34       & 1117.24    \\
                       &                    & Piske et al.~\cite{piske2025distance}                    & 109.56    & 142.01     & 87.58     & 142.01     & 85.56      & 142.01     & *           & *          \\
                       &                    & Ours                   & {\cellcolor{blue!13}8.46}      & {\cellcolor{blue!13}95.39}      & {\cellcolor{blue!13}8.38}      & {\cellcolor{blue!13}95.39}      & {\cellcolor{blue!13}8.29}       & {\cellcolor{blue!13}95.39}      & {\cellcolor{blue!13}8.35}        & {\cellcolor{blue!13}95.39}      \\ \cline{2-11}
                       & \multirow{3}{*}{8} & van Baarsen \& Pu~\cite{van2025}                   & 1278.99   & 16879.90   & 2653.76   & 22465.60   & 5830.80    & 28027.10   & 11697.40    & 33637.40   \\
                       &                    & Piske et al.~\cite{piske2025distance}                    & 579.98    & 2922.36    & 582.74    & 2922.36    & 574.98     & 2922.36    & *           & *          \\
                       &                    & Ours                   & {\cellcolor{blue!13}79.56}     & {\cellcolor{blue!13}1625.46}    & {\cellcolor{blue!13}79.00}     & {\cellcolor{blue!13}1625.46}    & {\cellcolor{blue!13}78.74}      & {\cellcolor{blue!13}1625.46}    & {\cellcolor{blue!13}79.04}       & {\cellcolor{blue!13}1625.46}  \\ \hline 
\multirow{9}{*}{$2^{18}$} & \multirow{3}{*}{2} & van Baarsen \& Pu~\cite{van2025}                   & {\cellcolor{blue!13}6.82}	&335.67	&{\cellcolor{blue!13}9.59}	&444.02	&{\cellcolor{blue!13}11.93}	&554.06	&{\cellcolor{blue!13}14.80}	&663.20
     \\
                       &                    & Piske et al.~\cite{piske2025distance}                    & 164	&136.48	&199.41	&136.48	 &* &* &* &*
          \\
                       &                    & Ours                   & 19.55	&{\cellcolor{blue!13}121.66}	&19.43	&{\cellcolor{blue!13}121.66}	&20.01	&{\cellcolor{blue!13}121.66}	&19.57	&{\cellcolor{blue!13}121.66}
    \\ \cline{2-11}
                       & \multirow{3}{*}{4} & van Baarsen \& Pu~\cite{van2025}                   & 58.44	&2241.18	&80.09	&2979.16	&108.71	&3728.65	&130.20	&4472.26
    \\
                       &                    & Piske et al.~\cite{piske2025distance}                    & 394.26	&567.23	&399.87	&567.23	&418.17	&567.23     & *           & *          \\
                       &                    & Ours                   & {\cellcolor{blue!13}36.20}	&{\cellcolor{blue!13}372.31}	&{\cellcolor{blue!13}35.41}	&{\cellcolor{blue!13}372.31}	&{\cellcolor{blue!13}36.03}	&{\cellcolor{blue!13}372.31}	&{\cellcolor{blue!13}35.89}	&{\cellcolor{blue!13}372.31}
      \\ \cline{2-11}
                       & \multirow{3}{*}{8} & van Baarsen \& Pu~\cite{van2025}                   & ---	 &---	 &---	 &---	&---	&---	&---	&---
   \\
                       &                    & Piske et al.~\cite{piske2025distance}                    & 2591.22	&11687.10	&2573.38	&11687.10	&2606.11	&11687.10	    & *           & *          \\
                       &                    & Ours                   & {\cellcolor{blue!13}379.89}	&{\cellcolor{blue!13}6494.57}	&{\cellcolor{blue!13}385.62}	&{\cellcolor{blue!13}6494.57}	&{\cellcolor{blue!13}388.46}	&{\cellcolor{blue!13}6494.57}	&{\cellcolor{blue!13}384.87}	&{\cellcolor{blue!13}6494.57}
  \\ \hline 
\end{tabular}%
}
\end{table*}

\textbf{Results of $L_\infty$ distance.}
As shown in Table~\ref{table: L0}, our protocol for $L_\infty$ distance demonstrates a significant advantage over~\cite{piske2025distance,van2025}, especially for relatively large parameter settings.
% . In particular, our protocol outperforms that of for relatively large parameter settings. 
For example, when $m = n = 2^{16}$ and $d = 8$, our protocol is $16{\sim}147\times$ faster than~\cite{van2025} depending on $\delta$, as their running time grows with $(\log \delta)^{d/2}$. 
% The reason is that the computational overhead of~\cite{van2025} grows as $(\log \delta)^{d/2}$. 
Although this step is executed in plaintext, when combined with the other components, the overall complexity becomes $n \cdot 2^d \cdot (\log \delta)^{d/2}$, which constitutes a significant bottleneck as both $\delta$ and $d$ increase.
Moreover, our construction further reduces communication by up to $20\times$, since unlike our OT-based fuzzy matching with reduced input domain, the protocol of \cite{van2025} incur the communication overhead dependent on the statistical security parameter.
% Our protocols benefit from a reduced input domain, whereas the bit length in~\cite{van2025} is $O(\lambda)$. As a result, 

Compared to~\cite{piske2025distance}, our protocol supports a larger input domain (e.g., 64-bit values in our implementation) while achieving lower running time and communication overhead. Specifically, when $m=n=2^{18}$ and $d=8$, our protocol is $7\times$ faster and incurs $1.8\times$ less communication cost than~\cite{piske2025distance}. 
When extended to larger input domains to support large $\delta$, the overhead of~\cite{piske2025distance} would increase substantially, further highlighting the advantage of our protocols.

% \blue{\textbf{Results of $L_p$ distances.} Tables~\ref{table: L1} and~\ref{table: L2} in Appendix \ref{appendix: Additional Experimental Results} show the results of $L_1$ and $L_2$ distances. We observe that our protocols for the $L_\infty$, $L_1$, and $L_2$ distance metrics exhibit similar overall overheads, because the constructions execute identical operations, differing exclusively in the $\mathcal{O}(m)$ invocations of fuzzy matching (instantiated for either $L_\infty$ or $L_p$). Due to our input domain reduction optimization, the concrete cost of these fuzzy matching operations does not dominate the total overhead. Compared to~\cite{van2025}, our protocol achieves up to a $5\times$ improvement in computational efficiency and up to a $34\times$ reduction in communication cost.
% As shown in Table~\ref{tab:complexities}, the overhead for $L_\infty$ distance of~\cite{van2025} includes a multiplicative factor of $2^d \cdot d \cdot \log \delta$, while that for $L_p$ distance includes multiplicative factors of $d \cdot \delta$ and $2^d \cdot d$. 
% This causes for smaller $\delta$, the $L_\infty$ protocol incurs higher overhead due to the dominant term $2^d \cdot d \cdot \log \delta$; for larger $\delta$, the $L_p$ protocol becomes more expensive due to the dominant term $d \cdot \delta$.
% }

{
\textbf{Results for $L_p$ Distances.} Tables~\ref{table: L1} and~\ref{table: L2} in Appendix~\ref{appendix: Additional Experimental Results} present the results for $L_1$ and $L_2$ distances. We observe that our protocols across the $L_\infty$, $L_1$, and $L_2$ distance metrics exhibit comparable overall overheads. This is because the constructions execute identical operations, differing exclusively in the $\mathcal{O}(m)$ invocations of fuzzy matching (instantiated for either $L_\infty$ or $L_p$). Due to our input domain reduction optimization, the concrete cost of these fuzzy matching operations does not dominate the total execution time.}

% \blue{
% Compared to~\cite{van2025}, our protocol achieves up to a $5\times$ computational speedup and a $34\times$ reduction in communication bandwidth. 
% Besides, we observe that for smaller values of $\delta$, the $L_\infty$ protocol in~\cite{van2025} incurs higher overhead than their $L_p$ protocol; conversely, for larger $\delta$, their $L_p$ protocol becomes more expensive. 
% This disparity stems from their asymptotic complexities.
% As outlined in Table~\ref{tab:complexities},
% the overhead of their $L_\infty$ protocol contains a multiplicative factor of $2^d \cdot d \cdot \log \delta$, whereas their $L_p$ construction is bottlenecked by factors of $d \cdot \delta$ and $2^d \cdot d$.
% }

{
Compared to~\cite{van2025}, our protocol achieves up to a $5\times$ computational speedup and a $34\times$ reduction in communication bandwidth. Furthermore, we observe a distinct threshold-dependent performance inversion in~\cite{van2025}: for smaller values of $\delta$, their $L_\infty$ protocol incurs higher overhead than their $L_p$ counterpart, whereas for larger $\delta$, their $L_p$ protocol becomes more expensive. This disparity stems from their asymptotic complexities. As outlined in Table~\ref{tab:complexities}, the overhead of their $L_\infty$ protocol scales with a multiplicative factor of $2^d d  \log \delta$ (dominating for small $\delta$), while their $L_p$ construction is bottlenecked by the terms $d \delta$ (dominating for large $\delta$) and $2^d  d$.
}

Compared to~\cite{piske2025distance}, when $m = n = 2^{18}$, $d = 4$, and $\delta = 256$, our protocol achieves up to a $25\times$ speedup and a $17\times$ reduction in communication cost for $L_1$ distance.
Since the protocol in~\cite{piske2025distance} supports only the setting $d = 2$ and $\delta \in \{10, 30\}$ for the $L_2$ distance, we do not include it in the Table~\ref{table: L2}. We evaluate our protocol under the same settings ($d=2$ and $\delta=\{8,32\}$), and observe that it reduces the communication and computational cost by up to $4\times$ and $10\times$, respectively. 

% The same advantage can also be observed for the $L_1$ and $L_2$ distances, with the results reported in Table~\ref{table: L1} and Table~\ref{table: L2}. Since the protocol in~\cite{piske2025distance} supports only the setting $d = 2$ and $\delta \in \{10, 30\}$ for the $L_2$ distance, we do not include it in the figure xxx. Compared with~\cite{van2025}, our protocol achieves up to a $5\times$ improvement in computational efficiency and a $34\times$ reduction in communication cost. Compared with~\cite{piske2025distance}, it achieves up to a $25\times$ improvement in computational efficiency and a $17\times$ reduction in communication cost, when $m = n = 2^{18}$, $d = 4$, and $\delta = 256$.

% In some special cases with small input sizes, our protocol is less efficient. We note that our design achieves substantially better performance improvements for large-scale inputs, since the overhead of the underlying silent OT~\cite{boyle2019efficient} can be amortized. 

% Moreover, since our protocol scales linearly with $\log \delta$, and in practical implementations $\log \delta$ must be aligned to multiples of 8 (i.e., one byte), the performance of our protocol remains nearly unchanged as $\delta$ increases from 16 to 1024 and may vary for bigger distance thresholds.

\blue{

}

In some special cases with small parameters, our protocol is less efficient, as the overhead of the underlying silent OT~\cite{boyle2019efficient} cannot be effectively amortized. 
Furthermore, the performance of our protocols remains nearly unchanged as $\delta$ increases from $16$ to $1024$, because in practice we round parameters that depend on $\log \delta$ to multiples of $8$ to facilitate efficient implementation.

% In some special cases with small parameters, our protocol is less efficient, since the overhead of the underlying silent OT~\cite{boyle2019efficient} can not be amortized.
% Besides, the performance of our protocols remains nearly unchanged as $\delta$ increases from 16 to 1024, because we align the parameter related to $\log \delta$ to be multiples of 8 for practical implementations.

\subsection{Performance for High-dimensional Sets}
\label{Performance for High-dimensional Sets}

\begin{table*}[!t]
\centering
\caption{Comparisons of fuzzy PSI over high-dimensional sets for $L_\infty$ distance. "*" indicates unsupported or no implementation for the specific parameter setting. The best result is marked in {\color{blue!30}blue}.}
\label{table: L_infty high}
\resizebox{0.91\textwidth}{!}{%
\begin{tabular}{|c|c|c|c|c|c|c|c|c|c|c|}
\hline
\multirow{2}{*}{\makecell{Set size \\ {$m = n$}} }    & \multirow{2}{*}{\makecell{Dim. \\ {$d$}}} & \multirow{2}{*}{Protocol}         &\multicolumn{2}{c|}{Threshold  $\delta = 16$} & \multicolumn{2}{c|}{Threshold  $\delta = 64$} & \multicolumn{2}{c|}{Threshold  $\delta = 256$} & \multicolumn{2}{c|}{Threshold  $\delta = 1024$} \\ \cline{4-11}                                                       
                       &                    &                        & Time (s)      & Comm. (MB)     & Time (s)      & Comm. (MB)      & Time (s)       & Comm. (MB)      & Time (s)        & Comm. (MB)      \\ \hline
\multirow{6}{*}{$2^{12}$}  & \multirow{2}{*}{16} & van Baarsen \& Pu~\cite{van2025}                   & 0.48	&47.24	&1.66	&177.97	&6.61	&701.40	&25.99	&2797.59

      \\
                       &                    & Ours                   & {\cellcolor{blue!13}0.42}	&{\cellcolor{blue!13}27.26}	&{\cellcolor{blue!13}0.47}	&{\cellcolor{blue!13}35.37} &	{\cellcolor{blue!13}0.61}	&{\cellcolor{blue!13}43.56}	&{\cellcolor{blue!13}0.70}	&{\cellcolor{blue!13}51.69}
       \\ \cline{2-11}
                       & \multirow{2}{*}{32} & van Baarsen \& Pu~\cite{van2025}                   & 0.91	&94.04	&3.42	&355.64	&13.35	&1403.19	&54.17	&5597.20

          \\
                       &                    & Ours                   & {\cellcolor{blue!13}0.73}	&{\cellcolor{blue!13}53.56}	&{\cellcolor{blue!13}0.95}	&{\cellcolor{blue!13}69.86}	&{\cellcolor{blue!13}1.12}	&{\cellcolor{blue!13}86.35}	&{\cellcolor{blue!13}1.40}	&{\cellcolor{blue!13}102.67}

       \\ \cline{2-11}
                       & \multirow{2}{*}{64} & van Baarsen \& Pu~\cite{van2025}                   &1.82	&187.68	&7.24	&711.17	&27.93	&2807.27	&109.52	&11199.90

    \\
                       
                       &                    & Ours                   & {\cellcolor{blue!13}1.43}	&{\cellcolor{blue!13}106.39}	&{\cellcolor{blue!13}1.97}	&{\cellcolor{blue!13}138.85}	&{\cellcolor{blue!13}2.50}	&{\cellcolor{blue!13}171.84}	&{\cellcolor{blue!13}3.00}	&{\cellcolor{blue!13}204.53}

     \\ \hline
\multirow{6}{*}{$2^{16}$} & \multirow{2}{*}{16} & van Baarsen \& Pu~\cite{van2025}                   & 7.68	&750.94	&28.84	&2847.02	&116.36	&11239.50	&*  &*

     \\
                       
                       &                    & Ours                   & {\cellcolor{blue!13}6.24}	&{\cellcolor{blue!13}424.14}	&{\cellcolor{blue!13}8.06}	&{\cellcolor{blue!13}554.11}	&{\cellcolor{blue!13}10.08}	&{\cellcolor{blue!13}686.16}	&{\cellcolor{blue!13}12.15}	&{\cellcolor{blue!13}817.09}

      \\ \cline{2-11}
                       & \multirow{2}{*}{32} & van Baarsen \& Pu~\cite{van2025}                   & 15.94	&1501.29	&62.01	&5695.39	&261.13	&22430.60	& *	& *

    \\
                       
                       &                    & Ours                   & {\cellcolor{blue!13}12.85}	&{\cellcolor{blue!13}846.85}	&{\cellcolor{blue!13}16.52}	&{\cellcolor{blue!13}1106.96}	&{\cellcolor{blue!13}20.75}	&{\cellcolor{blue!13}1371.10}	&{\cellcolor{blue!13}24.75}	&{\cellcolor{blue!13}1633.09}

      \\ \cline{2-11}
                       & \multirow{2}{*}{64} & van Baarsen \& Pu~\cite{van2025}                   & 33.03	&3002.60	&126.33	&11395.90 & *	& *	& *	& *

   \\
                       
                       &                    & Ours                   & {\cellcolor{blue!13}27.91}	&{\cellcolor{blue!13}1692.64}	&{\cellcolor{blue!13}33.68}	&{\cellcolor{blue!13}2213.14}	&{\cellcolor{blue!13}42.52}	&{\cellcolor{blue!13}2741.69}	&{\cellcolor{blue!13}50.47}	&{\cellcolor{blue!13}3266.10}

  \\ \hline 
  \multirow{6}{*}{$2^{18}$} & \multirow{2}{*}{16} & van Baarsen \& Pu~\cite{van2025}                   & 32.08	&3005.60	&140.46	&11398.90	&*	&*	&*	&*

     \\
                       
                       &                    & Ours                   & {\cellcolor{blue!13}25.36}	&{\cellcolor{blue!13}1695.64}	&{\cellcolor{blue!13}32.86}	&{\cellcolor{blue!13}2216.14}	&{\cellcolor{blue!13}42.69}	&{\cellcolor{blue!13}2744.69}	&{\cellcolor{blue!13}50.47}	&{\cellcolor{blue!13}3269.10}

      \\ \cline{2-11}
                       & \multirow{2}{*}{32} & van Baarsen \& Pu~\cite{van2025}                   & 66.62	&6009.38	&328.44	&22744.10	&*	&*	& *	& *

    \\
                       
                       &                    & Ours                   & {\cellcolor{blue!13}52.51}	&{\cellcolor{blue!13}3388.23}	&{\cellcolor{blue!13}69.22}	&{\cellcolor{blue!13}4429.27}	&{\cellcolor{blue!13}88.19}	&{\cellcolor{blue!13}5487.53}	&{\cellcolor{blue!13}105.49}	&{\cellcolor{blue!13}6536.69}

      \\ \cline{2-11}
                       & \multirow{2}{*}{64} & van Baarsen \& Pu~\cite{van2025}                   & 136.65	&12021.20	& *	& *	& *	& *	& *	& *

   \\
                       
                       &                    & Ours                   & {\cellcolor{blue!13}110.20}	&{\cellcolor{blue!13}6775.05}	&{\cellcolor{blue!13}150.63}	&{\cellcolor{blue!13}8859.02}	&{\cellcolor{blue!13}186.29}	&{\cellcolor{blue!13}10975.40}	&{\cellcolor{blue!13}245.43}	&{\cellcolor{blue!13}13075.30}

  \\ \hline
\end{tabular}%
}
\end{table*}

% We evaluate our protocol over high-dimensional datasets for all three distance metrics, and compare it with~\cite{van2025}, which uses the same assumption as ours and is the state-of-the-art fuzzy PSI protocol for high-dimensional settings. 

We evaluate our protocols for high-dimensional sets under three distance metrics and compare them with~\cite{van2025}, which adopts the same set assumptions as ours and is the state-of-the-art fuzzy PSI for high-dimensional settings.

% We keep all other parameters identical to those in the low-dimensional experiments, while increasing the dimension to $d =\{ 16, 32, 64\}$.

\textbf{Results of $L_\infty$ distance.} As shown in Table~\ref{table: L_infty high}, the overhead of our protocol for $L_\infty$ scales logarithmically with $\delta$, outperforming the protocol in~\cite{van2025} across all parameter settings, whose complexity grows linearly with $\delta$. 
In particular, our protocol exhibits a substantial advantage for large distance threshold $\delta$. When $m = n = 2^{12}$, $d = 64$, and $\delta = 1024$, our protocol requires $36\times$ less running time and $54\times$ less communication than~\cite{van2025}. 
We note that the protocol in~\cite{van2025} requires a large number of OPPRF invocations under some large parameter settings, and no concrete implementation is available for these cases. We thus mark these entries as “*”.

% , no concrete implementation is available when the input size exceeds $2^{31}$. We therefore mark these cases as “*”.

\textbf{Results of $L_p$ distances.} For the $L_1$ and $L_2$ distances, our protocol incurs additional overhead due to the invocation of OT-based fuzzy matching. Nevertheless, it still significantly outperforms~\cite{van2025} for most of the parameter settings. The results are presented in Tables~\ref{table: L_1 high} and~\ref{table: L_2 high} in Appendix \ref{appendix: Additional Experimental Results}. Specifically, our protocol achieves up to a $20\times$ reduction in running time and a $38\times$ reduction in communication cost.

\section{Conclusion}

% We present the scalable fuzzy PSI protocols for general $L_{p \in [1, \infty]}$ distance for both low- and high-dimensional sets.
% The core techniques are two fuzzy matching protocols, one based on our customized role-reversed OPRF and another on OT, allowing flexible instantiations depending on the input bit length. 
% We further present a new fuzzy PSI framework for low-dimensional sets based on our dual-layer hashing, which significantly reduces the number of fuzzy matching invocations.
% We additionally improve fuzzy PSI protocols for high-dimensional sets.
% Overall, our work improves previous fuzzy PSI works both concretely and asympotically.

% We present scalable fuzzy PSI protocols for general $L_{p \in [1, \infty]}$ distance, supporting both low- and high-dimensional sets. 
% Our core techniques include two asymptotically efficient fuzzy matching protocols: one based on a customized role-reversed OPRF and another based on OT, allowing flexible instantiations depending on the input bit length. 
% For low-dimensional sets, we propose a dual-layer hashing framework that substantially reduces the number of fuzzy matching invocations, while for high-dimensional sets, we improve existing fuzzy PSI constructions to achieve linear complexity in set size, dimension, and log complexity in the distance threshold. 
% Overall, our work advances prior fuzzy PSI protocols both concretely and asymptotically.

We present scalable fuzzy PSI protocols for general $L_{p \in [1, \infty]}$ distance, supporting both low- and high-dimensional sets. 
Our core techniques include two asymptotically efficient fuzzy matching protocols: one based on a customized role-reversed OPRF and another on OT.
% , enabling flexible instantiations depending on the input bit length. 
For low-dimensional sets, we introduce a dual-layer hashing framework that significantly reduces the number of fuzzy matching invocations, while for high-dimensional sets, we enhance existing fuzzy PSI constructions to achieve linear complexity in set size and dimension, with logarithmic complexity in the distance threshold. 
Overall, our work advances prior fuzzy PSI protocols both concretely and asymptotically.

\section*{Acknowledgments}
We would like to express our deepest gratitude to the reviewers for their invaluable help and constructive comments.

This work is supported by the Lee Kong Chian Chair Professorship, Singapore Management University.
Guomin Yang is supported by the Lee Kong Chian Fellowship awarded by Singapore Management University. 
This work is supported by the Singapore Ministry of Education (MOE) Academic Research Fund (AcRF) Tier 1 grants (Proposal ID: 24-SIS-SMU-052, 24-SIS-SMU-126). 
This work is also supported by the Nanyang Technological University Centre in Computational Technologies for Finance (NTU-CCTF), the National Research Foundation, Singapore, and Cyber Security Agency of Singapore under its National Cybersecurity R\&D Programme and CyberSG R\&D Cyber Research Programme Office.
This work is also supported by the National Natural Science Foundation of China under Grant 62502079, the China Postdoctoral Science Foundation under Grant BX20240053, and the Sichuan Natural Science Foundation 2026NSFSC1463.
Any opinions, findings, and conclusions or recommendations expressed in this material are those of the author(s) and do not necessarily reflect the views of NTU-CCTF, National Research Foundation, Singapore, Cyber Security Agency of Singapore as well as CyberSG R\&D Programme Office, Singapore.

\section*{Ethical Considerations}
This work provides effective solutions for fuzzy private set intersection tasks.
% , encouraging researchers to pay more attention to the privacy of set operations. 
The experiments in this paper are all based on synthetic datasets and do not contain any personal or illegal information. We firmly believe that our research was done ethically.

\section*{Open Science}
We follow the principles of the Open Science Policy. Our artifacts include the source code of our protocols, benchmarking scripts, and detailed operation documentation. We have consolidated these research artifacts into an open-source repository, and our artifacts are available at {\color{blue}\url{https://github.com/Th0masAndy/Scalable-FPSI}} and {\color{blue}\url{https://doi.org/10.5281/zenodo.20607518}}. 
% We will openly share our research artifact in the final version of our paper.
% We grant the scientific community unrestricted access to review, validate, and expand upon our work.

\section*{Generative AI Usage}

This paper used the LLM (including ChatGPT and Gemini) solely to assist with grammar checking and language polishing.

% \begin{acks}
% \end{acks}

\bibliographystyle{ACM-Reference-Format}
\bibliography{sample-base}

\appendix

\section{Assumption of Unique Ball/Center}
\label{appendix: Additional Details of Spatial Hashing}

Our work uses the unique ball/center as the assumption, which is essentially the same as that used in \cite{van2025, piske2025distance} with different representations.
For example, \cite{van2025} assumes that for any $L_{p\in[1, \infty]}$ distance, the sender’s points are $2\delta(d^{1/p})$-apart and the receiver’s points are $2\delta(d^{1/p}+1)$-apart.
As illustrated in \cite{van2024fuzzy, van2025}, this implies the unique ball/center. We recall their conclusion as follows.

% Baarsen and Pu \cite{van2024fuzzy, van2025} demonstrate that this assumption can be formalized through the distance between points. 

% \blue{As shown in Lemmas \ref{lemma: Unique Center} and \ref{lemma: Unique Ball}, BP25 formalizes these requirements through distance between points. 
% Specifically, }

\begin{lemma}[Unique Center \cite{van2024fuzzy, van2025}]
\label{lemma: Unique Center}

Suppose there are multiple $L_p$ balls ($p \in [1,\infty]$) with radius $\delta$
lying in a $d$-dimensional space which is tiled by cells with side length $2\delta$.
If these balls’ centers are at least $2\delta d^{1/p}$ apart from each other, then
for each cell, there is at most one center lying in this cell.
Specifically, if $p = \infty$, then this holds for disjoint balls,
since $2\delta d^{1/p}$ degrades to $2\delta$ in this case.

\end{lemma}

\begin{lemma}[Unique Ball \cite{van2024fuzzy, van2025}]
\label{lemma: Unique Ball}
Suppose there are multiple $L_p$ balls ($p \in [1,\infty]$) with radius $\delta$
lying in a $d$-dimensional space which is tiled by cells with side length $2\delta$.
If these balls’ centers are at least $2\delta(d^{1/p}+1)$ apart from each other,
then there exists at most one ball intersecting with the same cell.
Specifically, if $p = \infty$, then this holds for $L_\infty$ balls with
$4\delta$-apart centers.
\end{lemma}

% \begin{lemma}[Maximal Distance in a Cell \cite{van2024fuzzy, van2025}]
% Given two points $\mathbf{w}, \mathbf{q} \in U^d$ located in the same cell with 
% side length $2\delta$, then the distance between them is 
% $\mathrm{dist}_p(\mathbf{w}, \mathbf{q}) < 2\delta d^{1/p}$ where $p \in [1,\infty]$. 
% Specifically, if $p = \infty$, $\mathrm{dist}_\infty(\mathbf{w}, \mathbf{q}) < 2\delta$.
% \end{lemma}

\section{Details of OT-based Fuzzy Matching}
\label{appendix: OT-based Fuzzy Matching}

\subsection{Additional Functionalities}
\label{appendix: Additional Functionalities for OT-based Fuzzy Matching}

Our OT-based fuzzy matching protocols rely on several customized OT-based
primitives for securely computing the most significant bit, multiplication, and absolute
value. Below, we summarize the corresponding ideal functionalities and their communication overhead of state-of-the-art protocols.

\begin{enumerate}
    \item \textbf{Most significant bit (MSB).}
    The functionality $\Func[MSB]$ takes as input $[x]^\ell$ and outputs $[y]^1$, where $y = \mathsf{MSB}(x)$. Existing works~\cite{huang2022cheetah, rathee2020cryptflow2} provide efficient protocols with communication cost below $11\ell$ bits and $\log \ell$ communication rounds.

    \item \textbf{AND.}
    The functionality $\Func[AND]$ takes as input $\{[x_i]^1\}_{i \in [d]}$ and outputs $[y]^1$, where $y = x_1 \land \cdots \land x_d$. We generate random AND triples using silent OT~\cite{boyle2019efficient}, and then perform Beaver-style multiplications with $4(d-1)$ bits of communication and $\log d$ rounds.

    \item \textbf{Absolute values (ABS).} The functionality $\Func[ABS]$ takes as input $[x]^\ell$ and outputs $[y]^\ell$, where $y = |x|$. This functionality can be securely realized using MUX and MSB operations~\cite{rathee2021sirnn, rathee2020cryptflow2, hao2022iron, huang2022cheetah}.
    The protocol requires $13\ell$ bits of communication and $\log \ell$ rounds.

    % 11(\ell - 1) + 2(\ell + 1)

    \item \textbf{Multiplication (Mult).} The functionality $\Func[Mult]$ takes as input $([x]^\ell, [y]^\ell)$ and outputs $[z]^\ell$, where $z = x \cdot y$. We employ correlated-OT-based protocols from~\cite{demmler2015aby, rathee2021sirnn}, which incur communication of $\ell^2/2$ bits with 2 rounds.

\end{enumerate}

\subsection{OT-based Fuzzy Matching for $L_p$ Distance}
\label{appendix: OT-based Fuzzy Matching for $L_p$ Distance}

We can extend OT-based fuzzy matching to arbitrary $L_p$ distance.
In Figure \ref{Prot: fmatch-p-2pc}, we propose efficient protocols for widely used $L_1$ and $L_2$ distances by using MSB, ABS, and multiplication functionalities as shown in Appendix \ref{appendix: Additional Functionalities for OT-based Fuzzy Matching}.
Our protocol incurs an overhead of $O(pd \ell)$ in both communication\footnote{We note that the OT-based multiplication protocol incurs $O(\ell^2)$ communication. However, an alternative solution is to employ efficient RLWE-based additively homomorphic encryption, which incurs $O(\ell)$ communication. Therefore, we maintain the $O(pd\ell)$ communication overhead in our asymptotic analysis.} and computation.
The security of the protocol is shown as follows, and we defer the detailed correctness and security analysis to Appendix~\ref{proof: fuzzy-match-p-ot}.

\begin{theorem}
\label{thm: fuzzy-match-p-ot}
    The protocol $\Pi_\mathsf{FMatch\text{-}OT}^{L_p}$ in Figure \ref{Prot: fmatch-p-2pc} realizes the functionality \Func[FMatch] for $L_p$ distance in Figure \ref{Func:FMatch} against semi-honest adversaries in the $(\Func[ABS], \Func[Mult], \Func[MSB])$-hybrid model.
\end{theorem}

\begin{figure}[!t]
\begin{nfprot}{$\Pi_\mathsf{FMatch\text{-}ot}^{L_p}$}

\noindent \textbf{Parameter:} Sender \SSS and receiver \RRR. Threshold $\delta$.

\noindent \textbf{Input:} \SSS inputs $\vecq = (q_1, \ldots, q_d) \in \ZZ_{2^\ell}^{d}$ and \RRR inputs $\vecw = (w_1, \ldots, w_d) \in \ZZ_{2^\ell}^{d}$.

% $x=(x_1, \ldots, x_d) \in \ZZ^d_{2^u}$ and \RRR inputs $y=(y_1, \ldots, y_d) \in \ZZ^d_{2^u}$.

\noindent \textbf{Protocol:}

\begin{enumerate}

    \item For $k \in [d]$, \SSS and \RRR compute $[t_k]^{\ell} := q_k - w_k$. 

    \item For $k \in [d]$, \SSS and \RRR executes as follows and learn $[f_k]^{\ell}$.

    \begin{itemize}
        \item $p = 1$:  Invoke functionality $\Func[ABS]$ with input $[t_k]^{\ell}$.
        \item $p = 2$: Invoke functionality \Func[Mult] with input $([t_k]^{\ell}, [t_k]^{\ell})$.
    \end{itemize}

    % \item \underline{If $p = 1$}, for $k \in [d]$, \SSS and \RRR invoke functionality $\Func[ABS]$ with input $[t_k]^{\ell}$ and learn $[f_k]^{\ell}$.
    % \SSS and \RRR compute $[f_k]^{\ell} := -2 \cdot [g_k]^{\ell} + [t_k]^{\ell}$.

    % \item \underline{If $p = 2$}, for $k \in [d]$, \SSS and \RRR invoke functionality \Func[Mult] with input $([t_k]^{\ell}, [t_k]^{\ell})$ and learn  $[f_k]^{\ell}$.

    \item {\SSS and \RRR locally compute $[f]^{\ell} := \delta^p - \sum _{k\in[d]} [f_k]^{\ell}$.} \SSS and \RRR invoke functionality $\Func[MSB]$ with inputs $[f]^{\ell}$ and learn $[z]^1$.

    \item {\SSS sends $1 \xor [z]^1_\SSS$ to \RRR, who outputs $z := 1 \xor [z]^1_\SSS \xor [z]^1_\RRR$.} 
    
\end{enumerate}

\end{nfprot}
\caption{Protocol of fuzzy matching for $L_p$ distance from OT}
\label{Prot: fmatch-p-2pc}
\end{figure}

\iffalse
\subsection{Comparisons between OPRF-based and OT-based Fuzzy Matching}

\label{appendix: Comparisons between OPRF-based and OT-based Fuzzy Matching}

\blue{Communication $L_\infty$: $< 22 \ell d - 14 d -3$ bits}
\blue{$L_1$: $< (13 \ell - 9) d + 11 \ell - 10$ bits}
\blue{$L_2$: $< 4 \ell d + 11 \ell - 10$ bits}

Both our OPRF-based and OT-based constructions achieve desirable performance. Specifically, the OPRF-based protocol incurs communication overhead $O(d \lambda \log \delta)$ and computational overhead $O(d \log \delta)$, while the OT-based protocol incurs both communication and computational overhead of $O(d \ell)$.

When $\ell = O(\log \delta)$, the two constructions have comparable computational overhead $O(d \log \delta)$. However, the OT-based protocol achieves strictly better communication complexity $O(d \log \delta)$, as it removes the dependence on the statistical security parameter $\lambda$ required by the OPRF-based construction. As a trade-off, the OT-based protocol requires $\log \ell$ communication rounds due to the underlying OT-based primitives, in contrast to the constant-round complexity of the OPRF-based protocol. For $\ell = O(\log \delta)$, this translates to $O(\log \log \delta)$ rounds.

Therefore, the OT-based protocol is particularly well suited for inputs with small bit length, a regime that will be illustrated in our fuzzy PSI constructions for low-dimensional sets.
\fi

% \section{Additional Protocols of Fuzzy PSI}

\begin{table*}[!t]
\centering
\caption{Comparisons of fuzzy PSI over low-dimensional sets for $L_1$ distance. "*" indicates unsupported or no implementation for the specific parameter setting. The best result is marked in {\color{blue!30}blue}.}
\label{table: L1}
\resizebox{0.91\textwidth}{!}{%
\begin{tabular}{|c|c|c|c|c|c|c|c|c|c|c|}
\hline
\multirow{2}{*}{\makecell{Set size \\ {$m = n$}} }    & \multirow{2}{*}{\makecell{Dim. \\ {$d$}}} & \multirow{2}{*}{Protocol}         &\multicolumn{2}{c|}{Threshold  $\delta = 16$} & \multicolumn{2}{c|}{Threshold  $\delta = 64$} & \multicolumn{2}{c|}{Threshold  $\delta = 256$} & \multicolumn{2}{c|}{Threshold  $\delta = 1024$} \\ \cline{4-11}                                                       
                       &                    &                        & Time (s)      & Comm. (MB)     & Time (s)      & Comm. (MB)      & Time (s)       & Comm. (MB)      & Time (s)        & Comm. (MB)      \\ \hline
\multirow{9}{*}{$2^{12}$}  & \multirow{3}{*}{2} & van Baarsen \& Pu~\cite{van2025}                   & {\cellcolor{blue!13}0.11}	 &7.33	&{\cellcolor{blue!13}0.27}	&23.9	&0.91	&89.48	&3.35	&351.33
      \\
                       &                    & Piske et al.~\cite{piske2025distance}                    & 1.95	&5.9	&2.23	&13.88	&3.59	&43.73	&*	&*
          \\
                       &                    & Ours                   & 0.66	&{\cellcolor{blue!13}3.17}	&0.64	&{\cellcolor{blue!13}3.17}	&{\cellcolor{blue!13}0.79}	&{\cellcolor{blue!13}3.17}	&{\cellcolor{blue!13}0.75}	&{\cellcolor{blue!13}3.17}
       \\ \cline{2-11}
                       & \multirow{3}{*}{4} & van Baarsen \& Pu~\cite{van2025}                   & {\cellcolor{blue!13}0.27}	&18.01	&{\cellcolor{blue!13}0.63}	&51.04	&1.94	&182.15	&7.02	&705.97
      \\
                       &                    & Piske et al.~\cite{piske2025distance}                    & 4.63	&16.35	&5.36	&32.28	&8.26	&91.96	&*	&*
          \\
                       &                    & Ours                   & 0.87	&{\cellcolor{blue!13}6.84}	&0.89	&{\cellcolor{blue!13}6.84}	&{\cellcolor{blue!13}0.89}	&{\cellcolor{blue!13}6.84}	&{\cellcolor{blue!13}0.86}	&{\cellcolor{blue!13}6.84}
       \\ \cline{2-11}
                       & \multirow{3}{*}{8} & van Baarsen \& Pu~\cite{van2025}                   & 4.78	&198.62	&5.7	&264.58	&9.17	&526.82	&20.15	&1575.00
    \\
                       &                    & Piske et al.~\cite{piske2025distance}                    & 34.56	&197.63	&35.54	&229.47	&40.62	&348.8     & *           & *          \\
                       &                    & Ours                   & {\cellcolor{blue!13}4.73}	&{\cellcolor{blue!13}101.69}	&{\cellcolor{blue!13}4.64}	&{\cellcolor{blue!13}101.69}	&{\cellcolor{blue!13}4.38}	&{\cellcolor{blue!13}101.69}	&{\cellcolor{blue!13}3.60}	&{\cellcolor{blue!13}134.77}
     \\ \hline
\multirow{9}{*}{$2^{16}$} & \multirow{3}{*}{2} & van Baarsen \& Pu~\cite{van2025}                   & {\cellcolor{blue!13}1.45}	&111.14	&4.06	&376.74	&14.69	&1428.29	&58.81	&5626.30
     \\
                       &                    & Piske et al.~\cite{piske2025distance}                    & 39.46	&93.03	&46.25	&220.31	&79.94	&697.43      & *           & *          \\
                       &                    & Ours                   & 3.81	&{\cellcolor{blue!13}27.32}	&{\cellcolor{blue!13}3.94}	&{\cellcolor{blue!13}27.32}	&{\cellcolor{blue!13}3.96}	&{\cellcolor{blue!13}27.32}	&{\cellcolor{blue!13}4.20}	&{\cellcolor{blue!13}27.32}
      \\ \cline{2-11}
                       & \multirow{3}{*}{4} & van Baarsen \& Pu~\cite{van2025}                   & {\cellcolor{blue!13}4.76}	&282.25	&10.61	&811.74	&33.35	&2913.84	&168.91	&11312.50
    \\
                       &                    & Piske et al.~\cite{piske2025distance}                    & 91.97	&259.52	&109.06	&514.02	&159.56	&1468.18     & *           & *          \\
                       &                    & Ours                   & 6.72	&{\cellcolor{blue!13}85.73}	&{\cellcolor{blue!13}6.9}	&{\cellcolor{blue!13}85.73}	&{\cellcolor{blue!13}6.69}	&{\cellcolor{blue!13}85.73}	&{\cellcolor{blue!13}6.61}	&{\cellcolor{blue!13}85.73}
      \\ \cline{2-11}
                       & \multirow{3}{*}{8} & van Baarsen \& Pu~\cite{van2025}                   & 92.57	&3177.7	&109.22	&4235.31	&205.05	&8439.44	&485.63	&25184.90
   \\
                       &                    & Piske et al.~\cite{piske2025distance}                    & 592.08	&3157.32	&623.35	&3666.25	&734.61	&5574.42    & *           & *          \\
                       &                    & Ours                   & {\cellcolor{blue!13}78.13}	&{\cellcolor{blue!13}1606.23}	&{\cellcolor{blue!13}77.66}	&{\cellcolor{blue!13}1606.23}	&{\cellcolor{blue!13}77.7}	&{\cellcolor{blue!13}1606.23}	&{\cellcolor{blue!13}68.39}	&{\cellcolor{blue!13}2136.62}
  \\ \hline 
  \multirow{9}{*}{$2^{18}$} & \multirow{3}{*}{2} & van Baarsen \& Pu~\cite{van2025}                   &  {\cellcolor{blue!13}5.83}	 & 443.64	 & 17.82	 & 1507.24	 & 63.13	 & 5717.38	 & 264.72	 & 22468.90

     \\
                       &                    & Piske et al.~\cite{piske2025distance}                    & 171.61	 & 371.44	 & 196.34	 & 880.37	 & 323.33	 & 2788.54      & *           & *          \\
                       &                    & Ours                   & 15.81	 & {\cellcolor{blue!13}104.72}	 & {\cellcolor{blue!13}16.30}	 & {\cellcolor{blue!13}104.72}	 & {\cellcolor{blue!13}16.40}	 & {\cellcolor{blue!13}104.72}	 & {\cellcolor{blue!13}15.98}	 & {\cellcolor{blue!13}104.72}

      \\ \cline{2-11}
                       & \multirow{3}{*}{4} & van Baarsen \& Pu~\cite{van2025}                   & {\cellcolor{blue!13}19.89}	 & 1128.75	 & 46.41	 & 3248.83	 & 146.63	 & 11665.30	&*	&*
    \\
                       &                    & Piske et al.~\cite{piske2025distance}                    & 411.09	 & 1037.06	 & 465.84	 & 2054.81	 & 748.99	 & 5870.97     & *           & *          \\
                       &                    & Ours                   & 28.64	 & {\cellcolor{blue!13}338.53}	 & {\cellcolor{blue!13}28.71}	 & {\cellcolor{blue!13}338.53}	 & {\cellcolor{blue!13}29.30}	 & {\cellcolor{blue!13}338.53}	 & {\cellcolor{blue!13}28.94}	 & {\cellcolor{blue!13}338.53}

      \\ \cline{2-11}
                       & \multirow{3}{*}{8} & van Baarsen \& Pu~\cite{van2025}                   & 439.35	 & 12721.40	 & 521.57	 & 16955.50	 & 814.18	 & 33730.60	&*	&*
   \\
                       &                    & Piske et al.~\cite{piske2025distance}                    & 2639.26	 & 12626.70	 & 2781.08	 & 14663.3	 & 3361.54	 & 22294.10    & *           & *          \\
                       &                    & Ours                   & {\cellcolor{blue!13}366.60}	 & {\cellcolor{blue!13}6427.12}	 & {\cellcolor{blue!13}395.73}	 & {\cellcolor{blue!13}6427.12}	 & {\cellcolor{blue!13}373.99}	 & {\cellcolor{blue!13}6427.12}	 & {\cellcolor{blue!13}326.58}	 & {\cellcolor{blue!13}8551.02}

  \\ \hline
\end{tabular}%
}
\end{table*}

\begin{table*}[!t]
\centering
\caption{Comparisons of fuzzy PSI over low-dimensional sets for $L_2$ distance. "*" indicates unsupported or no implementation for the specific parameter setting. The best result is marked in {\color{blue!30}blue}.}
\label{table: L2}
\resizebox{0.91\textwidth}{!}{%
\begin{tabular}{|c|c|c|c|c|c|c|c|c|c|c|}
\hline
\multirow{2}{*}{\makecell{Set size \\ {$m = n$}} }    & \multirow{2}{*}{\makecell{Dim. \\ {$d$}}} & \multirow{2}{*}{Protocol}         &\multicolumn{2}{c|}{Threshold  $\delta = 16$} & \multicolumn{2}{c|}{Threshold  $\delta = 64$} & \multicolumn{2}{c|}{Threshold  $\delta = 256$} & \multicolumn{2}{c|}{Threshold  $\delta = 1024$} \\ \cline{4-11}                                                       
                       &                    &                        & Time (s)      & Comm. (MB)     & Time (s)      & Comm. (MB)      & Time (s)       & Comm. (MB)      & Time (s)        & Comm. (MB)      \\ \hline
\multirow{6}{*}{$2^{12}$}  & \multirow{2}{*}{2} & van Baarsen \& Pu~\cite{van2025}                   & {\cellcolor{blue!13}0.14}	& 7.83	& {\cellcolor{blue!13}0.29}	& 24.65	& 0.94	& 90.48	& 3.93	& 352.58

      \\
                       &                    & Ours                   & 0.59	& {\cellcolor{blue!13}2.88}	& 0.57	& {\cellcolor{blue!13}3.32}	& {\cellcolor{blue!13}0.58}	& {\cellcolor{blue!13}3.32}	& {\cellcolor{blue!13}0.57}	& {\cellcolor{blue!13}3.85}

       \\ \cline{2-11}
                       & \multirow{2}{*}{4} & van Baarsen \& Pu~\cite{van2025}                   & {\cellcolor{blue!13}0.33}	& 18.76	& {\cellcolor{blue!13}0.72}	& 52.17	& 2.09	& 183.65	& 7.82	& 707.84
          \\
                       &                    & Ours                   & 0.60	& {\cellcolor{blue!13}6.54}	& 0.73	& {\cellcolor{blue!13}8.13}	& {\cellcolor{blue!13}0.74}	& {\cellcolor{blue!13}8.13}	& {\cellcolor{blue!13}0.74}	& {\cellcolor{blue!13}9.88}

       \\ \cline{2-11}
                       & \multirow{2}{*}{8} & van Baarsen \& Pu~\cite{van2025}                   &5.32	& 199.87	& 6.41	& 266.46	& 10.08	& 529.32	& 22.89	& 1578.13

    \\
                       
                       &                    & Ours                   & {\cellcolor{blue!13}4.04}	& {\cellcolor{blue!13}101.39}	& {\cellcolor{blue!13}3.24}	& {\cellcolor{blue!13}134.97}	& {\cellcolor{blue!13}3.23}	& {\cellcolor{blue!13}134.97}	& {\cellcolor{blue!13}4.13}	& {\cellcolor{blue!13}169.34}

     \\ \hline
\multirow{6}{*}{$2^{16}$} & \multirow{2}{*}{2} & van Baarsen \& Pu~\cite{van2025}                   & {\cellcolor{blue!13}1.66}	& 119.14	& 4.42	& 388.74	& 15.92	& 1444.29	& 68.34	& 5646.3

     \\
                       
                       &                    & Ours                   & 3.12	& {\cellcolor{blue!13}27.19}	& {\cellcolor{blue!13}3.65}	& {\cellcolor{blue!13}34.17}	& {\cellcolor{blue!13}3.77}	& {\cellcolor{blue!13}34.17}	& {\cellcolor{blue!13}3.86}	& {\cellcolor{blue!13}42.44}

      \\ \cline{2-11}
                       & \multirow{2}{*}{4} & van Baarsen \& Pu~\cite{van2025}                   & 5.49	& 294.25	& 11.60	& 829.74	& 36.25	& 2937.84	& 162.54	& 11342.5

    \\
                       
                       &                    & Ours                   & {\cellcolor{blue!13}5.21}	& {\cellcolor{blue!13}85.85}	& {\cellcolor{blue!13}6.83}	& {\cellcolor{blue!13}111.34}	& {\cellcolor{blue!13}6.93}	& {\cellcolor{blue!13}111.34}	& {\cellcolor{blue!13}7.61}	& {\cellcolor{blue!13}139.33}

      \\ \cline{2-11}
                       & \multirow{2}{*}{8} & van Baarsen \& Pu~\cite{van2025}                   & 100.28	& 3197.70	& 123.78	& 4265.31	& 187.23	& 8479.44	& 423.36	& 25234.90

   \\
                       
                       &                    & Ours                   & {\cellcolor{blue!13}73.01}	& {\cellcolor{blue!13}1606.89}	& {\cellcolor{blue!13}58.71}	& {\cellcolor{blue!13}2145.39}	& {\cellcolor{blue!13}65.18}	& {\cellcolor{blue!13}2145.39}	& {\cellcolor{blue!13}73.79}	& {\cellcolor{blue!13}2696.60}

  \\ \hline 
  \multirow{6}{*}{$2^{18}$} & \multirow{2}{*}{2} & van Baarsen \& Pu~\cite{van2025}                   & {\cellcolor{blue!13}6.99}	&475.64	&19.69	&1555.24	&68.85	&5781.38	&298.37	&22548.90

     \\
                       
                       &                    & Ours                   & 13.35	&{\cellcolor{blue!13}105.38}	&{\cellcolor{blue!13}14.42}	&{\cellcolor{blue!13}133.38}	&{\cellcolor{blue!13}14.95}	&{\cellcolor{blue!13}133.38}	&{\cellcolor{blue!13}15.51}	&{\cellcolor{blue!13}166.64}

      \\ \cline{2-11}
                       & \multirow{2}{*}{4} & van Baarsen \& Pu~\cite{van2025}                   & {\cellcolor{blue!13}22.45}	&1176.75	&47.77	&3320.83	&187.45	&11761.30	& *	& *

    \\
                       
                       &                    & Ours                   & 22.93	&{\cellcolor{blue!13}340.32}	&{\cellcolor{blue!13}29.76}	&{\cellcolor{blue!13}442.52}	&{\cellcolor{blue!13}29.84}	&{\cellcolor{blue!13}442.52}	&{\cellcolor{blue!13}32.26}	&{\cellcolor{blue!13}554.86}

      \\ \cline{2-11}
                       & \multirow{2}{*}{8} & van Baarsen \& Pu~\cite{van2025}                   & 482.04	&12801.40	&592.76	&17075.50	&898.91	&33890.60	& *	& *

   \\
                       
                       &                    & Ours                   & {\cellcolor{blue!13}365.36}	&{\cellcolor{blue!13}6431.17}	&{\cellcolor{blue!13}350.83}	&{\cellcolor{blue!13}8587.77}	&{\cellcolor{blue!13}306.63}	&{\cellcolor{blue!13}8587.77}	&{\cellcolor{blue!13}383.05}	&{\cellcolor{blue!13}10795.40}

  \\ \hline
\end{tabular}%
}
\end{table*}

\begin{table*}[!t]
\centering
\caption{Comparisons of fuzzy PSI over high-dimensional sets for $L_1$ distance. "*" indicates unsupported or no implementation for the specific parameter setting. The best result is marked in {\color{blue!30}blue}.}
\label{table: L_1 high}
\resizebox{0.91\textwidth}{!}{%
\begin{tabular}{|c|c|c|c|c|c|c|c|c|c|c|}
\hline
\multirow{2}{*}{\makecell{Set size \\ {$m = n$}} }    & \multirow{2}{*}{\makecell{Dim. \\ {$d$}}} & \multirow{2}{*}{Protocol}         &\multicolumn{2}{c|}{Threshold  $\delta = 16$} & \multicolumn{2}{c|}{Threshold  $\delta = 64$} & \multicolumn{2}{c|}{Threshold  $\delta = 256$} & \multicolumn{2}{c|}{Threshold  $\delta = 1024$} \\ \cline{4-11}                                                       
                       &                    &                        & Time (s)      & Comm. (MB)     & Time (s)      & Comm. (MB)      & Time (s)       & Comm. (MB)      & Time (s)        & Comm. (MB)      \\ \hline
\multirow{6}{*}{$2^{12}$}  & \multirow{2}{*}{16} & van Baarsen \& Pu~\cite{van2025}                   & {\cellcolor{blue!13}0.46}	 & 49.49	 & {\cellcolor{blue!13}1.68}	 & 181.35	 & 7.29	 & 705.9	 & 26.95	 & 2803.22

      \\
                       &                    & Ours                   & 2.93	 & {\cellcolor{blue!13}39.32}	 & 3.18	 & {\cellcolor{blue!13}47.43}	 & {\cellcolor{blue!13}2.98}	 & {\cellcolor{blue!13}55.62}	 & {\cellcolor{blue!13}3.37}	 & {\cellcolor{blue!13}63.75}

       \\ \cline{2-11}
                       & \multirow{2}{*}{32} & van Baarsen \& Pu~\cite{van2025}                   & {\cellcolor{blue!13}0.90}	 & 98.29	 & {\cellcolor{blue!13}3.38}	 & 362.02	 & 13.68	 & 1411.69	 & 54.73	 & 5607.82

          \\
                       &                    & Ours                   & 5.29	 & {\cellcolor{blue!13}74.79}	 & 5.94	 & {\cellcolor{blue!13}91.09}	 & {\cellcolor{blue!13}5.66}	 & {\cellcolor{blue!13}107.58}	 & {\cellcolor{blue!13}5.96}	 & {\cellcolor{blue!13}123.90}

       \\ \cline{2-11}
                       & \multirow{2}{*}{64} & van Baarsen \& Pu~\cite{van2025}                   &{\cellcolor{blue!13}1.88}	 & 195.93	 & {\cellcolor{blue!13}6.86}	 & 723.55	 & 30.28	 & 2823.77	 & 115.75	 & 11220.50

    \\
                       
                       &                    & Ours                   & 9.49	 & {\cellcolor{blue!13}145.93}	 & 10.15	 & {\cellcolor{blue!13}178.39}	 & {\cellcolor{blue!13}10.31}	 & {\cellcolor{blue!13}211.37}	 & {\cellcolor{blue!13}11.36}	 & {\cellcolor{blue!13}244.06}

     \\ \hline
\multirow{6}{*}{$2^{16}$} & \multirow{2}{*}{16} & van Baarsen \& Pu~\cite{van2025}                   & {\cellcolor{blue!13}7.99}	 & 786.94	 & {\cellcolor{blue!13}28.99}	 & 2901.02	 & 115.97	 & 11311.50	&*  &*

     \\
                       
                       &                    & Ours                   & 40.97	 & {\cellcolor{blue!13}592.16}	 & 42.51 & 	{\cellcolor{blue!13}722.14}	 & {\cellcolor{blue!13}45.33}	 & {\cellcolor{blue!13}854.18}	 & {\cellcolor{blue!13}47.45}	 & {\cellcolor{blue!13}985.11}

      \\ \cline{2-11}
                       & \multirow{2}{*}{32} & van Baarsen \& Pu~\cite{van2025}                   & {\cellcolor{blue!13}16.13}	 & 1569.29	 & {\cellcolor{blue!13}60.04}	 & 5797.39	 & 252.21	 & 22566.60	& *	& *

    \\
                       
                       &                    & Ours                   & 81.10	 & {\cellcolor{blue!13}1164.54}	 & 82.17	 & {\cellcolor{blue!13}1424.65}	 & {\cellcolor{blue!13}86.62}	 & {\cellcolor{blue!13}1688.79}	 & {\cellcolor{blue!13}92.64}	 & {\cellcolor{blue!13}1950.79}

      \\ \cline{2-11}
                       & \multirow{2}{*}{64} & van Baarsen \& Pu~\cite{van2025}                   & {\cellcolor{blue!13}33.04}	 & 3134.60	 & {\cellcolor{blue!13}126.01}	 & 11593.90 & *	& *	& *	& *

   \\
                       
                       &                    & Ours                   & 158.10	 & {\cellcolor{blue!13}2309.64	} & 179.22	 &{\cellcolor{blue!13} 2830.14}	 & {\cellcolor{blue!13}179.47}	 & {\cellcolor{blue!13}3358.68}	 & {\cellcolor{blue!13}189.75}	 & {\cellcolor{blue!13}3883.10}

  \\ \hline 
  \multirow{6}{*}{$2^{18}$} & \multirow{2}{*}{16} & van Baarsen \& Pu~\cite{van2025}                   & {\cellcolor{blue!13}35.17}	 & 3149.60	 & {\cellcolor{blue!13}129.34}	 & 11614.90	&*	&*	&*	&*

     \\
                       
                       &                    & Ours                   & 173.76  & 	{\cellcolor{blue!13}2364.69}	 & 182.71	 & {\cellcolor{blue!13}2885.19}	 & {\cellcolor{blue!13}195.19}	 & {\cellcolor{blue!13}3413.73}	 & {\cellcolor{blue!13}225.93}	 & {\cellcolor{blue!13}3938.15}

      \\ \cline{2-11}
                       & \multirow{2}{*}{32} & van Baarsen \& Pu~\cite{van2025}                   & {\cellcolor{blue!13}70.20}	 & 6281.38	 & {\cellcolor{blue!13}293.72}	 & 23152.10	&*	&*	& *	& *

    \\
                       
                       &                    & Ours                   & 381.05	 & {\cellcolor{blue!13}4659.33}	 & 399.78	 & {\cellcolor{blue!13}5700.37}	 & {\cellcolor{blue!13}445.68}	 & {\cellcolor{blue!13}6758.63}	 & {\cellcolor{blue!13}438.60}	 & {\cellcolor{blue!13}7807.79}

      \\ \cline{2-11}
                       & \multirow{2}{*}{64} & van Baarsen \& Pu~\cite{van2025}                   & {\cellcolor{blue!13}142.61}	 & 12549.20	& *	& *	& *	& *	& *	& *

   \\
                       
                       &                    & Ours                   & 974.50	 & {\cellcolor{blue!13}9250.23}	 & {\cellcolor{blue!13}1023.99}	 & {\cellcolor{blue!13}11334.20}	 & {\cellcolor{blue!13}1069.14}	 & {\cellcolor{blue!13}13450.60}	 & {\cellcolor{blue!13}1118.85}	 & {\cellcolor{blue!13}15550.50}

  \\ \hline
\end{tabular}%
}
\end{table*}

\begin{table*}[!t]
\centering
\caption{Comparisons of fuzzy PSI over high-dimensional sets for $L_2$ distance. "*" indicates unsupported or no implementation for the specific parameter setting. The best result is marked in {\color{blue!30}blue}.}
\label{table: L_2 high}
\resizebox{0.91\textwidth}{!}{%
\begin{tabular}{|c|c|c|c|c|c|c|c|c|c|c|}
\hline
\multirow{2}{*}{\makecell{Set size \\ {$m = n$}} }    & \multirow{2}{*}{\makecell{Dim. \\ {$d$}}} & \multirow{2}{*}{Protocol}         &\multicolumn{2}{c|}{Threshold  $\delta = 16$} & \multicolumn{2}{c|}{Threshold  $\delta = 64$} & \multicolumn{2}{c|}{Threshold  $\delta = 256$} & \multicolumn{2}{c|}{Threshold  $\delta = 1024$} \\ \cline{4-11}                                                       
                       &                    &                        & Time (s)      & Comm. (MB)     & Time (s)      & Comm. (MB)      & Time (s)       & Comm. (MB)      & Time (s)        & Comm. (MB)      \\ \hline
\multirow{6}{*}{$2^{12}$}  & \multirow{2}{*}{16} & van Baarsen \& Pu~\cite{van2025}                   & {\cellcolor{blue!13}0.48}	& 51.74	& 1.71	& 184.72	& 6.76	& 710.40	& 28.63	& 2808.84

      \\
                       &                    & Ours                   & 1.51	& {\cellcolor{blue!13}50.35}	& {\cellcolor{blue!13}1.68}	& {\cellcolor{blue!13}58.47}	& {\cellcolor{blue!13}1.76}	& {\cellcolor{blue!13}66.66}	& {\cellcolor{blue!13}1.97}	& {\cellcolor{blue!13}74.79}

       \\ \cline{2-11}
                       & \multirow{2}{*}{32} & van Baarsen \& Pu~\cite{van2025}                   & {\cellcolor{blue!13}0.93}	& 102.54	& 3.44	& 368.39	& 13.82	& 1420.19	& 57.93	& 5618.45

          \\
                       &                    & Ours                   & 2.43	& {\cellcolor{blue!13}97.25}	& {\cellcolor{blue!13}2.65}	& {\cellcolor{blue!13}113.55}	& {\cellcolor{blue!13}2.86}	& {\cellcolor{blue!13}130.04}	& {\cellcolor{blue!13}3.23}	& {\cellcolor{blue!13}146.36}

       \\ \cline{2-11}
                       & \multirow{2}{*}{64} & van Baarsen \& Pu~\cite{van2025}                   &{\cellcolor{blue!13}1.88}	& 204.18	& 6.96	& 735.92	& 27.48	& 2840.27	& 116.44	& 11241.20

    \\
                       
                       &                    & Ours                   & 4.23	& {\cellcolor{blue!13}191.26}	& {\cellcolor{blue!13}4.72}	& {\cellcolor{blue!13}223.72}	& {\cellcolor{blue!13}5.04}	& {\cellcolor{blue!13}256.70}	& {\cellcolor{blue!13}5.83}	& {\cellcolor{blue!13}289.39}

     \\ \hline
\multirow{6}{*}{$2^{16}$} & \multirow{2}{*}{16} & van Baarsen \& Pu~\cite{van2025}                   & {\cellcolor{blue!13}8.14}	& 822.94	& 30.03	& 2955.02	& 117.34	& 11383.50	&*  &*

     \\
                       
                       &                    & Ours                   & 18.86	& {\cellcolor{blue!13}777.21}	& {\cellcolor{blue!13}20.44}	& {\cellcolor{blue!13}907.18}	& {\cellcolor{blue!13}22.12}	& {\cellcolor{blue!13}1039.23}	& {\cellcolor{blue!13}25.85}	& {\cellcolor{blue!13}1170.16}

      \\ \cline{2-11}
                       & \multirow{2}{*}{32} & van Baarsen \& Pu~\cite{van2025}                   &{\cellcolor{blue!13} 16.55}	& 1637.29	& 60.64	& 5899.39	& 263.34	& 22702.60	& *	& *

    \\
                       
                       &                    & Ours                   & 33.28	& {\cellcolor{blue!13}1535.14}	& {\cellcolor{blue!13}38.38}	& {\cellcolor{blue!13}1795.25}	& {\cellcolor{blue!13}41.36}	& {\cellcolor{blue!13}2059.38}	& {\cellcolor{blue!13}47.27}	& {\cellcolor{blue!13}2321.38}

      \\ \cline{2-11}
                       & \multirow{2}{*}{64} & van Baarsen \& Pu~\cite{van2025}                   & {\cellcolor{blue!13}33.01}	& 3266.60	& 126.18	& 11791.90 & *	& *	& *	& *

   \\
                       
                       &                    & Ours                   & 72.69	& {\cellcolor{blue!13}3051.35}	& {\cellcolor{blue!13}73.45}	& {\cellcolor{blue!13}3571.85}	& {\cellcolor{blue!13}82.49}	& {\cellcolor{blue!13}4100.39}	& {\cellcolor{blue!13}91.7}	& {\cellcolor{blue!13}4624.81}

  \\ \hline 
  \multirow{4}{*}{$2^{18}$} & \multirow{2}{*}{16} & van Baarsen \& Pu~\cite{van2025}                   & {\cellcolor{blue!13}35.68}	& 3293.60	& 130.30	& 11830.90	&*	&*	&*	&*

     \\
                       
                       &                    & Ours                   & 83.30	& {\cellcolor{blue!13}3111.16}	& {\cellcolor{blue!13}92.05}	& {\cellcolor{blue!13}3631.66}	& {\cellcolor{blue!13}102.04}	& {\cellcolor{blue!13}4160.21}	& {\cellcolor{blue!13}106.12}	& {\cellcolor{blue!13}4684.62}

      \\ \cline{2-11}
                       & \multirow{2}{*}{32} & van Baarsen \& Pu~\cite{van2025}                   & {\cellcolor{blue!13}77.92}	& 6553.38	& 341.46	& 23560.10	&*	&*	& *	& *

    \\
                       
                       &                    & Ours                   & 150.70	& {\cellcolor{blue!13}6152.83}	& {\cellcolor{blue!13}168.78}	& {\cellcolor{blue!13}7193.88}	& {\cellcolor{blue!13}183.38}	& {\cellcolor{blue!13}8252.13}	& {\cellcolor{blue!13}214.89}	& {\cellcolor{blue!13}9301.29}
                  
  \\ \hline
\end{tabular}%
}
\end{table*}

\section{{Details of Fuzzy PSI for High-dimensional Sets}}
\label{Appendix: Details of Fuzzy PSI for High-dimensional Sets}

% We present efficient fuzzy PSI protocols for high-dimensional sets.
% In this setting, the above protocols incur large overhead due to the factor $2^d$, which is a fundamental limitation of most spatial hashing–based approaches~\cite{garimella2022structure, garimella2024computation, bui2025new, richardson2024fuzzy, van2024fuzzy, van2025, piske2025distance}.
% To overcome this, we adopt the same assumption as in~\cite{van2025}, namely that the projections of both the sender’s and receiver’s balls are globally disjoint across all dimensions.
% As illustrated in~\cite{van2025}, this assumption is rather strong, \blue{and we discuss how to extend our protocols to the milder assumption in Section~\ref{subsec: Extension to Mild Assumptions}.}
% Leveraging our asymptotically efficient fuzzy matching, the resulting fuzzy PSI protocols achieve a logarithmic dependence on the distance threshold~$\delta$, an exponential improvement over the state-of-the-art protocol in~\cite{van2025}.

% \footnote{While~\cite{van2025} shows that with prefix-trie techniques, their protocol can in principle achieve $\log \delta$ dependence but with  $(\log \delta)^{d/2}$ computation cost, which becomes prohibitively large for high-dimensional sets.}. 

% \meng{write BP25 has security issue and we follow framework but address security issues.}

{In this section, we introduce the details of our fuzzy PSI for high-dimensional sets. }

\textbf{Fuzzy PSI for $L_\infty$ distance.}
We compile our OPRF-based fuzzy matching protocol to construct fuzzy PSI for $L_\infty$ distance.
% We do not adopt our OT-based construction, since without spatial hashing it does not admit a sufficiently reduced input domain, leading to large overhead especially for inputs with large bit length.
Our protocol follows a similar high-level framework to that of~\cite{van2025}; however, a key distinction is that our construction supports prefix-trie techniques in combination with our asymptotically efficient fuzzy matching.
As discussed in~\cite{van2025}, prefix-trie techniques are not incorporated in their protocol because doing so incurs a computation cost of $(\log \delta)^{d/2}$, which becomes prohibitively large in high-dimensional settings.

Recall that in our OPRF-based fuzzy matching protocol, for each dimension $k$, the OPPRF programs (a prefix representation of) the interval $[q_k - \delta, q_k + \delta]$ to a single random value $r_k$. To extend this construction to multi-point fuzzy PSI, the natural approach is to perform OPPRF evaluations over all sender and receiver points.
However, there are two important requirements. First, the intervals $[q_k - \delta, q_k + \delta]$ and $[q'_k - \delta, q'_k + \delta]$ corresponding to two distinct sender points cannot overlap,
since OPPRF cannot program the same input to two different values
$r_k$ and $r'_k$.
Second, any two receiver queries $w_k$ and $w'_k$ must not fall into the same programmed interval $[q_k - \delta, q_k + \delta]$;
otherwise, the receiver would obtain the same value $r_k$ for both queries, which leaks additional information about partial matches on individual dimensions.
Fortunately, the globally disjoint assumption satisfies these requirements.
Under this assumption, after invoking our role-reversed OPPRF,
the sender obtains an identifier $\mathsf{id}_{\vecq_j}$ for each
$\vecq_j \in Q$, and the receiver obtains an identifier
$\mathsf{id}_{\vecw_i}$ for each $\vecw_i \in W$, such that
$\mathsf{id}_{\vecq_j} = \mathsf{id}_{\vecw_i}$ if and only if
$\mathsf{dist}_\infty(\vecq_j, \vecw_i) \le \delta$.
Finally, both parties engage in a labeled PSI protocol (see Section~\ref{subsec: Weak Labeled PSI}), which enables the receiver to obtain $\vecq_j$ corresponding to the matched identifiers.

% Specifically, on the sender side, the globally disjoint assumption presents that for any $\vecq, \vecq^\prime$, $[q_{k} - \delta, q_{k} + \delta]$ and $[q_{k}^\prime - \delta, q_{k}^\prime + \delta]$ are disjoint for all $k \in[d]$.
% This enables the sender to program all intervals in one OPPRF.
% On the receiver side, this assumption ensures that the receiver will not query two nearby dimension values, since for any two points $\vecw, \vecw^\prime$, it holds that $w_{k}, w_{k}^\prime$ are at least $2\delta$ apart for all $k \in[d]$.
% For the prefix presentation, we follow the same idea as fuzzy matching to address the overhead issue caused by $O(\log \delta)$ OPPRF evaluations.

% Therefore, our fuzzy PSI protocol incurs a factor of $O(m+n)$ overhead increase compared to our OPRF-based fuzzy matching protocol.

\textbf{Fuzzy PSI for $L_p$ distance.}
We further extend the above protocol to $L_p$ distance.
The above protocol for $L_\infty$ distance before the last weak label PSI can be viewed as a fuzzy mapping operation.
This operation assign an unique identifier $\mathsf{id}_{\vecq_j}$ for each $\vecq_j$, similarly $\mathsf{id}_{\vecw_i}$ for each $\vecw_i$, satisfying $\mathsf{id}_{\vecq_j} = \mathsf{id}_{\vecw_i}$ if $\mathsf{dist}_{\infty}(\vecw_i, \vecq_j) \le \delta$.
To precisely determine the $L_p$ distance, both parties need to execute fuzzy matching on the pairs sharing the same identifier. 
{Since our OPRF-based fuzzy matching can not support $L_p$ distance, we use our OT-based fuzzy matching and follow the Cuckoo-hashing-based strategy of our fuzzy PSI for low-dimensional sets.}

We further optimize the input domain used in the OT-based fuzzy matching phase.
Our key observation is that for any points $\vecq_j$ and $\vecw_i$, having the same identifier $\mathsf{id}_{\vecq_j} = \mathsf{id}_{\vecw_i}$ implies $\mathsf{dist}_{\infty}(\vecw_i, \vecq_j) \le \delta$.
Consequently, for every dimension $k \in [d]$, it holds that
$|w_{i,k} - q_{j,k}| \le \delta$.
As shown in our OT-based fuzzy matching protocol in
Figure~\ref{Prot: fmatch-p-2pc}, the first step computes
$t_k = w_{i,k} - q_{j,k}$, which lies in the range $[-\delta, \delta]$.
Therefore, we can represent its secret shares using a smaller bit length.
This optimization significantly reduces the overhead of subsequent operations whose complexity depends on the bit length of the secret sharing.
We follow a strategy similar to our fuzzy PSI protocol for low-dimensional sets in Figure \ref{Prot: fpsi-p-low-px} to address false positives caused by the reduced input domain.

\section{Additional Experimental Results}
\label{appendix: Additional Experimental Results}

We provide additional experimental results for $L_p$ distance. Specifically, Tables \ref{table: L1} and \ref{table: L2} report the performance comparisons of fuzzy PSI over low-dimensional sets for $L_1$ and $L_2$ distance, respectively.
Tables \ref{table: L_1 high} and \ref{table: L_2 high} report the performance comparisons of fuzzy PSI over high-dimensional sets for $L_1$ and $L_2$ distance, respectively.

% Table \ref{table: L1} reports the performance comparisons of fuzzy PSI over low-dimensional sets for $L_1$ distance.

% Table \ref{table: L2} reports the performance comparisons of fuzzy PSI over low-dimensional sets for $L_2$ distance.

% Table \ref{table: L_1 high} reports the performance comparisons of fuzzy PSI over high-dimensional sets for $L_1$ distance.

% Table \ref{table: L_2 high} reports the performance comparisons of fuzzy PSI over high-dimensional sets for $L_2$ distance.

\section{Security Proofs}
\label{appendix: Security Proofs}

\subsection{Proof of Theorem \ref{thm: fuzzy-match-oprf}}
\label{Proof of thm: fuzzy-match-oprf}

\begin{proof}
\label{proof: fuzzy-match-oprf}
{We first establish the correctness of the protocol by analyzing the prefix matching and parameter bounds, and then construct simulators to demonstrate that the simulated execution is indistinguishable from the real-world protocol execution.}

 \textbf{Correctness.} 
{We first consider the matching case, where it holds that $\mathsf{dist}_\infty(\vecq, \vecw) \leq\delta$. 
This implies that $w_k \in [q_k-\delta, q_k+\delta]$ for all dimensions $k \in [d]$. 
By Theorem~\ref{theorem: prefix}, for each $k \in [d]$, there exists a unique index $t_k^* \in [u^\prime]$ such that $\tilde{w}_{k,t_k^*} \in \{\tilde{q}_{k,t}\}_{t\in[u]}$. 
By the correctness of the first $\Func[OPPRF]$, the receiver successfully learns $\hat{r}_{k, t_k^*} = r_k$, while the remaining evaluations $\{\hat{r}_{k, t}\}_{t \neq t_k^*}$ are independent and uniformly random. 
To ensure the second $\Func[OPPRF]$ is correctly programmed, all elements in $\{\hat{r}_{k, t}\}_{t \in [u^\prime]}$ must be distinct. As analyzed in the mismatch case below, setting $\ell_1$ appropriately guarantees this distinctness with overwhelming probability. 
Consequently, by the correctness of the second $\Func[OPPRF]$, the sender learns $s_k = \hat{s}_k$ for all $k \in [d]$. 
This yields $s = \hat{s}$, and by the correctness of $\Func[PEQT]$, the receiver obtains the correct output $1$.}

{We now consider the mismatch case where $\mathsf{dist}_\infty(\vecq, \vecw) > \delta$. 
This implies there exists at least one dimension $k^*$ such that $w_{k^*} \notin [q_{k^*}-\delta, q_{k^*}+\delta]$. 
By Theorem~\ref{theorem: prefix}, the corresponding prefix sets are disjoint: $\{\tilde{q}_{k^*,t}\}_{t\in[u]} \cap \{\tilde{w}_{k^*,t}\}_{t\in[u^\prime]} = \emptyset$. 
Therefore, by the correctness of $\Func[OPPRF]$, all evaluations $\{\hat{r}_{k^*,t}\}_{t \in [u^\prime]}$ obtained by the receiver are uniformly random and independent of uniformly random $r_{k^*}$. 
To ensure correctness, we must ensure two conditions: (1) no value in $\{\hat{r}_{k^*, t}\}_{t \in [u^\prime]}$ accidentally equals $r_{k^*}$ to avoid false positives, and (2) all values in $\{\hat{r}_{k^*,t}\}_{t \in [u^\prime]}$ are distinct, preventing collisions when programming the second $\Func[OPPRF]$.
In a single dimension, both requirements hold except with probability at most $(u^\prime + 1)^2 / 2^{\ell_1}$. 
% The probability that both requirements hold in a single dimension is at least $1 - (u^\prime + 1)^2 / 2^{\ell_1}$. 
By setting the field size to $\ell_1 := \lambda + 2\log (u^\prime + 1) + \log d$, a union bound guarantees these conditions hold simultaneously across all $d$ dimensions except with negligible probability in the security parameter $\lambda$. 
Under this parameterization, the evaluation query $r_{k^*}$ is unprogrammed in the second $\Func[OPPRF]$, resulting in $s_{k^*}$ being uniformly random. 
Because the dimensions are evaluated independently, the aggregate sum $s = \sum_{k\in[d]} {s}_{k}$ and $\hat{s} = \sum_{k\in[d]} \hat{s}_{k}$ are also uniformly random.
Setting $\ell_2 := \lambda$ ensures the probability of an accidental collision $s = \hat{s}$ is bounded by $1/2^\lambda$, which is negligible. 
Thus, $\Func[PEQT]$ outputs $0$, completing the correctness analysis.}

 \textbf{Corrupted Sender.} 
{We construct a simulator $\Sim[\SSS]$ for a semi-honest corrupted sender. Given the sender's input $\vecq$, $\Sim[\SSS]$ runs the adversary $\AAA$ and proceeds as follows:}
\begin{enumerate}
    \item {\textit{Simulating $\Func[OPPRF]$ (Step 2):} $\Sim[\SSS]$ receives the programmed points $\{(k\|\tilde{q}_{k, t}, r_{k})\}_{k \in [d], t \in [u]}$ generated by $\AAA$. \Sim[\SSS] samples a random function $F$ such that $F(k \| \tilde{q}_{k, t}) = r_{k}$ for $k \in [d], t \in [u]$, and sends $\OOO^F$ to \AAA.}

    \item {\textit{Simulating $\Func[OPPRF]$ (Step 3):} After receiving $\{k\|r_{k}\}_{k \in [d]}$ from \AAA, \Sim[\SSS] samples $\{s_{ k}\}_{ k \in [d]} \from \FF_{2^{\ell_2}}$, and sends them to \AAA.}
    
    \item {\textit{Simulating $\Func[PEQT]$ (Step 5):} $\Sim[\SSS]$ receives $s$ from $\AAA$, and returns nothing.}
\end{enumerate}

% \blue{We argue the simulated view of $\AAA$ is identically distributed to the real execution. The first OPPRF transcript consists solely of $\OOO^F$, which is identically distributed to the real execution by definition of the OPPRF functionality. For the second OPPRF, in the real protocol, each value $s_k$ is either equal to the receiver's choice $\hat{s}_k$ (if a match occurs) or is a random function output. Because the honest receiver samples $\hat{s}_k \leftarrow \FF_{2^{\ell_2}}$ uniformly at random, and because evaluations are distributed uniformly over the field, $s_k$ is always a uniformly random value independent of the sender's view. $\Sim[\SSS]$ perfectly replicates this distribution by explicitly sampling uniform values in $\FF_{2^{\ell_2}}$.}

{We demonstrate the security of the simulation through a sequence of hybrid games:}

\begin{itemize}
    \item {$\mathsf{Hyb}_0$: This is the real protocol execution.}
    
    \item {$\mathsf{Hyb}_1$: This hybrid behaves identically to $\mathsf{Hyb}_0$, except that instead of passing the honest receiver's input $\hat{s}$ to $\Func[PEQT]$ (Step 5), \Sim[\SSS] receives $s$ from $\AAA$ and returns nothing, which perfectly simulates the $\Func[PEQT]$'s behavior. Thus, the views in $\mathsf{Hyb}_0$ and $\mathsf{Hyb}_1$ are perfectly identical.}
    
    \item {$\mathsf{Hyb}_2$: This hybrid behaves identically to $\mathsf{Hyb}_1$, except that instead of passing the honest receiver's programmed points to  $\Func[OPPRF]$ (Step 3), \Sim[\SSS] receives $\AAA$'s queries $\{k\|r_k\}_{k \in [d]}$ and returns uniformly random values $s_k \leftarrow \FF_{2^{\ell_2}}$ for all $k \in [d]$ to \AAA.
    In $\mathsf{Hyb}_1$, $s_k$ is either equal to the receiver's choice $\hat{s}_k$ (upon a match) or is a random function output. Because the honest receiver samples $\hat{s}_k \leftarrow \FF_{2^{\ell_2}}$ uniformly at random, every $s_k$ in $\mathsf{Hyb}_1$ is perfectly uniform and independent of \AAA's view. Thus, the views in $\mathsf{Hyb}_1$ and $\mathsf{Hyb}_2$ are perfectly identical.}
    
    \item {$\mathsf{Hyb}_3$: This hybrid behaves identically to $\mathsf{Hyb}_2$, except that instead of passing honest receiver's inputs $\{k\|\tilde{w}_{k, t}\}_{k \in [d], t \in [u^\prime]}$ to $\Func[OPPRF]$ (Step 2), \Sim[\SSS] receives the programmed points $\{(k\|\tilde{q}_{k, t}, r_{k})\}_{k \in [d], t \in [u]}$ from $\AAA$, samples a random function $F$ such that $F(k \| \tilde{q}_{k, t}) = r_{k}$ for $k \in [d], t \in [u]$, and sends $\OOO^F$ to \AAA. Since the real functionality also returns $\OOO^F$ according to these programmed points, the views in $\mathsf{Hyb}_2$ and $\mathsf{Hyb}_3$ are perfectly identical.}
\end{itemize}

{The final hybrid $\mathsf{Hyb}_3$ is exactly the simulated execution of \Sim[\SSS]. Therefore, the simulated view is perfectly indistinguishable from the real execution.}

 \textbf{Corrupted Receiver.} 
{We construct a simulator $\Sim[\RRR]$ for a semi-honest corrupted receiver. Given the receiver's input $\vecw$ and the ideal output $z$, $\Sim[\RRR]$ runs $\AAA$ and proceeds as follows:}
\begin{enumerate}
    \item {\textit{Simulating $\Func[OPPRF]$ (Step 2):} After receiving the $\AAA$'s queries $\{k\|\tilde{w}_{k, t}\}_{k \in [d], t \in [u^\prime]}$, $\Sim[\RRR]$ samples $\{\hat{r}_{k, t}\}_{k \in [d], t \in [u^\prime]} \leftarrow \FF_{2^{\ell_1}}$, and returns them to $\AAA$.}
    \item {\textit{Simulating $\Func[OPPRF]$ (Step 3):} $\Sim[\RRR]$ receives the programmed points $\{(k\|\hat{r}_{k, t}, \hat{s}_{k})\}_{k \in [d], t \in [u^\prime]}$ from $\AAA$. $\Sim[\RRR]$ samples a random function $F^\prime$ such that $F^\prime(k\|\hat{r}_{k, t}) = \hat{s}_{k}$ for $k \in [d], t \in [u^\prime]$ and returns $\OOO^{F^\prime}$ to $\AAA$.}
    \item {\textit{Simulating $\Func[PEQT]$ (Step 5):} $\Sim[\RRR]$ receives the aggregate sum $\hat{s}$ from $\AAA$, and returns the ideal output $z$ to $\AAA$.}
\end{enumerate}

% \blue{We argue the simulated view of $\AAA$ is identically distributed to the real execution. 
% In Step 2 of the real protocol, $\AAA$ receives either the sender's random payload $r_k$ or a random evaluation of the function $F$. Because the honest sender samples $r_k \leftarrow \FF_{2^{\ell_1}}$ uniformly at random, 
% % and as shown in the correctness analysis, collisions occur with negligible probability, 
% the distribution of the evaluations $\{\hat{r}_{k, t}\}$ is uniformly random. $\Sim[\RRR]$ perfectly captures this by explicitly providing uniform values. 
% In Step 3, the second OPPRF transcript consists solely of $\OOO^{F^\prime}$, which is identically distributed to the real execution by definition of the OPPRF functionality.
% In Step 5, the $\Func[PEQT]$ output in the real world corresponds exactly to the ideal functionality's boolean output $z$. Thus, the proof is complete.}

{We demonstrate the security of the simulation through a sequence of hybrid games:}

\begin{itemize}
    \item {$\mathsf{Hyb}_0$: This is the real protocol execution.}
    
    \item {$\mathsf{Hyb}_1$: This hybrid behaves identically to $\mathsf{Hyb}_0$, except that instead of passing the honest sender's input ${s}$ to $\Func[PEQT]$ (Step 5), \Sim[\RRR] receives $\AAA$'s aggregate sum $\hat{s}$ and directly returns the ideal output $z$ to \AAA. 
    The only difference between $\mathsf{Hyb}_0$ and $\mathsf{Hyb}_1$ occurs if the protocol produces an incorrect evaluation. As established in the correctness analysis, the probability is statistically negligible.
    Thus, the views in $\mathsf{Hyb}_0$ and $\mathsf{Hyb}_1$ are statistically indistinguishable.}

    % We modify the simulation of $\Func[PEQT]$ (Step 5). Instead of computing an actual equality test, the simulator receives $\AAA$'s aggregate sum $\hat{s}$ and directly returns the ideal output $z^*$ obtained from $\Func[FMatch]$. 
    % The only difference between $\mathsf{Hyb}_2$ and $\mathsf{Hyb}_3$ occurs if the protocol produces an incorrect evaluation (e.g., $z^* = 0$ but $s = \hat{s}$). As established in the correctness analysis, the probability of false positives due to key collisions or accidental aggregate sum collisions is strictly bounded by $2^{-\lambda+1}$. Therefore, the views in these two hybrids are statistically indistinguishable ($\mathsf{Hyb}_2 \approx_s \mathsf{Hyb}_3$).

    \item {$\mathsf{Hyb}_2$:  This hybrid behaves identically to $\mathsf{Hyb}_1$, except that instead of passing the honest sender's queries to $\Func[OPPRF]$ (Step 3), \Sim[\RRR] receives the programmed points $\{(k\|\hat{r}_{k, t}, \hat{s}_{k})\}_{k\in[d], t\in[u^\prime]}$ from $\AAA$, samples a random function $F^\prime$ such that $F^\prime(k\|\hat{r}_{k, t}) = \hat{s}_{k}$ for $k \in [d], t \in [u^\prime]$ and returns $\OOO^{F^\prime}$ to $\AAA$. Since $\mathsf{Hyb}_1$ also returns $\OOO^{F^\prime}$ according to these programmed points, the views in $\mathsf{Hyb}_1$ and $\mathsf{Hyb}_2$ are perfectly identical.}

    % We modify the simulation of the second $\Func[OPPRF]$ (Step 3). The simulator receives the programmed points $\{(k\|\hat{r}_{k, t}, \hat{s}_{k})\}$ from $\AAA$ and returns a dummy handle $\OOO^{F^\prime}$. This perfectly mimics the ideal functionality ($\mathsf{Hyb}_1 \equiv \mathsf{Hyb}_2$).
    \item {$\mathsf{Hyb}_3$: This hybrid behaves identically to $\mathsf{Hyb}_2$, except that instead of passing the honest sender's programmed points to  $\Func[OPPRF]$ (Step 2), \Sim[\RRR] receives $\AAA$'s queries $\{k\|\tilde{w}_{k, t}\}_{k \in [d], t \in [u^\prime]}$ and returns uniformly random values $\{\hat{r}_{k, t}\}_{k\in[d], t\in[u^\prime]} \leftarrow \FF_{2^{\ell_1}}$ to \AAA.
    In $\mathsf{Hyb}_2$, $\AAA$ receives either the sender's payload $r_k$ or a random function evaluation. Because the honest sender samples $\{r_k\}_{k \in [d]} \leftarrow \FF_{2^{\ell_1}}$ uniformly at random, every returned evaluation is uniformly distributed. Thus, the views in $\mathsf{Hyb}_2$ and $\mathsf{Hyb}_3$ are perfectly identical.}
    
    % We modify the simulation of $\Func[PEQT]$ (Step 5). Instead of computing an actual equality test, the simulator receives $\AAA$'s aggregate sum $\hat{s}$ and directly returns the ideal output $z^*$ obtained from $\Func[FMatch]$. 
    % The only difference between $\mathsf{Hyb}_2$ and $\mathsf{Hyb}_3$ occurs if the protocol produces an incorrect evaluation (e.g., $z^* = 0$ but $s = \hat{s}$). As established in the correctness analysis, the probability of false positives due to key collisions or accidental aggregate sum collisions is strictly bounded by $2^{-\lambda+1}$. Therefore, the views in these two hybrids are statistically indistinguishable ($\mathsf{Hyb}_2 \approx_s \mathsf{Hyb}_3$).
\end{itemize}

{The final hybrid $\mathsf{Hyb}_3$ is exactly the simulated execution of \Sim[\RRR]. Therefore, the simulated view is statistically indistinguishable from the real execution.}

\end{proof}

% \noindent \textbf{Security against a Corrupted Sender.} 

% \noindent \textbf{Security against a Corrupted Receiver.} 

\subsection{Proof of Theorem \ref{thm: fuzzy-match-ot}}
\label{Proof of thm: fuzzy-match-ot}

\begin{proof}
\label{proof: fuzzy-match-ot}
{We first establish the correctness of the protocol, and then construct simulators to demonstrate that the simulated execution is perfectly indistinguishable from the real-world protocol execution.}

 \textbf{Correctness.} {The correctness of the protocol relies on accurately extracting the Most Significant Bit (MSB) of the terms $\delta - (q_k - w_k)$ and $\delta + (q_k - w_k)$. Assuming the inputs $q_k, w_k$, and the threshold $\delta$ are strictly bounded by $z$, the values of these two terms fall within the interval $(-z, 2z)$. To prevent wrap-around in the two's complement representation, the ring $\ZZ_{2^\ell}$ must be parameterized such that $\ell \ge \log z + 2$. Under this condition, these MSBs are correctly evaluated, and the overall correctness of the protocol directly follows from the correctness of the underlying $\Func[MSB]$ and $\Func[AND]$ functionalities.}

 \textbf{Corrupted Sender.} {We construct a simulator $\Sim[\SSS]$ for a semi-honest corrupted sender. Given the sender's input $\vecq$, $\Sim[\SSS]$ operates by running the adversary $\AAA$ and proceeds as follows:}
\begin{enumerate}
    \item {\textit{Simulating $\Func[MSB]$ (Step 2):} After receiving $\delta - [t_k]_\SSS^{\ell}$ from \AAA, $\Sim[\SSS]$ samples uniformly random shares $[a_k]_\SSS^1 \leftarrow \ZZ_2$ for all $k \in [d]$, and sends them to $\AAA$.}
    
    \item {\textit{Simulating $\Func[MSB]$ (Step 3):} After receiving $\delta + [t_k]_\SSS^{\ell}$ from \AAA, $\Sim[\SSS]$ samples uniformly random shares $[b_k]_\SSS^1 \leftarrow \ZZ_2$ for all $k \in [d]$, and sends them to $\AAA$.}
    
    \item {\textit{Simulating $\Func[AND]$ (Step 4):} After receiving \AAA's inputs, $\Sim[\SSS]$ samples a random share $[z]_\SSS^1 \leftarrow \ZZ_2$ and sends it to $\AAA$.}
\end{enumerate}

{We argue the simulated view of $\AAA$ is identically distributed to the real execution. Since the protocol operates in the $(\Func[MSB], \Func[AND])$-hybrid model, the messages received by $\AAA$ consist solely of its own secret shares. In the real execution, the shares $\{[a_k]^1_{\SSS}\}_{k \in [d]}$, $\{[b_k]^1_{\SSS}\}_{k \in [d]}$, and $[z]^1_{\SSS}$ are generated as uniformly random secret shares of the underlying functionality outputs. In the simulated execution, $\Sim[\SSS]$ explicitly samples them from the uniform distribution over $\ZZ_2$. Thus, the simulated view of $\AAA$ is identically distributed to its view in the real protocol.}

 \textbf{Corrupted Receiver.} {We then construct a simulator $\Sim[\RRR]$ for a semi-honest corrupted receiver. Given the receiver's input $\vecw$ and the ideal output $z$ from $\Func[FMatch]$, $\Sim[\RRR]$ runs $\AAA$ and proceeds as follows:}
\begin{enumerate}
    \item {\textit{Simulating $\Func[MSB]$ (Step 2):} After receiving $\delta - [t_k]_\RRR^{\ell}$ from \AAA, $\Sim[\RRR]$ samples uniform random shares $[a_k]_\RRR^1 \leftarrow \ZZ_2$ for all $k \in [d]$ and sends them to $\AAA$.}
    
    \item {\textit{Simulating $\Func[MSB]$ (Step 3):} After receiving $\delta + [t_k]_\RRR^{\ell}$ from \AAA, $\Sim[\RRR]$ samples uniform random shares $[b_k]_\RRR^1 \leftarrow \ZZ_2$ for all $k \in [d]$ and sends them to $\AAA$.}
    
    \item {\textit{Simulating $\Func[AND]$ (Step 4):} After receiving \AAA's inputs, $\Sim[\RRR]$ samples a random share $[z]_\RRR^1 \leftarrow \ZZ_2$ and sends it to $\AAA$.}
    
    \item {\textit{Simulating the final message (Step 5):} $\Sim[\RRR]$ computes the simulated sender's message as $[z]_\SSS^1 := [z]_\RRR^1 \oplus z \in \ZZ_2$, and sends it to $\AAA$.}
\end{enumerate}

{We argue the simulated view of $\AAA$ is identically distributed to the real execution. Similar to the corrupted sender case, the shares $\{[a_k]^1_{\RRR}\}_{k \in [d]}$, $\{[b_k]^1_{\RRR}\}_{k \in [d]}$, and $[z]^1_{\RRR}$ are uniformly random in both the real and simulated executions. The only constraint on the view is the final message $[z]^1_{\SSS}$. In the real protocol, $[z]^1_{\SSS}$ is the honest sender's share, which satisfies $[z]^1_{\SSS} \oplus [z]^1_{\RRR} = z$. In the simulation, $\Sim[\RRR]$ explicitly programs $[z]^1_{\SSS}$ to satisfy this exact relation using the ideal output $z$. Therefore, the simulated view of $\AAA$ and the real protocol view are perfectly identically distributed. Thus, the proof is complete.}
\end{proof}

\subsection{Proof of Theorem \ref{thm: fuzzy-match-p-ot}}
\label{Proof of thm: fuzzy-match-p-ot}

\begin{proof}
\label{proof: fuzzy-match-p-ot}
{We first establish the correctness of the protocol, and then construct simulators to demonstrate that the simulated execution is perfectly indistinguishable from the real-world protocol execution.}

 \textbf{Correctness.} {The correctness of the protocol relies on accurately extracting the Most Significant Bit (MSB) of the term $f := \delta^p - \sum_{k\in[d]} f_k$, where $f_k:= |q_k - w_k|^p$.
Assuming the inputs $q_k, w_k$ and the threshold $\delta$ are bounded by $z$, the sum $f$ falls within the interval $(- d \cdot z^{p}, z^{p})$.
To prevent wrap-around in the ring $\ZZ_{2^\ell}$, the ring must be parameterized such that $\ell \ge  \log_2 (d \cdot z^p) + 1$. Under this condition, the MSB accurately evaluates whether the distance exceeds the threshold, and the overall correctness directly follows from the correctness of the underlying $\Func[ABS]$, $\Func[Mult]$, and $\Func[MSB]$ functionalities.}

 \textbf{Corrupted Sender.} {We construct a simulator $\Sim[\SSS]$ for a semi-honest corrupted sender. Given the sender's input $\vecq$, $\Sim[\SSS]$ operates by running the adversary $\AAA$ and proceeds as follows:}
\begin{enumerate}
    \item {\textit{Simulating $\Func[ABS]$ or $\Func[Mult]$ (Step 2):} After receiving $[t_k]_\SSS^{\ell}$ from \AAA, $\Sim[\SSS]$ samples uniform random shares $[f_k]_\SSS^\ell \leftarrow \ZZ_{2^\ell}$ each all $k \in [d]$ and sends them to $\AAA$.}
    
    \item {\textit{Simulating $\Func[MSB]$ (Step 3):} After receiving $[f]_\SSS^{\ell}$ from \AAA, $\Sim[\SSS]$ samples a uniform random share $[z]_\SSS^1 \leftarrow \ZZ_2$ and sends it to $\AAA$.}
\end{enumerate}

{We argue the simulated view of $\AAA$ is identically distributed to the real execution. In the $(\Func[ABS], \Func[Mult], \Func[MSB])$-hybrid model, all intermediate values received by $\AAA$ are purely additive secret shares. In the real protocol, the shares $\{[f_k]_\SSS^\ell\}_{k \in [d]}$ and $[z]_\SSS^1$ are uniformly distributed. In the simulation, $\Sim[\SSS]$ explicitly samples them from the uniform distributions over $\ZZ_{2^\ell}$ and $\ZZ_2$, respectively. Thus, the simulated view of $\AAA$ is identically distributed to its view in the real protocol.}

 \textbf{Corrupted Receiver.} {We construct a simulator $\Sim[\RRR]$ for a semi-honest corrupted receiver. Given the receiver's input $\vecw$ and the ideal output $z \in \ZZ_2$ from $\Func[FMatch]$, $\Sim[\RRR]$ runs $\AAA$ and proceeds as follows:}
\begin{enumerate}
    \item {\textit{Simulating $\Func[ABS]$ or $\Func[Mult]$ (Step 2):} After receiving $[t_k]_\RRR^{\ell}$ from \AAA, $\Sim[\RRR]$ samples uniform random shares $[f_k]_\RRR^\ell \leftarrow \ZZ_{2^\ell}$ for each $k \in [d]$ and sends them to $\AAA$.}
    
    \item {\textit{Simulating $\Func[MSB]$ (Step 3):} After receiving $[f]_\RRR^{\ell}$ from \AAA, $\Sim[\RRR]$ samples a uniform random share $[z]_\RRR^1 \leftarrow \ZZ_2$ and sends it to $\AAA$.}
    
    \item {\textit{Simulating the final message (Step 4):} $\Sim[\RRR]$ computes the message as $m_\SSS := z \oplus [z]_\RRR^1 \in \ZZ_2$. $\Sim[\RRR]$ then sends $m_\SSS$ to $\AAA$, simulating the message $(1 \oplus [z]_\SSS^1)$ from the honest sender.}
\end{enumerate}

{We argue the simulated view of $\AAA$ is identically distributed to the real execution. Similar to the corrupted sender case, the shares $\{[f_k]_\RRR^\ell\}_{k \in [d]}$ and $[z]^1_{\RRR}$ are uniformly random in both executions. The final message $m_\SSS$ in the real protocol is $1 \oplus [z]_\SSS^1$, which perfectly satisfies the output relation $z = m_\SSS \oplus [z]_\RRR^1$. In the simulation, $\Sim[\RRR]$ directly constructs $m_\SSS$ to satisfy this exact relation using the ideal output $z$. Therefore, the simulated view of $\AAA$ and the real protocol view are perfectly identically distributed. Thus, the proof is complete.}
\end{proof}

\subsection{Proof of Theorem \ref{thm: fuzzy PSI L-p-low-px}}
\label{Proof of Theorem fuzzy PSI L-p-low-px}

\begin{proof}
{We first analyze the protocol's correctness and then prove the security of the protocol by constructing simulators for both a corrupted sender and a corrupted receiver.}

 \textbf{Correctness.}
{ 
We first consider the matching case, where for a sender point $\vecq_j \in Q$, there exists a receiver point $\vecw_i \in W$ such that $\mathsf{dist}_p(\vecw_i, \vecq_j) \le \delta$. 
By the definition of spatial hashing, there must exist a unique index $h^* \in [2^d]$ such that the computed cell identifiers match: $\CCC_{\vecw_i,h^*} = \CCC_{\vecq_j}$. 
By the unique center assumption, the sender's points are sufficiently separated such that each grid cell contains at most one sender point. Consequently, the Cuckoo hash table $T_q$ is successfully constructed with overwhelming probability. Let $j^*$ denote the bin index of the tuple $(\CCC_{\vecq_j}, \vecq_j)$ in $T_q$. 
Because both parties use identical hash functions, the receiver's corresponding tuple $(\CCC_{\vecw_i,h^*}, \vecw_i)$ is mapped to the same bin index $T_w[j^*]$. 
By the unique ball assumption, a cell intersects at most one $\delta$-radius receiver ball. Thus, the cell identifiers $\CCC_{\vecw_i,h}$ are unique for all $i \in [n], h \in [2^d]$, ensuring no key collisions exist within the programmed sets $\{L_j\}_{j \in [m^\prime]}$ provided to $\Func[OPPRF]$. 
By the correctness of $\Func[OPPRF]$, the sender queries with $(\CCC_{\vecq_j}\|\tau)$ from $T_q[j^*]$ and successfully learns the programmed payload:
\[
\hat{s}_{j^*} \,\|\, \hat{\vecr}_{j^*}
= s_{j^*} \,\|\, \bigl(
\mathsf{shift}_{\lceil \log 6\delta \rceil}
(\vecw_i, \CCC_{\vecw_i,h^*} - 2\delta) - \vecr_{j^*} \bmod 2^{\ell_2}
\bigr).
\]
Meanwhile, the receiver holds $s_{j^*}$ and $\vecr_{j^*}$ for the $j^*$-th bin of
$T_w$.

To ensure that the computation will not overflow in OT-based fuzzy matching, we set the bit length $\ell_2$ according to the analysis in Appendix~\ref{Proof of thm: fuzzy-match-ot} and Appendix~\ref{Proof of thm: fuzzy-match-p-ot}.
Given that the input is less than $6\delta$ in the reduced domain, we set $\ell_2 := \lceil \log (6\delta)\rceil + 2$ for $L_\infty$ distance and $\ell_2 := \lceil \log d (6\delta)^p \rceil + 1$ for $L_{p}$ distance.
Note that for $L_2$ distance, $\ell_2$ can be further optimized to $\lceil \log d (3\delta)^2 \rceil$, according to the fact that if the cell of $\vecq_j$ and $\vecw_i$'s ball is matched then $|q_{j, k} - w_{i, k}| < 3\delta$ for all $k \in [d]$.
Finally, by the correctness of the functionalities
\Func[ssFMatch], \Func[ssPEQT], \Func[AND], \Func[Perm], and \Func[OT],
the receiver \RRR correctly obtains $\vecq_j$.
}

{We next consider the mismatch case, where for a sender point $\vecq_j \in Q$, it holds that $\mathsf{dist}_p(\vecw_i, \vecq_j) > \delta$ for all $\vecw_i \in W$. Let $j^*$ denote its bin index in $T_q$. We distinguish two sub-cases:
\emph{Case 1 (Cell Collision, Distance Mismatch):} There exists a $\vecw_i \in W$ and $h^* \in [2^d]$ such that $\CCC_{\vecw_i,h^*} = \CCC_{\vecq_j}$. Because the distance strictly exceeds $\delta$, the correctness of $\Func[ssFMatch]$ ensures the reconstructed output is $\nu_{j^*} = 0$.
\emph{Case 2 (No Cell Collision):} For all $\vecw_i \in W$ and $h \in [2^d]$, $\CCC_{\vecw_i,h} \neq \CCC_{\vecq_j}$. The sender's query to $\Func[OPPRF]$ is unprogrammed, yielding a uniformly random value $\hat{s}_{j^*} \in \FF_{2^{\ell_1}}$. The probability of an accidental collision $\hat{s}_{j^*} = s_{j^*}$ is $2^{-\ell_1}$. Setting $\ell_1 := \lambda + \log m^\prime$ ensures that by a union bound, the probability of any such collision across all $m^\prime$ bins is bounded by $2^{-\lambda}$. Thus, except with negligible probability, $\Func[ssPEQT]$ correctly ensures the reconstructed output is $\mu_{j^*} = 0$.
In both sub-cases, the $\Func[AND]$ functionality produces shares of $0$, and the final $\Func[OT]$ delivers $\bot$ to the receiver. This completes the correctness proof.}

% To prevent overflow during the distance evaluations within $\Func[ssFMatch]$, we bound the bit length $\ell_2$. Because the shifted coordinates represent relative offsets within a cell, the maximum coordinate difference is bounded by $6\delta$. Thus, we set $\ell_2 := \lceil \log (6\delta)\rceil + 2$ for the $L_\infty$ distance, and $\ell_2 := \lceil \log d (6\delta)^p \rceil + 1$ for the $L_p$ distance. 

\textbf{Corrupted Sender.} { We construct a simulator $\Sim[\SSS]$ for a semi-honest corrupted sender. Given the sender's input $Q$, $\Sim[\SSS]$ runs the adversary $\AAA$ and proceeds as follows:}
\begin{enumerate}
    \item {\textit{Simulating $\Func[OPPRF]$ (Step 5):} After receiving the queries, $\Sim[\SSS]$ samples uniformly random values $\hat{s}_j \leftarrow \FF_{2^{\ell_1}}$ and $\hat{\vecr}_j \leftarrow \ZZ^d_{2^{\ell_2}}$ for all $j \in [m^\prime]$, and returns $\{\hat{s}_j \|\hat{\vecr}_j\}_{j \in [m^\prime]}$ to $\AAA$.}
    
    \item {\textit{Simulating $\Func[ssFMatch]$, $\Func[ssPEQT]$, and $\Func[AND]$ (Steps 6-8):} After receiving \AAA's inputs to these functionalities, $\Sim[\SSS]$ simulates the outputs by sampling uniformly random bits $\nu^\SSS_j, \mu^\SSS_j, b^\SSS_j \leftarrow \FF_2$ for all $j \in [m^\prime]$ and returns them to $\AAA$.}
    
    \item {\textit{Simulating $\Func[Perm]$ (Step 9):} After receiving $\{b_j^\SSS\}_{j \in [m^\prime]}$ from \AAA, $\Sim[\SSS]$ samples a random permutation $\pi$ over $[m^\prime]$ and returns it to $\AAA$.}
    
    \item {\textit{Simulating $\Func[OT]$ (Step 10):} After receiving OT's inputs from \AAA, $\Sim[\SSS]$ provides no output to $\AAA$.}
\end{enumerate}

% \blue{We argue that the simulated view is perfectly indistinguishable from the real execution. In Step 5 of the real protocol, by the assumptions on both parties' sets, if a valid spatial match occurs, $\hat{\vecr}_j$ is masked by the receiver's uniformly random value $\vecr_j$; if no match occurs, it is a random function evaluation. In both cases, $\hat{\vecr}_j$ (and similarly $\hat{s}_j$) is uniformly distributed over its respective domain, completely independent of the sender's input. $\Sim[\SSS]$ identically simulates this by sampling uniform values. In Steps 6-8, the ideal functionalities $\Func[ssFMatch]$, $\Func[ssPEQT]$, and $\Func[AND]$ explicitly output uniform secret shares to the sender. $\Sim[\SSS]$ faithfully generates these uniform bits. The permutation $\pi$ in Step 9 is explicitly uniform in both worlds. Thus, the views are identically distributed.}

{We demonstrate the security of the simulation through a sequence of hybrid games:}

\begin{itemize}
    \item {$\mathsf{Hyb}_0$: The real protocol execution in the hybrid model.}

    \item {$\mathsf{Hyb}_1$: This hybrid behaves identically to $\mathsf{Hyb}_0$, except that instead of invoking \Func[Perm] (Step 9) and \Func[OT] (Step 10), \Sim[\SSS] returns a random permutation $\pi$ over $[m^\prime]$ as the output of \Func[Perm], and provides no output for \AAA in \Func[OT], which perfectly simulates the functionalities' behavior. Thus, the views in $\mathsf{Hyb}_0$ and $\mathsf{Hyb}_1$ are perfectly identical.}

    \item {$\mathsf{Hyb}_2$: This hybrid behaves identically to $\mathsf{Hyb}_1$ except the simulation of \Func[ssFMatch], \Func[ssPEQT], and \Func[AND] (Steps 6-8). \Sim[\SSS] receives $\AAA$'s inputs and simply returns uniformly random bits $\nu^\SSS_j, \mu^\SSS_j, b^\SSS_j \leftarrow \FF_2$ for all $j \in [m^\prime]$. Because these ideal functionalities output uniform secret shares to \AAA in $\mathsf{Hyb}_1$, the views in $\mathsf{Hyb}_1$ and $\mathsf{Hyb}_2$ are perfectly identical.}

    \item {$\mathsf{Hyb}_3$: This hybrid behaves identically to $\mathsf{Hyb}_2$, except that instead of passing the honest receiver's programmed points to  $\Func[OPPRF]$ (Step 5), $\Sim[\RRR]$ receives $\AAA$'s queries and returns uniformly random values $\hat{s}_j \leftarrow \FF_{2^{\ell_1}}$ and $\hat{\vecr}_j \leftarrow \ZZ^d_{2^{\ell_2}}$ for all $j \in [m^\prime]$. In $\mathsf{Hyb}_2$, if a valid spatial match occurs, $\hat{\vecr}_j$ is masked by the receiver's uniformly random value $\vecr_j$, making it uniformly distributed. If no match occurs, it is a random function evaluation. The case of $\hat{s}_j$ is similar. Because the honest receiver samples all $s_j$ and $\vecr_j$ uniformly at random, every evaluation $\AAA$ receives in $\mathsf{Hyb}_2$ is perfectly uniform and independent of $\AAA$'s input. Thus, the views in $\mathsf{Hyb}_2$ and $\mathsf{Hyb}_3$ are perfectly identical.}

\end{itemize}

{The final hybrid $\mathsf{Hyb}_3$ is exactly the simulated execution of \Sim[\SSS]. Therefore, the simulated view is perfectly indistinguishable from the real execution.}

\textbf{Corrupted Receiver.} { We construct a simulator $\Sim[\RRR]$ for a semi-honest corrupted receiver. Given the receiver's input $W$ and the ideal output set $Z$ from $\Func[FPSI]$, $\Sim[\RRR]$ runs $\AAA$ and proceeds as follows:}
\begin{enumerate}

    \item {\textit{Simulating $\Func[OPPRF]$ (Step 5):} After receiving the programmed points $\{L_j\}_{j \in [m^\prime]}$ from $\AAA$, $\Sim[\RRR]$ samples a random function $F$ programmed on $\{L_j\}_{j \in [m^\prime]}$. \Sim[\RRR] sends $\OOO^F$  to \AAA.}
    
    \item {\textit{Simulating $\Func[ssFMatch]$, $\Func[ssPEQT]$, and $\Func[AND]$ (Steps 6-8):} After receiving $\AAA$'s inputs, $\Sim[\RRR]$ samples uniformly random bits $\nu^\RRR_j, \mu^\RRR_j, b^\RRR_j \leftarrow \FF_2$ for all $j \in [m^\prime]$ and returns them to $\AAA$.}
    
    \item {\textit{Simulating $\Func[Perm]$ and $\Func[OT]$ (Steps 9-10):} To simulate the final delivery of the intersection, $\Sim[\RRR]$ pads the ideal output set $Z$ with $\bot$ elements until it contains exactly $m^\prime$ items. $\Sim[\RRR]$ shuffles this padded set randomly to yield $\{z_j\}_{j \in [m^\prime]}$. For each $j \in [m^\prime]$, $\Sim[\RRR]$ defines $\hat{b}_j := 0$ if $z_j = \bot$, and $\hat{b}_j := 1$ otherwise. After receiving \AAA's inputs, $\Sim[\RRR]$ returns $\{\hat{b}_j\}_{j \in [m^\prime]}$ as the output of $\Func[Perm]$ and $\{z_j\}_{j \in [m^\prime]}$ as the output of the final $\Func[OT]$.}
\end{enumerate}

% \blue{We argue the simulated view of $\AAA$ is identically distributed to the real execution. In Step 5, the receiver obtains a programmed random function, matching the ideal functionality. In Steps 6-8, the intermediate values $\nu^\RRR_j, \mu^\RRR_j, b^\RRR_j$ are uniform secret shares in the real protocol, which $\Sim[\RRR]$ correctly models by sampling uniformly. The crucial simulation occurs in Steps 9-10. In the real protocol, $\Func[Perm]$ applies a uniformly random permutation to the underlying match bits, and $\Func[OT]$ returns the permuted intersection elements interspersed with $\bot$. Because the permutation is random and unknown to the receiver, the sequence of actual payloads and $\bot$s is uniformly distributed over all size-$m^\prime$ arrays containing exactly $|Z|$ payloads. $\Sim[\RRR]$ identically reconstructs this distribution by randomly shuffling the padded ideal output set $Z$ and setting the selection bits $\hat{b}_j$ consistently. Thus, the simulation integrates the ideal functionality's output into the receiver's view, completing the proof.}

{We demonstrate the security of the simulation through a sequence of hybrid games:}

\begin{itemize}
    \item {$\mathsf{Hyb}_0$: This is the real protocol execution.}
    
    \item {$\mathsf{Hyb}_1$: This hybrid behaves identically to $\mathsf{Hyb}_0$ except the simulation of \Func[Perm] and \Func[OT] (Steps 9-10).
    \Sim[\RRR] pads the ideal output set $Z$ with $\bot$ elements until it contains exactly $m^\prime$ items. It then randomly shuffles this padded set to yield $\{z_j\}_{j \in [m^\prime]}$. For each $j \in [m^\prime]$, \Sim[\RRR] defines $\hat{b}_j := 1$ if $z_j \neq \bot$, and $\hat{b}_j := 0$ otherwise. It returns $\{\hat{b}_j\}_{j \in [m^\prime]}$ as the output of \Func[Perm] and $\{z_j\}_{j \in [m^\prime]}$ as the output of \Func[OT]. 
    In $\mathsf{Hyb}_0$, \Func[Perm] applies a uniformly random permutation to the underlying match bits, and \Func[OT] returns the corresponding intersection elements interspersed with $\bot$s. Because the permutation is uniformly random and completely hidden from \AAA, the sequence of bits and retrieved payloads is simply a uniformly random arrangement of exactly $|Z|$ matches and $m^\prime - |Z|$ $\bot$s. 
    The only difference between $\mathsf{Hyb}_0$ and $\mathsf{Hyb}_1$ occurs if the protocol produces an incorrect evaluation. As established in the correctness analysis, the probability is statistically negligible.
    Thus, the views in $\mathsf{Hyb}_0$ and $\mathsf{Hyb}_1$ are statistically indistinguishable.}
    
    \item {$\mathsf{Hyb}_2$: This hybrid behaves identically to $\mathsf{Hyb}_1$ except the simulation of \Func[ssFMatch], \Func[ssPEQT], and \Func[AND] (Steps 6-8). \Sim[\RRR] receives $\AAA$'s inputs to these functionalities and returns uniformly random bits $\nu^\RRR_j, \mu^\RRR_j, b^\RRR_j \leftarrow \FF_2$ for all $j \in [m^\prime]$. In Steps 6-8, the intermediate values $\nu^\RRR_j, \mu^\RRR_j, b^\RRR_j$ are uniform secret shares in $\mathsf{Hyb}_1$, which $\Sim[\RRR]$ correctly models by sampling uniformly. Thus, the views in $\mathsf{Hyb}_1$ and $\mathsf{Hyb}_2$ are perfectly identical.}

    \item {$\mathsf{Hyb}_3$: This hybrid behaves identically to $\mathsf{Hyb}_2$ except the simulation of \Func[OPPRF] (Step 5).
    After receiving the programmed points $\{L_j\}_{j \in [m^\prime]}$ from $\AAA$, $\Sim[\RRR]$ samples a random function $F$ programmed on $\{L_j\}_{j \in [m^\prime]}$. \Sim[\RRR] sends $\OOO^F$ to \AAA. Since $\mathsf{Hyb}_2$ also returns $\OOO^{F}$ according to these programmed points, the views in $\mathsf{Hyb}_2$ and $\mathsf{Hyb}_3$ are perfectly identical.}

\end{itemize}

{The final hybrid $\mathsf{Hyb}_3$ is exactly the simulated execution of \Sim[\RRR]. Therefore, the simulated view is statistically indistinguishable from the real execution.}

\end{proof}

\subsection{Proof of Theorem \ref{thm: fuzzy PSI L-inf-high-px}}

\label{Proof of thm: fuzzy PSI L-inf-high-px}

\begin{proof}
\label{proof: Proof of fuzzy PSI L-inf-high-px}
{We first analyze the protocol's correctness and then prove the security of the protocol by constructing simulators for both a corrupted sender and a corrupted receiver.}
 
 \textbf{Correctness.} 
{Before analyzing correctness for different cases, we illustrate the distinctness of programmed points in two invocations of OPPRF. Below, we consider any dimension $k \in [d]$. 
(1) \textit{In the first OPPRF invocation,} the globally disjoint assumption of the sender set ensures that for any $j, j^\prime \in [m]$, $[q_{j, k}-\delta, q_{j, k}+\delta] \cap [q_{j^\prime, k}-\delta, q_{j^\prime, k}+\delta] = \emptyset$ and hence $|q_{j, k}-q_{j^\prime, k}| > 2\delta$. By Corollary~\ref{corollary: Distinctness of PxTrie}, this implies the intersection of their prefixes is empty: $\{\tilde{q}_{j, k,t}\}_{t\in[u]} \cap \{\tilde{q}_{j^\prime, k,t}\}_{t\in[u]} = \emptyset$. 
Therefore, all programmed points $\{\tilde{q}_{j, k,t}\}_{j \in [m], t\in[u]}$ of the first OPPRF are distinct.
(2) \textit{In the second OPPRF invocation,} the globally disjoint assumption of the receiver set ensures that for any $i, i^\prime \in [n]$, $[w_{i, k}-\delta, w_{i, k}+\delta] \cap [w_{i^\prime, k}-\delta, w_{i^\prime, k}+\delta] = \emptyset$ and hence $|w_{i, k}-w_{i^\prime, k}| > 2\delta$.
By Corollary~\ref{corollary: Distinctness of PxPath}, this implies the intersection of their prefixes is empty: $\{\tilde{w}_{i, k,t}\}_{t\in[u^\prime]} \cap \{\tilde{w}_{i^\prime, k,t}\}_{t\in[u^\prime]} = \emptyset$. 
Moreover, by Corollary~\ref{corollary: uniqueness}, for each sender prefix set $\{\tilde{q}_{j, k,t}\}_{t\in[u]}$, there is at most one receiver prefix set $\{\tilde{w}_{i, k,t}\}_{t\in[u^\prime]}$ with a unique $t$ intersected with it. 
Therefore, by the correctness of the first OPPRF, $\{\hat{r}_{i, k,t}\}_{i \in [n], t\in[u^\prime]}$ are independent and uniformly random, since $\hat{r}_{i, k,t}$ is either a random function value (unmatched with all programmed points) or some $r_{j, k}$ (matched with a point).
By setting appropriate bit lengths as analyzed below, all values within the set $\{\hat{r}_{i, k, t}\}_{i \in [n], t \in [u^\prime]}$ for a given dimension $k$ are distinct.
This ensures that all programmed points of the second OPPRF are distinct.
}

% By setting appropriate bitlength $\ell_1$, then all programmed points $\{\hat{r}_{i, k,t}\}_{i \in [n], t\in[u^\prime]}$ of the second OPPRF are distinct with overwhelming probability.

{
With the correct invocations of OPPRF established, we proceed to analyze the correctness for different cases.
}

{
We first consider the matching case where for a sender point $\vecq_j \in Q$, there exists $\vecw_i \in W$ such that $\mathsf{dist}_\infty(\vecw_i, \vecq_j) \leq \delta$. 
This implies that $w_{i,k} \in [q_{j,k}-\delta, q_{j,k}+\delta]$ for all dimensions $k \in [d]$. 
By Theorem~\ref{theorem: prefix}, for each dimension $k \in [d]$, there exists a unique index $t^*_k \in [u^\prime]$ such that $\tilde{w}_{i, k,t^*_k} \in \{\tilde{q}_{j, k, t}\}_{t\in[u]}$. 
By the correctness of the first $\Func[OPPRF]$, the receiver successfully learns $\hat{r}_{i, k, t^*_k} = r_{j, k}$.
Similarly, by the correctness of the second $\Func[OPPRF]$, the sender learns $s_{j, k} = \hat{s}_{i, k}$.
This results in the aggregate sums matching $s_j = \hat{s}_i$.
Finally, by the correctness of $\Func[wLPSI]$, the receiver queries $\hat{s}_i$, and successfully obtains $s_j$ and the associated payload $\vecq_j$.
}
% , while as illustrated above the remaining evaluations $\{\hat{r}_{i, k, t}\}_{t \neq t^*}$ are uniformly random and independent. 
% By the globally disjoint assumption, the intervals $[q_{j,k}-\delta, q_{j,k}+\delta]$ are strictly non-overlapping for all $j \in [m]$. 
% Consequently, $\tilde{w}_{i, k,t^*}$ matches the prefix of exactly one sender point $\vecq_j$ in dimension $k$. 
% Because the inputs to the second $\Func[OPPRF]$ are well-formed (as shown in the collision bounds below, all $\hat{r}$ values are distinct), the sender learns $s_{j, k} = \hat{s}_{i, k}$ for all $k \in [d]$. 
% This results in the aggregate sums matching ($s_j = \hat{s}_i$). 

{
We next consider the mismatch case, where for a sender point $\vecq_j \in Q$, it holds that $\mathsf{dist}_\infty(\vecw_i, \vecq_j) > \delta$ for all $\vecw_i \in W$. 
This yields that for any $\vecw_i$, there exists at least one dimension $k^*$ such that $w_{i,k^*} \notin [q_{j, k^*}-\delta, q_{j, k^*}+\delta]$.
By Theorem~\ref{theorem: prefix}, this implies that the intersection of their prefixes is empty: $\{\tilde{q}_{j, k^*,t}\}_{t\in[u]} \cap \{\tilde{w}_{i, k^*,t}\}_{t\in[u^\prime]} = \emptyset$.
By the correctness of the first $\Func[OPPRF]$, all evaluations $\{\hat{r}_{i, k^*, t}\}_{t \in [u^\prime]}$ obtained by the receiver are uniformly random and are independent of $r_{j, k^*}$.
To ensure correctness of the entire protocol, we must guarantee two conditions: (1) no random function evaluation $\{\hat{r}_{i, k^*, t}\}_{t \in [u^\prime]}$ accidentally collides with any of the sender's payloads $r_{j, k^*}$ (preventing false positives), and (2) all values within the set $\{\hat{r}_{i, k^*, t}\}_{i \in [n], t \in [u^\prime]}$ are strictly distinct, ensuring the receiver does not provide duplicate keys when programming the second $\Func[OPPRF]$. 
Considering the execution of OPPRF on the two sets, the total number of points drawn in each dimension is bounded by $(n u^\prime + m)$. The probability of any collision occurring among these points is bounded by $(n u^\prime + m)^2 / 2^{\ell_1}$. 
Considering all $d$ dimensions, setting $\ell_1 := \lambda + 2\log (n u^\prime + m) + \log d$ ensures the probability of a collision is bounded by $1/2^\lambda$, which is negligible. 
Therefore, $r_{j, k^*} \notin \{\hat{r}_{i, k^*, t}\}_{t \in [u^\prime]}$ with overwhelming probability.

Under the above condition, by the correctness of the second $\Func[OPPRF]$, {the sender receives a uniformly random $s_{j, k^*}$}, independent to $\hat{s}_{i, k^*}$, since $s_{j, k^*}$ is either a random value (unmatched with all programmed points) or some other $\hat{s}_{i^\prime, k^*}$ (matched with a point from $\hat{r}_{i^\prime, k^*, t}$).
Because dimensions are evaluated independently, and because for each $i \in [n]$, there exists $k^*$ such that $s_{j, k^*}$ and $\hat{s}_{i, k^*}$ are uniformly random and independent, the aggregate sums $s_j$ and $\hat{s}_i$ are uniformly random and independent. Setting $\ell_2 := \lambda + \log m + \log n$ guarantees that the probability of an accidental collision $s_j = \hat{s}_i$ across all $m \times n$ possible pairs is bounded by $(m n)/2^{\ell_2} \le 1/2^\lambda$. 
Thus, the $\Func[wLPSI]$ query fails, returning nothing for this mismatch, completing the correctness proof.}

 \textbf{Corrupted Sender.} {We construct a simulator $\Sim[\SSS]$ for a semi-honest corrupted sender. Given the sender's input $Q$, $\Sim[\SSS]$ runs the adversary $\AAA$ and proceeds as follows:}
\begin{enumerate}
    \item {\textit{Simulating $\Func[OPPRF]$ (Step 3):} $\Sim[\SSS]$ receives the programmed points $\{(k\|\tilde{q}_{j, k, t}, r_{j, k})\}_{j \in [m], k \in [d], t \in [u]}$ generated by $\AAA$. $\Sim[\SSS]$ samples a random function $F$ such that $F(k\|\tilde{q}_{j, k, t}) = r_{j, k}$ for $j \in [m], k \in [d], t \in [u]$, and returns $\OOO^F$ to $\AAA$.}
    
    \item {\textit{Simulating $\Func[OPPRF]$ (Step 4):} After receiving the \AAA's queries $\{k\|r_{j, k}\}_{j \in [m], k \in [d]}$, $\Sim[\SSS]$ samples uniformly random values $s_{j, k} \leftarrow \FF_{2^{\ell_2}}$ for all $j \in [m], k \in [d]$, and sends them to $\AAA$.}
    
    \item {\textit{Simulating $\Func[wLPSI]$ (Step 6):} After receiving $\{(s_j, \vecq_j)\}_{j \in [m]}$, $\Sim[\SSS]$ returns nothing.}
\end{enumerate}

% \blue{We argue the simulated view of $\AAA$ is identically distributed to the real execution. In the first OPPRF, the sender receives $\OOO^F$ according to the sender's input, perfectly matching the simulated step. In the second OPPRF, as illustrated in the correctness analysis, the real protocol dictates that the sender's output $s_{j,k}$ is either equal to the receiver's programmed share $\hat{s}_{i,k}$ (if a valid prefix match occurred) or is a random output. Because the honest receiver samples all values $\hat{s}_{i,k} \leftarrow \FF_{2^{\ell_2}}$ uniformly at random, every element $s_{j,k}$ returned to the sender is uniformly distributed. $\Sim[\SSS]$ identically mimics this distribution by explicitly sampling uniform values from $\FF_{2^{\ell_2}}$.}

{We demonstrate the security of the simulation through a sequence of hybrid games:}

\begin{itemize}
    \item {$\mathsf{Hyb}_0$: This is the real protocol execution.}
    
    \item {$\mathsf{Hyb}_1$: This hybrid behaves identically to $\mathsf{Hyb}_0$, except that instead of invoking $\Func[wLPSI]$ (Step 6), \Sim[\SSS] receives $\{(s_j, \vecq_j)\}_{j \in [m]}$ from $\AAA$ and returns nothing, which perfectly simulates the $\Func[wLPSI]$'s behavior. Thus, the views in $\mathsf{Hyb}_0$ and $\mathsf{Hyb}_1$ are perfectly identical.} 
    
    \item {$\mathsf{Hyb}_2$: This hybrid behaves identically to $\mathsf{Hyb}_1$, except that instead of passing the honest receiver's programmed points to  $\Func[OPPRF]$ (Step 4), \Sim[\SSS] directly returns uniformly random values $\{s_{j, k}\}_{j \in [m], k \in [d]} \leftarrow \FF_{2^{\ell_2}}$ to \AAA. In $\mathsf{Hyb}_1$, as illustrated in the correctness analysis, $s_{j,k}$ is either equal to the receiver's choice $\hat{s}_{i,k}$ or a random function output. Because the honest receiver samples $\hat{s}_{i,k} \leftarrow \FF_{2^{\ell_2}}$ uniformly at random, every $s_{j,k}$ in $\mathsf{Hyb}_1$ is perfectly uniform and independent of \AAA's view. Thus, the views in $\mathsf{Hyb}_1$ and $\mathsf{Hyb}_2$ are perfectly identical.}

    \item {$\mathsf{Hyb}_3$: This hybrid behaves identically to $\mathsf{Hyb}_2$, except that instead of passing honest receiver's queries to $\Func[OPPRF]$ (Step 3), \Sim[\SSS] receives the programmed points $\{(k\|\tilde{q}_{j, k, t}, r_{j, k})\}_{j\in[m], k\in[d], t\in[u]}$ from $\AAA$, samples a random function $F$ such that $F(k\|\tilde{q}_{j, k, t}) = r_{j, k}$ for $j\in[m], k\in[d], t\in[u]$, and sends $\OOO^F$ to \AAA. Since $\mathsf{Hyb}_2$ also returns $\OOO^F$ according to these programmed points, the views in $\mathsf{Hyb}_2$ and $\mathsf{Hyb}_3$ are perfectly identical.}

\end{itemize}

{The final hybrid $\mathsf{Hyb}_3$ is exactly the simulated execution of \Sim[\SSS]. Therefore, the simulated view is perfectly indistinguishable from the real execution.}

 \textbf{Corrupted Receiver.} {We construct a simulator $\Sim[\RRR]$ for a semi-honest corrupted receiver. Given the receiver's input $W$ and the ideal output set $Z$ from $\Func[FPSI]$, $\Sim[\RRR]$ runs $\AAA$ and proceeds as follows:}
\begin{enumerate}
    \item {\textit{Simulating $\Func[OPPRF]$ (Step 3):} After receiving the \AAA's queries $\{k\|\tilde{w}_{i, k, t}\}_{i \in [n], k \in [d], t \in [u^\prime]}$, $\Sim[\RRR]$ samples uniformly random values $\{\hat{r}_{i, k, t}\}_{i \in [n], k \in [d], t \in [u^\prime]}$, and returns them to $\AAA$.}
    
    \item {\textit{Simulating $\Func[OPPRF]$ (Step 4):} After receiving the programmed points $\{(k\|\hat{r}_{i, k, t}, \hat{s}_{i, k})\}_{i \in [n], k \in [d], t \in [u^\prime]}$ from $\AAA$, $\Sim[\RRR]$ samples a random function $F^\prime$ such that $F^\prime(k\|\hat{r}_{i, k, t}) = \hat{s}_{i, k}$ for $i \in [n], k \in [d], t \in [u^\prime]$. $\Sim[\RRR]$ returns $\OOO^{F^\prime}$ to $\AAA$.}

    \item {\textit{Simulating $\Func[wLPSI]$ (Step 6):}  
    After receiving \AAA's input $\{\hat{s}_i\}_{i \in [n]}$, for every point $\vecq_j \in Z$, $\Sim[\RRR]$ identifies unique $i^* \in [n]$ such that $\mathsf{dist}_\infty(\vecw_{i^*}, \vecq_j) \le \delta$. $\Sim[\RRR]$ constructs the output set $I := \{(\hat{s}_{i^*}, \vecq_j) \mid \vecq_j \in Z\}$ and returns $I$ to $\AAA$.}
\end{enumerate}

% \blue{We argue the simulated view of $\AAA$ is identically distributed to the real execution. In Step 3, as illustrated in the correctness analysis, the real protocol dictates that the sender's output $\hat{r}_{i,k,t}$ is uniformly random. $\Sim[\RRR]$ perfectly captures this by explicitly providing uniform values.
% In Step 4, the receiver receives $\OOO^{F^\prime}$ based on the receiver's input, thereby perfectly matching the simulated step. In Step 6, the true $\Func[wLPSI]$ outputs the intersection using the keys generated in the protocol. Because the globally disjoint assumption guarantees that a receiver point $\vecw_{i^*}$ can match at most one sender point $\vecq_j$, the relation between the receiver's keys $\hat{s}_{i^*}$ and the payload $\vecq_j$ is perfectly unambiguous. $\Sim[\RRR]$ constructs the output set $I$ exactly as it would appear in the real protocol, successfully embedding the ideal functionality's output $Z$ into the simulated view. Thus, the proof is complete.}

{We demonstrate the security of the simulation through a sequence of hybrid games:}

\begin{itemize}
    \item {$\mathsf{Hyb}_0$: This is the real protocol execution.}
    
    \item {$\mathsf{Hyb}_1$: This hybrid behaves identically to $\mathsf{Hyb}_0$ except the simulation of \Func[wLPSI] (Step 6). After receiving \AAA's input $\{\hat{s}_i\}_{i \in [n]}$, for each $\vecq_j \in Z$, $\Sim[\RRR]$ identifies the unique $i^* \in [n]$ such that $\mathsf{dist}_\infty(\vecw_{i^*}, \vecq_j) \le \delta$. This uniqueness is guaranteed by the globally disjoint assumption. 
    $\Sim[\RRR]$ constructs the output set $I := \{(\hat{s}_{i^*}, \vecq_j) \mid \vecq_j \in Z\}$, and returns $I$ to $\AAA$. The only difference between $\mathsf{Hyb}_0$ and $\mathsf{Hyb}_1$ occurs if the protocol produces an incorrect evaluation. As established in the correctness analysis, the probability is statistically negligible.
    Thus, the views in $\mathsf{Hyb}_0$ and $\mathsf{Hyb}_1$ are statistically indistinguishable.}
    
    \item {$\mathsf{Hyb}_2$: This hybrid behaves identically to $\mathsf{Hyb}_1$ except that instead of invoking $\Func[OPPRF]$ (Step 4), after receiving the programmed points $\{(k\|\hat{r}_{i, k, t}, \hat{s}_{i, k})\}_{i \in [n], k \in [d], t \in [u^\prime]}$ from $\AAA$, $\Sim[\RRR]$ samples a random function $F^\prime$ such that $F^\prime(k\|\hat{r}_{i, k, t}) = \hat{s}_{i, k}$ for $i \in [n], k \in [d], t \in [u^\prime]$ and returns $\OOO^{F^\prime}$ to $\AAA$. Since $\mathsf{Hyb}_1$ also returns $\OOO^{F^\prime}$ according to these programmed points, the views in $\mathsf{Hyb}_1$ and $\mathsf{Hyb}_2$ are perfectly identical.}

    \item {$\mathsf{Hyb}_3$: This hybrid behaves identically to $\mathsf{Hyb}_2$, except that instead of passing the honest sender's points to  $\Func[OPPRF]$ (Step 3), $\Sim[\RRR]$ receives $\AAA$'s queries $\{k\|\tilde{w}_{i, k, t}\}_{i \in [n], k \in [d], t \in [u^\prime]}$ and returns uniformly random values $\{\hat{r}_{i, k, t}\}_{i \in [n], k \in [d], t \in [u^\prime]} \leftarrow \FF_{2^{\ell_1}}$. 
    As illustrated in the correctness analysis, in $\mathsf{Hyb}_2$, $\AAA$ receives either the sender's payload $r_{j,k}$ or a random function evaluation. Because the honest sender samples $r_{j,k} \leftarrow \FF_{2^{\ell_1}}$ uniformly at random, every returned evaluation is uniformly distributed. Thus, the views in $\mathsf{Hyb}_2$ and $\mathsf{Hyb}_3$ are perfectly identical.}
    
\end{itemize}

{The final hybrid $\mathsf{Hyb}_3$ is exactly the simulated execution of \Sim[\RRR]. Therefore, the simulated view is statistically indistinguishable from the real execution.}

\end{proof}

\subsection{Proof of Theorem \ref{thm: fuzzy PSI L-p-high-px}}

\label{Proof: of Theorem fuzzy PSI L-p-high-px}

\begin{proof}
\label{proof: fuzzy PSI L-p-high-px}
{The correctness of the protocol is directly derived from the correctness of the $L_\infty$ protocol for high-dimensional sets and the low-dimensional protocol, and hence we omit it here. Below, we mainly prove the security of the protocol by constructing simulators for both a corrupted sender and a corrupted receiver.}

 \textbf{Corrupted Sender.} 
{We construct a simulator $\Sim[\SSS]$ for a semi-honest corrupted sender. Given the sender's input $Q$, $\Sim[\SSS]$ runs the adversary $\AAA$ and proceeds as follows:}
\begin{enumerate}
    \item {\textit{Simulating fuzzy mapping (Step 1):} $\Sim[\SSS]$ invokes the simulator for the corrupted sender in the $\Pi_\mathsf{FPSI\text{-}high}^{L_\infty}$ protocol and sends the resulting transcript to $\AAA$.} 
    % $\Sim[\SSS]$ records the generated identifiers $\{\mathsf{id}_{\vecq_j}\}_{j \in [m]}$ from this simulated transcript.
    
    \item {\textit{Simulating $\Func[OPPRF]$ (Step 5):} After receiving the programmed points $\{L_i\}_{i \in [n^\prime]}$ from $\AAA$, $\Sim[\SSS]$ samples a random function $G$ programmed on $\{L_i\}_{i \in [n^\prime]}$. \Sim[\SSS] sends $\OOO^G$  to \AAA.}
    
    \item {\textit{Simulating $\Func[FMatch]$ and $\Func[OT]$ (Steps 6-7):} $\Sim[\SSS]$ receives its shares $\vecg_i$ to $\Func[FMatch]$ and $(\bot, \vecg_i)$ to $\Func[OT]$ from $\AAA$, and returns nothing.}
\end{enumerate}

% \blue{We argue the simulated view of $\AAA$ is identically distributed to the real execution. In Step 1, the indistinguishability follows directly from the proven security of the $L_\infty$ protocol. In Step 5, the adversary receives $\OOO^F$ according to the adversary's input, perfectly matching the simulated step.
% % the real $\Func[OPPRF]$ provides only a handle to the sender, which is perfectly simulated by $\OOO^G$. 
% In Steps 6 and 7, the adversary receives nothing in both the ideal and real worlds.}
% . Thus, the overall simulated view is computationally indistinguishable from the real protocol.

{We demonstrate the security of the simulation through a sequence of hybrid games:}

\begin{itemize}
    \item {$\mathsf{Hyb}_0$: This is the real protocol execution.}
    
    \item {$\mathsf{Hyb}_1$: This hybrid behaves identically to $\mathsf{Hyb}_0$ except the simulation of $\Func[FMatch]$ and $\Func[OT]$ (Steps 6-7). $\Sim[\SSS]$ receives $\AAA$'s inputs $\vecg_i$ to $\Func[FMatch]$ and $(\bot, \vecg_i)$ to $\Func[OT]$ and returns nothing, which perfectly simulates the functionalities' behavior. Thus, the views in $\mathsf{Hyb}_0$ and $\mathsf{Hyb}_1$ are perfectly identical.}

    \item {$\mathsf{Hyb}_2$: This hybrid behaves identically to $\mathsf{Hyb}_2$ except the simulation of \Func[OPPRF] (Step 5).
    After receiving the programmed points $\{L_i\}_{i \in [n^\prime]}$ from $\AAA$, $\Sim[\SSS]$ samples a random function $F$ programmed on $\{L_i\}_{i \in [n^\prime]}$. \Sim[\SSS] sends $\OOO^G$ to \AAA. Since $\mathsf{Hyb}_1$ also returns $\OOO^{G}$ according to these programmed points, the views in $\mathsf{Hyb}_1$ and $\mathsf{Hyb}_2$ are perfectly identical.}

    \item {$\mathsf{Hyb}_3$: This hybrid behaves identically to $\mathsf{Hyb}_2$, except that $\Sim[\RRR]$ invokes the receiver simulator for the $\Pi_\mathsf{FPSI\text{-}high}^{L_\infty}$ protocol and sends the transcript to $\AAA$. The perfect indistinguishability of the views in $\mathsf{Hyb}_2$ and $\mathsf{Hyb}_3$ follows from the security of the $L_\infty$ protocol.}

    % ($\mathsf{Hyb}_2 \equiv \mathsf{Hyb}_3$).
\end{itemize}

{The final hybrid $\mathsf{Hyb}_3$ is exactly the simulated execution of \Sim[\SSS]. Therefore, the simulated view is perfectly indistinguishable from the real execution.}

\textbf{Corrupted Receiver.} 
{We construct a simulator $\Sim[\RRR]$ for a semi-honest corrupted receiver. Given the receiver's input $W$ and the ideal output set $Z$ from $\Func[FPSI]$, $\Sim[\RRR]$ runs $\AAA$ and proceeds as follows:}
\begin{enumerate}
    \item {\textit{Simulating fuzzy mapping (Step 1):} $\Sim[\RRR]$ invokes the simulator for the corrupted receiver in the $\Pi_\mathsf{FPSI\text{-}high}^{L_\infty}$ protocol, purposefully omitting the final intersection delivery. $\Sim[\RRR]$ sends the transcript to $\AAA$}
    % and records the identifiers $\{\mathsf{id}_{\vecw_i}\}_{i \in [n]}$.
    
    \item {\textit{Simulating $\Func[OPPRF]$ (Step 5):} After receiving queries from $\AAA$, $\Sim[\RRR]$ samples uniformly random values $\hat{\vecg}_i \leftarrow \ZZ_{2^{\ell^\prime}}^d$ for all $i \in [n^\prime]$ and returns them to $\AAA$.}
    
    \item {\textit{Simulating $\Func[FMatch]$ and $\Func[OT]$ (Steps 6-7):} For each bin $i \in [n^\prime]$, $\Sim[\RRR]$ identifies if there are a mapped element $T_w[i] = (\mathsf{id}_{\vecw_i}\|\tau, \vecw_i)$ and a $\vecq \in Z$ such that $\mathsf{dist}_p(\vecq, \vecw_i) \leq \delta$. If such a match exists, $\Sim[\RRR]$ sets the simulated match bit $b_i := 1$ and constructs the OT payload as $\vecz_i := \vecq - \hat{\vecg}_i \bmod 2^{\ell^\prime}$. Otherwise (or if the bin is empty), $\Sim[\RRR]$ sets $b_i := 0$ and $\vecz_i := \bot$. After receving \AAA's inputs, $\Sim[\RRR]$ returns $\{b_i\}_{i \in [n^\prime]}$ and $\{\vecz_i\}_{i \in [n^\prime]}$ to $\AAA$ as the outputs of the functionalities.}
\end{enumerate}

% \blue{We argue the simulated view of $\AAA$ is identically distributed to the real execution. Step 1 indistinguishability follows from the security of the underlying $L_\infty$ protocol. In Step 5, the real protocol delivers $\hat{\vecg}_i = \vecq_j - \vecg_i \bmod 2^{\ell^\prime}$ if the identifiers match, or a random function output otherwise. Because the honest sender samples $\vecg_i \leftarrow \ZZ_{2^{\ell^\prime}}^d$ uniformly at random, the masked value $\vecq_j - \vecg_i$ is perfectly uniform and independent of $\vecq_j$. Thus, all outputs $\hat{\vecg}_i$ received by $\AAA$ in the real protocol are uniformly distributed over $\ZZ_{2^{\ell^\prime}}^d$, exactly matching the simulated distribution. 

% In Steps 6 and 7, $\Sim[\RRR]$ computes the exact ideal outputs for the target functionalities. By the globally disjoint assumption, $\vecw_i$ can match at most one $\vecq \in Z$, meaning the check is unambiguous. In the real protocol, if $b_i = 1$, the receiver obtains the sender's uniform mask $\vecg_i$ from $\Func[OT]$. In the simulation, $\Sim[\RRR]$ programs $\vecz_i := \vecq - \hat{\vecg}_i$. Since $\hat{\vecg}_i$ represents the simulated $\vecq - \vecg_i$, this mapping implicitly defines $\vecz_i = \vecg_i$, perfectly mimicking the relationship between the $\Func[OPPRF]$ outputs and the $\Func[OT]$ payloads while matching the ideal intersection $Z$.}

{We demonstrate the security of the simulation through a sequence of hybrid games:}

\begin{itemize}
    \item {$\mathsf{Hyb}_0$: This is the real protocol execution.}
    
    \item {$\mathsf{Hyb}_1$: This hybrid behaves identically to $\mathsf{Hyb}_0$ except the simulation of $\Func[FMatch]$ and $\Func[OT]$ (Steps 6-7).
    For each bin $i \in [n^\prime]$, if $T_w[i] = (\mathsf{id}_{\vecw_i}\|\tau, \vecw_i)$, $\Sim[\RRR]$ checks if there exists an ideal match $\vecq \in Z$ such that $\mathsf{dist}_p(\vecq, \vecw_i) \leq \delta$. By the globally disjoint assumption, $\vecw_i$ can match at most one $\vecq \in Z$, meaning this check is unambiguous.
    If a match exists, $\Sim[\RRR]$ sets the simulated match bit $b_i := 1$ and constructs the OT payload as $\vecz_i := \vecq - \hat{\vecg}_i \bmod 2^{\ell^\prime}$. Otherwise, $\Sim[\RRR]$ sets $b_i := 0$ and $\vecz_i := \bot$. $\Sim[\RRR]$ returns $\{b_i\}_{\in [n^\prime]}$ and $\{\vecz_i\}_{\in [n^\prime]}$ to $\AAA$. 
    In $\mathsf{Hyb}_0$, if $b_i = 1$, \AAA obtains the sender's mask $\vecg_i$ from $\Func[OT]$. In the simulation, setting $\vecz_i := \vecq - \hat{\vecg}_i$ implicitly defines $\vecz_i = \vecg_i$, perfectly mimicking the relationship between the $\Func[OPPRF]$ uniform masks and the $\Func[OT]$ payloads while embedding the true ideal intersection $Z$. The only difference between $\mathsf{Hyb}_0$ and $\mathsf{Hyb}_1$ occurs if the protocol produces an incorrect evaluation that is statistically negligible.
    Thus, the views in $\mathsf{Hyb}_0$ and $\mathsf{Hyb}_1$ are statistically indistinguishable.}
    
    \item {$\mathsf{Hyb}_2$: This hybrid behaves identically to $\mathsf{Hyb}_1$, except that instead of passing the honest sender's programmed points to  $\Func[OPPRF]$ (Step 5), $\Sim[\RRR]$ returns uniformly random values $\hat{\vecg}_i \leftarrow \ZZ_{2^{\ell^\prime}}^d$ for all $i \in [n^\prime]$.  
    In $\mathsf{Hyb}_1$, $\AAA$ receives $\hat{\vecg}_i = \vecq_j - \vecg_i \bmod 2^{\ell^\prime}$ if the identifiers match, or a random function output otherwise. Because the honest sender samples $\vecg_i \leftarrow \ZZ_{2^{\ell^\prime}}^d$ uniformly at random, the masked value $\vecq_j - \vecg_i$ is perfectly uniform and independent of $\vecq_j$. Thus, all outputs received by $\AAA$ in $\mathsf{Hyb}_1$ are uniformly distributed. Thus, the views in $\mathsf{Hyb}_1$ and $\mathsf{Hyb}_2$ are perfectly identical.}

    \item {$\mathsf{Hyb}_3$: This hybrid behaves identically to $\mathsf{Hyb}_2$, except that $\Sim[\RRR]$ invokes the receiver simulator for the $\Pi_\mathsf{FPSI\text{-}high}^{L_\infty}$ protocol, purposefully omitting the final intersection delivery. $\Sim[\RRR]$ sends the transcript to $\AAA$.
    % and records the identifiers $\{\mathsf{id}_{\vecw_i}\}_{i \in [n]}$.
    The perfect indistinguishability of the views in $\mathsf{Hyb}_2$ and $\mathsf{Hyb}_3$ follows from the security of the $L_\infty$ protocol.}
    
\end{itemize}

{The final hybrid $\mathsf{Hyb}_3$ is exactly the simulated execution of \Sim[\RRR]. Therefore, the simulated view is statistically indistinguishable from the real execution.}

\end{proof}

\end{document}